\documentclass[11pt]{article}
\usepackage{amsfonts}
\usepackage[usenames,dvipsnames,svgnames,table]{xcolor}
\usepackage{tikz}
\usetikzlibrary{shapes.misc,decorations.pathreplacing}

\usepackage{amssymb}
\usepackage{amsthm}
\usepackage{graphicx}
\usepackage{amsmath}
\usepackage[longnamesfirst]{natbib}
\usepackage{rotating}
\usepackage{afterpage}
\usepackage{cancel}
\usepackage{hyperref}
\usepackage{enumitem}
\usepackage[bottom]{footmisc}

\usepackage{caption}
\usepackage{subcaption}

\usepackage{setspace}
\hypersetup{
     colorlinks   = true,
     citecolor    = Blue,
     urlcolor     = Black,
     linkcolor    = Blue
}

\newcommand{\bc}{\begin{center}}
\newcommand{\ec}{\end{center}}
\newcommand{\be}{\begin{equation}}
\newcommand{\ee}{\end{equation}}
\newcommand{\bea}{\begin{eqnarray}}
\newcommand{\eea}{\end{eqnarray}}
\newcommand{\bean}{\begin{eqnarray*}}
\newcommand{\eean}{\end{eqnarray*}}
\newcommand{\bt}{\begin{tabular}}
\newcommand{\et}{\end{tabular}}

\usepackage{url}       %
\usepackage{xcolor}
\usepackage[margin=1in]{geometry}
\usepackage{graphicx}

\newtheorem{theorem}{Theorem}

\newtheorem{assumption}[theorem]{Assumption}

\newtheorem{definition}[theorem]{Definition}

\newtheorem{lemma}[theorem]{Lemma}

\newtheorem{proposition}[theorem]{Proposition}
\newtheorem{remark}[theorem]{Remark}

\numberwithin{theorem}{section}

\newcommand{\argmin}{\operatorname*{argmin}}

\newcommand{\PR}{\mathbb{P}}

\newcounter{saveeqn}

\begin{document}

\title{\vspace*{-0.5 in} \textbf{Robust Forecasting}}

\author{
		Timothy Christensen\\ {\em \small New York University} 
		\and
        Hyungsik Roger Moon\\ {\em \small Univ. of Southern California} \\
        {\em \small Schaeffer Center, and Yonsei University}
        \and
        Frank Schorfheide\thanks{Correspondence:
        		T. Christensen: Department of Economics, New York University, 19 West 4th Street, 6th floor, New York, NY 10012. Email: \texttt{timothy.christensen@nyu.edu}.
                H.R. Moon: Department of Economics, University of Southern California, KAP 300, Los Angeles, CA 90089. Email: \texttt{moonr@usc.edu}. 
                F. Schorfheide: Department of Economics, 133 S. 36th Street, 
                University of Pennsylvania, Philadelphia, PA 19104-6297. Email: 
                \texttt{schorf@ssc.upenn.edu}. 
                We are grateful for comments and suggestions from the 
                participants of the 2019 USC INET Panel Data Forecasting 
                Conference, the 2020 Econometric Society World Congress, the Penn 
                Econometrics Lunch seminar, and the Yale Econometrics seminar. We thank Zhan Gao, James Nesbit, and Boyuan Zhang for proofreading the manuscript.    
                This material is based upon work supported by the National Science Foundation under Grants No. SES-1919034 (Christensen), SES-1625586 (Moon), and SES-1851634 (Schorfheide).} \\
        {\em \small University of Pennsylvania} \\ {\em \small CEPR, NBER, and PIER} 
     }

\date{This Version: December 10, 2020}
\maketitle

\begin{abstract}
\noindent
	We use a decision-theoretic framework to study the problem of forecasting discrete outcomes when the forecaster is unable to discriminate among a set of plausible forecast distributions because of partial identification or concerns about model misspecification or structural breaks. We derive ``robust'' forecasts which minimize maximum risk or regret over the set of forecast distributions. We show that for a large class of models including semiparametric panel data models for dynamic discrete choice, the robust forecasts depend in a natural way on a small number of convex optimization problems which can be simplified using duality methods. Finally, we derive ``efficient robust'' forecasts to deal with the problem of first having to estimate the set of forecast distributions and  develop a suitable asymptotic efficiency theory. Forecasts obtained by replacing nuisance parameters that characterize the set of forecast distributions with efficient first-stage estimators can be strictly dominated by our efficient robust forecasts. 
\end{abstract}

\noindent JEL CLASSIFICATION: C11, C14, C23, C53

\noindent KEY\ WORDS: Statistical Decision Theory, Dynamic Discrete Choice, Forecasting, Identification, Minimax Loss, Minimax Regret, Panel Data Models, Robustness, Structural Breaks.

\thispagestyle{empty}
\setcounter{page}{0}
\newpage

\section{Introduction}

In this paper, we study the problem of forecasting discrete outcomes when the researcher is unable to discriminate among a set of plausible forecast distributions. There are several reasons why the forecaster might face uncertainty about the forecast distribution. A leading case is partial identification, in which the data up to the forecast origin only set-identify a subset of parameters of the forecasting model. Uncertainty about the forecast distribution can also arise when the forecaster expands the set of  models to accommodate concerns about model misspecification or structural breaks between the in-sample and forecast period.

Suppose that a subset of parameters of a forecasting model are only set-identified. Should the lack of point identification be a concern for the forecaster? At first glance the answer appears to be ``no'': if the parameters in the identified set generate different forecasts and some of these forecasts are less accurate than others, then we should be able to discriminate among the parameters based on the observed data. To the extent that we are unable to do so, the parameterizations should be observationally equivalent and therefore generate the same forecasts. This intuition is confirmed in the context of vector autoregressions (VARs): while the structural form of the VAR may only be set-identified, forecasts only utilize the reduced form of the VAR which is directly identifiable. This intuition is also confirmed in the context of dynamic linear factor models. The parameters are only identified up to a particular normalization of the latent factors, but each normalization leads to identical forecasts. However, the intuition is wrong in many other important settings.

Our paper makes several contributions. First,  we show that the VAR intuition does not apply when forecasting using dynamic discrete choice models for panel data. As is well known  \citep{HonoreTamer2006,Chamberlain2010,Chernozhukovetal2013}, the homogeneous parameters and the correlated random effects distribution are set-identified when no parametric assumptions are made about the random effects distribution. We demonstrate that in the panel dynamic discrete choice setting different parameters in the identified set lead to different forecasts, some more accurate than others. Unlike a VAR, the panel dynamic discrete choice model has a non-Markovian structure due to the sequential learning about  heterogeneous coefficients. As a consequence, parameterizations that are indistinguishable based on a panel of length $T$ may become distinguishable in a panel of length $T+1$.

Second, we construct forecasts that are ``robust'' to uncertainty about the parameterization $\theta$ of the forecasting model among a set of parameterizations $\Theta_0$ that are observationally equivalent at the forecast origin $T$. We refer to this as \emph{uncertainty about the forecast distribution} for short. Our robust forecasts minimize either maximum risk or maximum regret (i.e. risk relative to the infeasible Bayes decision under the true forecast distribution) over the set of  forecast distributions. We show that for binary (or classification) loss, quadratic loss, and logarithmic loss, the optimal binary forecast under either robustness criterion depends on two extremum problems which characterize the smallest and largest conditional probabilities for the  outcomes being forecast over the set of forecasting distributions. Similarly, robust forecasts in the multinomial case depend in a natural way on a small number of extremum problems. Further, we show that these extremum problems can be simplified by duality arguments for a broad class of models.

Our robust forecasts are not only applicable in settings in which parameters are set-identified, but also in environments in which the forecaster is concerned about model misspecification or there is a structural break at the forecast origin. The common feature is that the forecasts depend on an unknown parameter $\theta$ which takes values in a set $\Theta_0$. Under misspecification or structural breaks, the set $\Theta_0$ indexes an enlarged class of models representing plausible deviations from a benchmark model.

Third, we derive ``efficient robust'' forecasts to deal with the problem of 
first estimating the set $\Theta_0$ prior to making the forecast. To do so,  we 
express $\Theta_0$ as a (set-valued) function of an identifiable 
reduced-form parameter $P$. In order to develop an optimality theory, we 
evaluate forecasts by their integrated maximum risk or regret, averaging over 
both $P$ and the data. Under this criterion, the optimal forecast is what we 
call the \emph{Bayesian robust forecast}. It is obtained by minimizing the posterior 
maximum risk or regret which conditions on the data and averages out $P$ based 
on its posterior distribution.

Fourth, we develop an asymptotic efficiency theory for forecasting discrete outcomes under uncertainty about the parameterization of the forecast distribution when $\Theta_0$ is estimated. We show that in binary and multinomial discrete forecasting problems, forecasts that are asymptotically equivalent to the Bayesian robust forecasts minimize asymptotic integrated maximum risk or regret. We refer to such forecasts as \emph{asymptotically efficient-robust}. We demonstrate that forecasts obtained by replacing $P$ with an efficient first-stage estimator can be strictly dominated by the Bayesian robust forecast. This suboptimality of plug-in forecasts arises in settings in which key statistics that determine the robust forecast are only directionally, but not fully, differentiable with respect to the identifiable reduced-form parameters. Bagged predictors (see \cite{Breiman1996}) that replace posterior averaging with averaging across the bootstrap distribution of an efficient estimator of $P$, on the other hand, tend to be asymptotically efficient-robust.  

Our paper is related to several literatures. For forecasting short time-series using panel data see, e.g., \cite{Baltagi2008}, \cite{GuKoenker2014}, \cite{Liu2016}, and \cite{LiuMoonSchorfheide2018,LiuMoonSchorfheide2015}. Applications of partial identification in nonlinear panel data analysis include \cite{HonoreTamer2006}  and \cite{Chernozhukovetal2013}. Much of our paper is devoted to forecasting binary outcomes which has been previously considered by, for instance, \cite{ElliottLieli2013}, \cite{LahiriYang2013}, and \cite{ElliottTimmermann2016}.

There is an extensive literature on statistical decision theory following \cite{Wald1950}. Closely related to our approach are $\Gamma$-minimax (or $\Gamma$-minimax regret) decisions in robust Bayes analysis \citep{Robbins1951,Berger1985}. In economics, this approach is also related to the multiple priors framework of  \cite{GilboaSchmeidler1989} and the robustness literature following \cite{HansenSargent2001}. For econometric applications, \cite{Chamberlain2000,Chamberlain2001} derives minimax decision rules under point identification. 
\cite{Kitagawa2012},  \cite{GiacominiKitagawa2018}, and \cite{GiacominiKitagawaUhlig2019} study robust Bayesian analysis under set identification.

\cite{HiranoWright2017} study  the problem of forecasting continuous outcomes under uncertainty about predictor variables in a weak predictor local asymptotic setting. Despite several differences between their work and ours,\footnote{For instance, uncertainty about the parameterization of the forecasting model is resolved asymptotically in their framework whereas it persists in our setting.} they also find that bagging can reduce asymptotic risk.
Discrete forecasting has a similar structure to statistical treatment assignment and our efficiency results are related to efficiency results in that literature, most notably \cite{HiranoPorter2009}. In their setting, a welfare contrast is a smooth function of a point-identified, regularly estimable parameter. Their efficient rules are based on plugging-in an efficient estimator of the parameter. In our setting, uncertainty about the forecast distribution can introduce a type of non-smoothness to the robust forecasting problem. In consequence, our efficient robust forecasts differ from plug-in rules.

The remainder of this paper is organized as follows. Section~\ref{sec:setup} describes the setup, our objectives, and introduces motivating examples. Sections~\ref{sec:binary} and \ref{sec:multinomial} derive our robust and efficient robust forecasts for binary and multinomial forecasting settings, respectively. Section~\ref{sec:binary} also contains an application to panel models for dynamic binary choice. Section~\ref{sec:uncertainty} presents the main results on asymptotic efficiency. Appendix~\ref{appsec:computation} discusses computation for a broad class of models including semiparametric panel data models. Appendix~\ref{appsec:binary} contains additional results on robust binary forecasts and all proofs are relegated to Appendix~\ref{appsec:proofs}.

\section{Setup, Motivating Examples, and Objectives}
\label{sec:setup}

\subsection{Setup}

The econometrician wishes to forecast a random variable $Y$ taking values in a finite set $\mathcal Y$. The econometrician assumes $Y$ is distributed according to an (unknown) distribution in a family of forecast distributions $\big\{ \PR_\theta(Y = y ) : \theta \in \Theta_0 \big\}$, where $\theta \in \Theta$ denotes a vector of parameters and $\Theta_0 \subseteq \Theta$ indexes the set of forecast distributions over which the forecaster seeks robustness. The forecast distributions $\PR_\theta(Y = y ) $ may be conditioned on covariates observed by the econometrician when making the forecast, but we suppress this dependence in what follows to simplify notation.

As we discussed in the Introduction, there are several reasons why the forecaster might be uncertain about the forecast distribution. A leading case is partial identification, in which $\Theta_0$ represents the identified set of parameters that are observationally equivalent up to the forecast origin. 
Uncertainty about the forecast distribution can also arise under concerns about model misspecification or structural breaks. In these settings, $\Theta_0$ indexes an enlarged class of models representing plausible deviations from a benchmark model. 

\subsection{Motivating Examples}

To fix ideas and illustrate the broad applicability of our results, we now present a number of examples of where this forecasting problem arises. The first four examples use a panel data model for dynamic discrete choice to show how our approach accommodates concerns about model misspecification and structural breaks in a unified manner, though these are relevant concerns for {\em any} forecasting model. The last two examples involve counterfactuals and treatment assignments.

\paragraph*{Example 1: Semiparametric random effects model for dynamic binary choice.}\label{example:binary}
Let
\be
   Y_{it+1} = \mathbb{I} [ \lambda_i + \beta Y_{it} \ge U_{it+1} ], \quad
   \PR( U_{it+1} \le u | Y_i^t = y^t, \lambda_i = \lambda) = \Phi_{t+1}(u),
   \label{eq:benchmark.pddc}
\ee
where $\mathbb{I}[ y \ge a] = 1$ if $y \ge a$ and $0$ otherwise, $Y_i^t = (Y_{i1},...,Y_{it})'$, and $y^t \in \{0,1\}^t$.
The econometrician observes $Y_i^T = (Y_{i1},...,Y_{iT})'$ for $i=1,\ldots,n$ where $T$  is fixed. 
To avoid the initial conditions problem, the econometrician treats $Y_{i0}$ as unobserved and specifies a joint distribution $\Pi_{\lambda,y}$ over $Y_{i0}$ and $\lambda_i$. 
As is well known, $\beta$ is not point-identified when $T$ is small and no parametric restrictions are placed on $\Pi_{\lambda,y}$ (see, e.g., \cite{Cox1958}, \cite{Chamberlain1985}, and \cite{Magnac2000}). Moreover, $\Pi_{\lambda,y}$ and the $\Phi_t$ are not nonparametrically point-identified for any $T$.

The econometrician wishes to forecast individual-level outcomes $Y_{iT+1}$ conditional upon an individual's history $Y_i^T = y^T \in \{0,1\}^T$. Suppose that the econometrician assumes each of the $\Phi_t$ takes a parametric form $\Phi$, such as logistic or standard normal. The identified set $\Theta_0$ then consists of all $(\Pi_{\lambda,y},\beta)$ for which the model-implied probabilities of observing sequences $Y_i^T = y^T \in \{0,1\}^T$ are equal to the probabilities observed in the data up to date $T$:
\begin{equation}\label{eq:idset:panel} 
 \Theta_0 = \big\{ \theta = (\beta,\Pi_{\lambda,y}) \in \Theta : p(y^T |  \beta, \Pi_{\lambda,y}) = p(y^T)  \;\; \forall \, y^T \in \{0,1\}^T  \big\}\,,
\end{equation}
where $ p(y^T |  \beta, \Pi_{\lambda,y})$ denotes the model-implied probabilities and $p(y^T)$ denotes the true (population) probabilities of observing $Y_i^T = y^T$. 
In the above notation, $Y = Y_{iT+1}$ and the forecast probability $\PR_\theta$ denotes the conditional probability over $Y_{iT+1}$ given $Y_i^T = y^T$:
\begin{equation}\label{eq:prob:panel} 
 \PR_\theta(Y = 1) := \PR_\theta ( Y_{iT+1} = 1 |Y_i^T =y^T ) = \frac{ \int  \Phi(\beta y_{iT} + \lambda) p(y^T | y_0,\lambda;  \beta ) \mathrm{d}\Pi_{\lambda,y}(\lambda,y_0) } 
	{ \int   p(y^T | y_0,\lambda;  \beta ) \mathrm{d}\Pi_{\lambda,y}(\lambda,y_0) } \,. \quad \Box
\end{equation}

\paragraph*{Example 2: Misspecification.} Consider the setup described in Example 1, but suppose that the econometrician adopted a parametric correlated random effect model, $\Pi_{\lambda,y} = \Pi(\lambda,y_0; \xi)$ for $\xi \in \Xi$, a set of auxiliary parameters. The econometrician is worried that this parametric random effects specification is misspecified, and so allows for the possibility that $\Pi_{\lambda,y} \in N(\xi)$, a neighborhood of $\Pi(\lambda,y; \xi)$. Suppose the econometrician again sets $\Phi_t = \Phi$ for all $t$.
The  set $\Theta_0$ is 
\[ 
 \Theta_0 = \{ \theta = (\beta,\xi,\Pi_{\lambda,y}) \in \Theta : p(y^T |  \beta, \Pi_{\lambda,y}) = p(y^T)  \;\; \forall \, y^T \in \{0,1\}^T \;\; \mbox{and} \;\; \Pi_{\lambda,y} \in N(\xi) \}\,.
\]
This setup was considered by \cite{BonhommeWeidner} under \emph{local} misspecification,  where $N(\xi)$ are Kullback--Leibler neighborhoods $N(\xi) = \big\{ \Pi : K(\Pi \,\|\, \Pi(\,\cdot\,;\xi)) \leq \delta \big\}$ for each $\xi \in \Xi$ with  $\delta \downarrow 0$ as $n \to \infty$, so that worst-case misspecification bias and sampling uncertainty are of the same order asymptotically.\footnote{Note that the emphasis of \cite{BonhommeWeidner} is on estimating posterior average effects whereas we focus on forecasting discrete (e.g. individual-level) outcomes.} We instead treat $\delta > 0$ as fixed, allowing \emph{global} misspecification. $\Box$

\paragraph*{Example 3: Structural breaks.} Three types of breaks can, in principle, occur at the forecast origin $T$ in Example 1: a break in the distribution of the $U_{it}$, a break in the individual effects $\lambda_i$, and a break in $\beta$.
Suppose the econometrician again takes $\Phi_t = \Phi$ for dates $t = 1,\ldots,T$, but allows for the possibility that $\Phi_{T+1} \neq \Phi$. For instance, the econometrician might want to allow for $\Phi_{T+1} \in N$, a neighborhood of $\Phi$. 
Even if $\beta$ and $\Pi_{\lambda,y}$ were known at date $T$, there would still be a set of forecast distributions for $Y_{iT+1}$ corresponding to different $\Phi_{T+1} \in N$. 
Using the above notation, we can redefine $\Theta_0$ as
\[
  \Theta_0 = \big\{ \theta = (\beta,\Pi_{\lambda,y},\Phi_{T+1}) \in \Theta : p(y^T |  \beta, \Pi_{\lambda,y}) = p(y^T)  \;\; \forall \, y^T \in \{0,1\}^T \;\; \mbox{and} \;\; \Phi_{T+1} \in N \big\}\,, 
\]
and replace $\Phi$ in the definition of $\mathbb{P}_\theta(Y=1)$ in (\ref{eq:prob:panel}) with $\Phi_{T+1}$.
Breaks in $\lambda_i$ can be viewed as a location shift of the distribution $\Phi_t$ and are subsumed under breaks in the distribution of $U_{it}$ for suitable choice of $N$. 
Breaks in $\beta$ can be handled by defining 
\[
  \Theta_0 = \big\{ \theta = (\beta,\beta_{T+1},\Pi_{\lambda,y}) \in \Theta : p(y^T |  \beta, \Pi_{\lambda,y}) = p(y^T)  \;\; \forall \, y^T \in \{0,1\}^T \;\; \mbox{and} \;\; | \beta - \beta_{T+1}| \leq \delta \big\}\,, 
\]
and by replacing $\Phi(\beta y_{iT} + \lambda) $ in (\ref{eq:prob:panel}) with $\Phi(\beta_{T+1} y_{iT} + \lambda) $. $\Box$

\paragraph*{Example 4: Semiparametric random effects model for dynamic multinomial choice.} Let
\[
 Y_{it+1} = \arg \max_m \left( U_{it+1}^m + \varepsilon_{it+1}^m \right) \,, \;\;
 U_{it+1}^0 = 0 \,, \;\; 
 U_{it+1}^m = u_{mt+1}(X_{it},Y_i^{t};\phi,\lambda_i,Y_{i0}) \,, \;\; m = 1,\ldots,M \,,
\]
where $\varepsilon_{it} = (\varepsilon_{it}^0,\ldots,\varepsilon_{it}^M)'$ is a vector of utility shocks with $\varepsilon_{it}|X_{it},\lambda_i \sim \Phi_t$ for each $t$, $\Phi_t$ is a potentially time-varying distribution, $\phi$ is a vector of homogeneous parameters, $\lambda_i$ is a vector of heterogeneous parameters, and $X_{it}$ is a vector of exogenous regressors. 
The econometrician observes $Y_i^T = (Y_{i1},\ldots,Y_{iT})'$ and $X_i^T = (X_{i1},\ldots,X_{iT})'$ for $i = 1,\ldots,n$ where $T$ is fixed and $n \to \infty$. To avoid the initial conditions problem, the econometrician specifies a joint distribution $\Pi_{\lambda,y}$ for $(\lambda_i,Y_{i0})$. As with Example 1, identification of model parameters with fixed $T$ and $n \to \infty$ is delicate, especially when parametric assumptions about the $\Phi_t$ and/or $\Pi_{\lambda,y}$ are relaxed; see \cite{HonoreKyriazidou2000}, \cite{Chernozhukovetal2013}, \cite{KhanOuyangTamer2019}, and references therein. Identified sets and forecast probabilities for individual-level outcomes are constructed in a similar manner to Example 1. $\Box$

\paragraph*{Example 5: Counterfactuals in structural models.} Counterfactuals in structural models are also subsumed in our framework when the outcome of interest is discrete, as is often the case for static or dynamic models of discrete choice or discrete games (e.g. firm entry/exit). In the above notation, $\theta$ are the structural parameters estimated under one policy regime, the econometrician wishes to predict a variable $Y$, and the model implies that $Y$ is distributed according to $\PR_\theta$ under the intervention. Partial identification can arise on two fronts. First, the model may itself be specified flexibly, leading to a non-singleton identified set $\Theta_0$ of structural parameters. Second, the potential for multiple equilibria and lack of knowledge about an equilibrium selection mechanism under the intervention may lead to a nontrivial set of forecast distributions. This can be subsumed by treating the selection mechanism itself as part of $\theta$, with $\Theta_0$ indexing distributions in a manner that is robust to the type of selection mechanism (see, e.g.,  \cite{Jia2008}, \cite{CilibertoTamer}, and \cite{Grieco2014}). $\Box$

\paragraph*{Example 6: Treatment assignment.} The problem of making discrete forecasts has a very similar structure to a treatment assignment problem, e.g., determining whether an individual should be vaccinated. Suppose the econometrician has access to a sample of observational data of size $n$ and observes for an individual $i$ the triplet $X_i = \big(D_i, W_{0i}(1-D_i), W_{1i} D_i \big)$, where $D_i$ is a treatment indicator, and $W_0$ and $W_1$ are the potential outcomes (``welfare'') of the untreated and the treated individuals. As the sample size tends to infinity, the econometrician is able to estimate the reduced form expectations $P= \big( \mathbb{E}[D], \mathbb{E}[W_0(1-D)], \mathbb{E}[W_1D] \big)$. 

Let $\mathbb{P}_\theta$ denote the joint distribution of $(D,W_{0},W_{1})$. To keep the example simple, we assume that the potential outcomes are binary and take values $\{a_{0}, a_{1}\}$ and $\{b_{0}, b_{1}\}$. Thus, the distribution $F_\theta$ is discrete with support on 
$\{0,1\} \times \{a_{0}, a_{1}\} \times \{b_{0}, b_{1}\}$. The support points of the potential outcome distribution can be easily point-identified based on two untreated (and two treated) individuals with different outcomes. Thus, we exclude the support points from the definition of $\theta$ and $P$. Using the notation that $\theta_{ijk} = \mathbb{P}\big( D=i, W_0=a_{j}, W_1=b_{k} \big)$, the identified set $\Theta_0(P)$ is defined by the following set of linear restrictions: 
\begin{eqnarray*}
	\quad 0 \le \theta_{ijk} \le 1, && \sum_{i=0,1} \sum_{j=0,1} \sum_{k=0,1} \theta_{ijk} = 1, \quad \mathbb{E}[D]  = \sum_{j=0,1} \sum_{k=0,1} \theta_{1jk}, \\
	\mathbb{E}[W_0(1-D)] &=& \sum_{k=0,1} a_0 \theta_{00k} + a_1\theta_{01k}, \quad
	\mathbb{E}[W_1 D] = \sum_{j=0,1} b_0 \theta_{1j0} + b_1\theta_{1j1},
\end{eqnarray*}
where $\theta$ stacks the $\theta_{ijk}$ probabilities.

Define the indicator variable $Y=\mathbb{I}[W_1 \ge W_0]$ which measures whether the treatment effect is (weakly) positive or not. From a policy maker's perspective the treatment effect for an individual not included in the initial trial is uncertain and the treatment decision $d \in \{0,1\}$ can be viewed as a forecast of $Y \in \{ 0,1\}$ with the understanding that the individual should be treated if the point forecast of $Y$ is one and not treated otherwise. 

\cite{Dehejia2005} analyzed this problem in a decision-theoretic framework under point identification with binary treatments. However, the example highlights the well-known result that the distribution of welfare rankings can be partially identified (see, e.g., \cite{Manski1996,Manski2000} and \cite{HeckmanSmithClements1997}).
\cite{Manski2000,Manski2002,Manski2004,Manski2007treatment} used a decision-theoretic framework to analyze optimal treatment in a planning problem under partial identification and advocated minimax and minimax regret approaches.\footnote{\citeauthor{Manski2002} focuses on population welfare objective whereas here we focus on individual-level outcomes. See also \cite{ManskiTetenov2007}, \cite{HiranoPorter2009}, \cite{Tetenov2012}, and \cite{KitagawaTetenov2018} amongst others, for the analysis of treatment rules under social welfare objectives.} $\Box$

\subsection{Objectives}

We derive two types of forecasts that deal with uncertainty about the forecast distribution. The first are \emph{robust forecasts} which seek to robustify the forecast with respect to $Y$ being distributed according to any distribution in the class $\{\PR_\theta : \theta \in \Theta_0\}$. We use minimax and minimax regret criteria as our notion of robustness. The second are \emph{efficient robust forecasts} which deal with the additional problem of having to first estimate $\Theta_0$ from data. The exposition in the remainder of this section and in Section~\ref{sec:binary} focuses on binary outcomes. We will consider extensions to multinomial outcomes in Section~\ref{sec:multinomial}.

\noindent {\bf Known $\Theta_0$.} Given a decision space $\mathcal D \subseteq [0,1]$, a loss function $\ell : \{0,1\} \times \mathcal{D} \to \mathbb{R}_+$ and $\theta \in \Theta_0$, the \emph{risk} of $ d \in \mathcal D$ under the forecast distribution $\PR_\theta$ is 
\[
 \mathbb E_\theta[ \ell(Y,d) ] = \ell(0,d) \, {\textstyle \PR_\theta} ( Y = 0 ) + \ell(1, d) \, {\textstyle \PR_\theta}( Y = 1 ) \,.
\]
The expectation $\mathbb E_\theta$ and forecast probabilities $\PR_\theta$ may condition on covariates observed by the econometrician when making the forecast; we have suppressed this dependence to simplify notation. The \emph{$\theta$-optimal forecast}, denoted  $d_\theta^*$, minimizes risk under $\PR_\theta$:
\[
 \mathbb E_\theta[ \ell(Y, d_\theta^* ) ] = \inf_{d \in \mathcal D} \mathbb E_\theta[ \ell(Y, d ) ] \,.
\]
A \emph{minimax} forecast solves
\be
 \inf_{d \in \mathcal{D}} \; \sup_{\theta \in \Theta_0} \; \mathbb{E}_{\theta} [ \ell(Y, d) ] \,.
 \label{eq:minimax}
\ee
The \emph{regret} of a forecast is its risk in excess of the risk of the $\theta$-optimal forecast. A \emph{minimax regret} forecast  solves
\be
 \inf_{d \in \mathcal{D}} \; \sup_{\theta \in \Theta_0} \Big( \mathbb{E}_{\theta} [ \ell(Y, d) ] - \mathbb{E}_\theta [ \ell(Y,d_\theta^*) ] \Big) \,.
 \label{eq:minimax.regret}
\ee
Robust forecasts are derived under these criteria in Sections \ref{subsec:binary.robust} and \ref{subsec:multinomial.robust} for binary and multinomial outcomes, respectively.

\noindent {\bf Estimated $\Theta_0$.} In many scenarios the researcher might not know $\Theta_0$ and will therefore need to first estimate the set (or features of $\Theta_0$ germane to the forecasting problem) from a sample of data of size $n$.\footnote{The known-$\Theta_0$ case can be viewed as the limit as $n \to \infty$.} Our \emph{efficient robust forecasts} deal with the additional uncertainty that arises from not knowing $\Theta_0$. Here ``efficient robust'' forecasts are those for which the maximum risk or regret is as close as possible to the maximum risk or regret of an oracle forecast with known $\Theta_0$ (see Section \ref{sec:uncertainty}). This efficiency notion recognizes that uncertainty about the true identity of the forecast distribution among the class $\{\PR_\theta : \theta \in \Theta_0\}$ is the dominant consideration asymptotically, but that estimation error may nevertheless have a material impact on the forecast in any finite sample. 

To make the analysis of efficient robust forecasts tractable but reasonably general, we assume that the model for the data, say $X_n$, and outcome $Y$ is indexed by $\theta$ and a $k$-dimensional vector of reduced-form parameters $P \in \mathcal P \subseteq \mathbb R^k$. The parameters $\theta$ and $P$ are linked by a known mapping $P \mapsto \Theta_0(P)$. For partially identified forecasting models,  $\Theta_0(P)$ denotes the identified set if $P$ was the true reduced form parameter value.  We assume that $X_n$ and $Y$ are related to $\theta$ and $P$ in the following manner:
\begin{align}
 {\textstyle \PR_\theta}(Y = y| X_n, P) & = {\textstyle \PR_\theta}(Y = y) \,, \label{eq:pr:forecast}\\
 X_n | \theta, P & \sim F_{n,P} \,. \label{eq:pr:data}
\end{align}
Condition (\ref{eq:pr:forecast}) implies that $Y$ does not depend on the data $X_n$ or $P$ beyond dependence through $\theta$. Condition (\ref{eq:pr:data}) says that the distribution of the data is fully summarized by $P$, which is standard for estimation and inference under partial identification; see, e.g., \cite{MoonSchorfheide2012}. Examples 1-6 can be shown to fit this framework. Here we just discuss Example 1 for brevity.

\paragraph*{Example 1 (continued).} In this example, the econometrician observes data $X_n = (Y_i^T)_{i=1}^n$. The data are used to estimate the vector $P = (p(y^T))_{y^T \in \{0,1\}^T}$, which collects the probabilities of observing sequences $y^T \in \{0,1\}^T$. The mapping $\Theta_0(P)$ from $P$ to $\theta = (\beta,\Pi_{\lambda,y})$ is defined in (\ref{eq:idset:panel}). Given a history $y^T$, the distribution of $Y \equiv Y_{iT+1}$ is fully summarized by $\theta$; see (\ref{eq:prob:panel}). Moreover, the distribution of the data is itself multinomial over the different realizations of $y^T$ with probabilities $P$. Thus $F_{n,P}$ is the product of $n$ multinomial distributions with probabilities $P$. In this model $y_0$ and $\lambda$ are unit specific and the forecasts are constructed based on the posterior distribution of $(y_0,\lambda)$ conditional on $\theta = (\beta,\Pi_{\lambda,y})$, see (\ref{eq:prob:panel}). $\Box$

\section{Binary Forecasts}
\label{sec:binary}

In this section we consider the binary forecasting problem. First, in Section \ref{subsec:binary.point} we review several common binary loss functions and their corresponding $\theta$-optimal forecasts. In Section~\ref{subsec:binary.robust} we derive forecasts that are robust to uncertainty about the forecast distribution and in Section \ref{subsec:binary.efficient} we construct efficient robust forecasts for the case of an estimated $\Theta_0$. The forecasts are summarized in Table \ref{tab:summary.binary} below. Section \ref{subsec:example} presents an application to semiparametric panel data models for dynamic binary choice. 

\subsection{$\theta$-Optimal Forecasts}
\label{subsec:binary.point}

We consider three loss functions  to evaluate forecast accuracy: binary (or classification) loss,  quadratic loss, and log predictive probability score.

\paragraph*{Binary (or Classification) Loss.} The binary loss function for $\mathcal D = \{0,1\}$ is
\be
    \ell_b \big( y,d \big) = a_{10} \mathbb{I}[ y=1, \, d=0 ] + a_{01} \mathbb{I}[ y=0, \, d=1 ] \,,
    \label{eq:loss.binary}
\ee
where $a_{10}, a_{01} \geq 0$. A special case with $a_{10} = a_{01}$ is classification loss $\ell_b(y,d) = \mathbb{I}[ y \neq d ]$.\footnote{When $a_{10} = a_{01}$, it is without loss of generality to normalize their common value to $1$.} The $\theta$-optimal forecast is
\be
d_{b, \theta}^* = \mathbb{I} \left[ {\textstyle \PR_\theta}(Y = 1) \ge {\textstyle \frac{a_{01}}{a_{01}+a_{10}}} \right] \label{eq:forecast.bayes}
\ee
and its risk is
\be
 a_{10} {\textstyle \PR_\theta}(Y = 1) \, \wedge \, a_{01} {\textstyle \PR_\theta}(Y = 0) \,, \label{eq:risk.bayes}
\ee
where $a \wedge b = \min\{a,b\}$. The $\theta$-optimal forecast is not unique when $\PR_\theta(Y = 1) = \frac{a_{01}}{a_{01}+a_{10}}$. In this case, however, all $\theta$-optimal forecasts differ only in their handling of ties and have the same risk.

\paragraph*{Quadratic Loss.} The quadratic loss for $d \in \mathcal D = [0,1]$ is 
\be
    \ell_q \big( y,d \big) = \big( y - d \big)^2 \,.
    \label{eq:loss.quadratic}
\ee
 With ${\cal D} = [0,1]$ the $\theta$-optimal forecast is the mean of $Y$ under the forecast distribution:
\be
    d_{q,\theta}^* = \mathbb{E}_\theta[ Y ] ={\textstyle \PR_\theta}(Y = 1) \,. \label{eq:forecast.bayes.quadratic}
\ee

\paragraph*{\bf Log Loss.}  Here the loss function for $\mathcal D = [0,1]$ is
\begin{equation}
   \ell_p \big( y,d \big) 
    = - \mathbb{I}[ y=1 ] \log d - \mathbb{I}[ y=0 ] \log (1-d) \,.
    \label{eq:loss.log}
\end{equation}
The $\theta$-optimal forecast is also the mean:
\be
d_{p,\theta}^* = {\textstyle \PR_\theta}(Y = 1) \,. \label{eq:forecast.bayes.log}
\ee
Although the $\theta$-optimal forecasts under quadratic loss and log loss are the same, their risks are different: the risk under quadratic loss is the variance of the forecast distribution,  whereas the risk under log loss is the entropy of the forecast distribution.

\subsection{Robust Forecasts}
\label{subsec:binary.robust}

We now relax the assumption that the forecast distribution $\PR_\theta$ is known and derive forecasts that are robust with respect to $\PR_\theta$ being any member of the set of forecast distributions $\{\PR_\theta : \theta \in \Theta_0\}$. Note, however, that in this section we treat $\Theta_0$ as known. 

The minimax and minimax regret forecasts will depend on the lower and upper values of the forecast probabilities as $\theta$ varies over $\Theta_0$:
\begin{align}
 p_L & := \inf_{\theta \in \Theta_0} \; {\textstyle \PR_\theta}(Y = 1) \,, \label{eq:p.lower}\mbox{ and}  \\
 p_U & := \sup_{\theta \in \Theta_0} \; {\textstyle \PR_\theta}(Y = 1) \,. \label{eq:p.upper} 
\end{align}
Although our characterizations are general, the challenge in implementing robust forecasts is to solve these extremum problems which will typically require exploiting some additional structure. 
Appendix \ref{appsec:computation} shows how duality methods may be used to simplify computation in a broad class of models that includes, but is not limited to, semiparametric dynamic binary choice models.

\subsubsection{Minimax Forecasts}

\paragraph*{Binary (or Classification) Loss.} We first derive the forecast that solves (\ref{eq:minimax}) for the binary loss function $\ell_b$ from (\ref{eq:loss.binary}) and decision space $\mathcal D = \{0,1\}$. The maximum risk of $d\in \{0,1\}$ is
\begin{equation}\label{eq:r.d.minimax}
	\sup_{\theta \in \Theta_0} \; \mathbb{E}_\theta [ \ell_b(Y, d)  ] 
	= \left[ \begin{array}{ll}
	a_{01} - a_{01} p_L & \mbox{if $d = 1$}\,, \\
	a_{10} \, p_U  & \mbox{if $d = 0$}\,.
	\end{array} \right.
\end{equation}
The minimax forecast for binary (or classification) loss is therefore
\be
	d_{b,mm} = \mathbb{I} \left[ a_{01} \le a_{01} p_L  +  a_{10} p_U  \right] \label{eq:d.minimax}
\ee
and the minimax risk is 
\[
	{\cal R}^*_{b,mm} =\left(   a_{01}- a_{01}p_L \right)  \, \wedge \, \left( a_{10} p_U \right) \,. 
\]
The minimax binary forecast is not unique when $a_{01} = a_{01} p_L  +  a_{10} p_U$. In this case, each minimax forecast differs only in its handling of ties and has the same maximum risk.

\paragraph*{Quadratic Loss.} We now derive the forecast that solves (\ref{eq:minimax}) for the quadratic loss function $\ell_q$ from (\ref{eq:loss.quadratic}) and decision space $\mathcal D = [0,1]$. The maximum risk of $d \in [0,1]$ is
\begin{equation} \label{eq:r.q.minimax}
  \sup_{\theta \in \Theta_0} \; \mathbb{E}_\theta [ \ell_q(Y, d)  ] 
  = \left[ \begin{array}{ll}
  p_U(1-2d) + d^2 & \mbox{if $d < \frac{1}{2}$}\,, \\
  p_L(1-2d) + d^2 & \mbox{if $d > \frac{1}{2}$}\,, \\
  \frac{1}{4} & \mbox{if $d = \frac{1}{2}$} \,.
  \end{array} \right.
\end{equation}
The minimax forecast is therefore
\be
 d_{q,mm} = \left[ \begin{array}{ll}
  p_U & \mbox{if $p_U \leq \frac{1}{2}$}\,, \\
  p_L & \mbox{if $p_L \geq \frac{1}{2}$}\,, \\
  \frac{1}{2} & \mbox{otherwise}\,, 
  \end{array} \right. \label{eq:minimax.quadratic}
\ee
and the minimax risk is 
\[
{\cal R}^*_{q,mm} = \left[ \begin{array}{ll}
  p_U(1-p_U)  & \mbox{if $p_U \leq \frac{1}{2}$}\,, \\
  p_L(1-p_L)  & \mbox{if $p_L \geq \frac{1}{2}$}\,, \\
  \frac{1}{4} & \mbox{otherwise} \,.
  \end{array} \right. 
\]

\paragraph*{Log Loss.} The minimax forecast $d_{q,mm}$ is also minimax for the log loss function $\ell_p$ from (\ref{eq:loss.log}) and decision space $\mathcal D = [0,1]$; see Appendix \ref{appsec:binary}.

\subsubsection{Minimax Regret Forecasts}

\paragraph*{Binary (or Classification) Loss.}
We first derive the forecast that solves (\ref{eq:minimax.regret}) 
for the binary loss function from (\ref{eq:loss.binary}) and decision space  ${\cal D} = \{0,1\}$. In view of (\ref{eq:risk.bayes}), the inner maximization problem in (\ref{eq:minimax.regret}) becomes
\begin{equation} \label{eq:regret.binary.pointwise}
  \sup_{\theta \in \Theta_0} \bigg( \mathbb{E}_\theta [ \ell_b(Y,d)  ] - a_{10} {\textstyle \PR_\theta}(Y = 1) \, \wedge \, a_{01} {\textstyle \PR_\theta}(Y = 0) \bigg) 
  = \left[ \begin{array}{ll}
 \left(a_{01} - (a_{01} + a_{10}) \, p_L \right)_+ & \mbox{if $d = 1$,} \\
 \left( (a_{01} + a_{10}) \, p_U - a_{01} \right)_+ & \mbox{if $d = 0$,}
 \end{array} \right.
\end{equation}
where $a_+ = \max\{a,0\}$. 
Therefore, the minimax regret forecast is
\be
 d_{b,mmr}
 = \mathbb{I} \left[  \left({\textstyle \frac{a_{01}}{a_{01} + a_{10}}} - p_L \right)_+ \leq \left( p_U- {\textstyle \frac{a_{01}}{a_{01} + a_{10}}} \right)_+ \right] \label{eq:forecast.mmr}  
\ee
and its maximum regret is
\be
 {\cal R}^*_{b,mmr} = \left(a_{01} - (a_{01} + a_{10})  p_L \right)_+  \wedge \left( (a_{01} + a_{10})  p_U - a_{01} \right)_+\,. \nonumber
\ee
The minimax and minimax regret binary forecasts for classification loss (i.e., $a_{01} = a_{10}$) are the same; see Appendix \ref{appsec:binary}.
As with other forecasts for $\mathcal D = \{0,1\}$, the minimax regret forecast is not necessarily unique. Non-uniqueness arises whenever $( \frac{a_{01}}{a_{01} + a_{10}} - p_L )_+ = ( p_U - { \frac{a_{01}}{a_{01} + a_{10}}} )_+$. If so, each minimax regret forecast has the same maximum regret and differs only in its handling of ties.

\paragraph*{Quadratic Loss.} We now derive the forecast that solves (\ref{eq:minimax.regret}) for the quadratic loss function $\ell_q$ from (\ref{eq:loss.quadratic}) and decision space  ${\cal D} = [0,1]$. By convexity,  the maximum regret of $d \in [0,1]$ is
\begin{equation} \label{eq:r.q.minimax.regret}
  \sup_{\theta \in \Theta_0} \; \mathbb{E}_{\theta} [ \ell_q(Y,d) ] 
  = \left[ \begin{array}{ll}
  (p_U - d)^2 & \mbox{if $d \leq \frac{p_L + p_U}{2}$}\,, \\
  (p_L - d)^2 & \mbox{if $d \geq \frac{p_L + p_U}{2}$}\,.
  \end{array} \right.
\end{equation}
The minimax regret forecast for quadratic loss is therefore the midpoint of the extreme forecast probabilities:
\[
 d_{q,mmr} = \frac{p_L + p_U}{2} 
\]
and the minimax regret is
\[
 {\cal R}^*_{q,mmr} =\left( \frac{p_U - p_L}{2} \right)^2\,. \nonumber
\]

\paragraph*{Log Loss.} Finally, we derive the forecast that solves (\ref{eq:minimax.regret}) 
for the log loss function from (\ref{eq:loss.log}) and decision space  ${\cal D} = [0,1]$. The minimax forecast is the $\theta$-optimal forecast if $p_L = p_U$. Suppose $p_L < p_U$. The regret of any $d \in [0,1]$ is the Kullback--Leibler (KL) divergence
\[
 {\textstyle \PR_\theta} ( Y = 1 ) \log \left(\frac{{\textstyle \PR_\theta} ( Y = 1 )}{d}\right) 
 + {\textstyle \PR_\theta} ( Y = 0 ) \log \left(\frac{{\textstyle \PR_\theta} ( Y = 0 )}{1-d}\right)\,.
\]
By convexity, the maximum regret must be obtained at either $p_L$ or $p_U$:
\begin{equation} \label{eq:r.p.minimax.regret}
 \sup_{\theta \in \Theta_0} \Big( \mathbb E_\theta [ \ell_p(Y , d) ] -\mathbb E_\theta [ \ell_p(Y , d^*_{p,\theta}) ]  \Big) 
 = \max_{p \in \{p_L, p_U\}} \left(p \log \left( \frac{p}{d} \right) + (1-p) \log \left( \frac{1-p}{1-d} \right) \right) \,.
\end{equation}
When $p = p_L$, term in parentheses is increasing for $d \geq p_L$ and when $p = p_U$ the term in parentheses is decreasing for $d \leq p_L$. The maximum regret is therefore minimized by choosing $d$ to equate the two values. The minimax regret forecast under the log-scoring rule $d_{p,mmr}$ uniquely solves
\begin{equation} \label{eq:r.p.minimax.regret.forecast}
 \log \left( \frac{ d_{p,mmr} }{1 - d_{p,mmr}} \right)  = \frac{h(p_U) - h(p_L)}{p_U - p_L} \,,
\end{equation}
where $h(p) = - p \log p - (1-p) \log (1-p)$ is the entropy of a Bernoulli distribution with success probability $p$. The minimax regret forecast is therefore the value $p$ that minimizes the maximum KL divergence between the Bernoulli distribution and the forecast distribution $\PR_\theta$ over $\theta \in \Theta_0$.

\subsection{Efficient Robust Forecasts}
\label{subsec:binary.efficient}

We now dispense with the assumption that $\Theta_0$ is known. We consider the setting described at the end of Section \ref{sec:setup} in which the econometrician wishes to forecast $Y$ having observed data $X_n$. In order to develop an optimality theory, we evaluate forecasts by their {\em integrated maximum risk}, defined as
\be
 {\cal B}_{mm}^n(d_n;\pi) = \int \int \sup_{\theta \in \Theta_0(P)} \; \mathbb{E}_\theta[\ell(Y,d(X_n))] \, \mathrm d \Pi_n(P|X_n) \, \mathrm dF_n(X_n) \,.
\ee
We will use a similar definition of {\em integrated maximum regret}, denoted by ${\cal B}_{mmr}(d_n;\pi)$.
Here, $\pi$ is a prior distribution for $P$ with support ${\cal P}$, $\Pi_n(P|X_n)$ is the posterior distribution of $P$ after having observed the data $X_n$, and $F_n(X_n)$ is the marginal distribution of the data $X_n$. As in Section~\ref{subsec:binary.robust}, conditional on $P \in \mathcal P$, we consider the maximum risk or regret over $\Theta_0(P)$. This is the maximum risk faced by the forecaster if $P$ were the true reduced-form parameter. We then average across the joint distribution of $(P,X_n)$ to obtain the integrated maximum risk. The factorization of the joint distribution in the conditional distribution $\Pi_n(P|X_n)$ and the marginal distribution $F_n(X_n)$ highlights the well-known result that the integrated (maximum) risk is minimized by choosing the forecast that minimizes the {\em posterior (maximum) risk} for each realization $X_n$. We denote the optimal forecasts under the risk and regret objectives by $d_{b,mm}$ and $d_{b,mmr}$, respectively. We refer to them as {\em Bayesian robust forecasts} and derive explicit formulas in the remainder of this subsection.

\begin{remark} \normalfont
While it may seem asymmetric to use an integrated (or Bayes) criterion to deal with $P$ but minimax risk or regret to deal with $\theta$ conditional on $P$, there are two reasons for doing so. The first is from a robust Bayes perspective on the forecasting problem under partial identification. If the true value $P_0$ is identified and consistently estimable, then the posterior for $P$ will not depend on $\pi$ asymptotically. In contrast, the data do not update prior beliefs about $\theta$ over the identified set $\Theta_0$. Therefore, the posterior for $\theta$ in a Bayesian implementation will depend on the prior asymptotically (see, e.g., \cite{MoonSchorfheide2012}). Our use of minimax criteria to deal with partial identification of $\theta$ can be motivated from robustness considerations with respect to the prior on $\theta$ \citep{Kitagawa2012,GiacominiKitagawa2018}. The second is practical: average criteria  lead to tractable, easily implementable forecasts. $\Box$
\end{remark}

\subsubsection{Minimax Forecasts}

We now make dependence of $p_L$ and $p_U$ on the reduced-form parameter explicit by writing 
\begin{align*}
 p_L(P) & := \inf_{\theta \in \Theta_0(P)} \mathbb P_\theta(Y = 1) \,,  &
 p_U(P) & := \sup_{\theta \in \Theta_0(P)} \mathbb P_\theta(Y = 1) \,.
\end{align*} 

\paragraph*{Binary (or Classification) Loss.}
In view of (\ref{eq:r.d.minimax}), the posterior average maximum risk of choosing $d \in \mathcal \{0,1\}$ is 
\begin{equation}\label{eq:r.d.minimax}
	\int \sup_{\theta \in \Theta_0(P)} \; \mathbb{E}_\theta [ \ell_b(Y, d)  ]  \, \mathrm d \Pi_n(P|X_n)
	= \left[ \begin{array}{ll}
	a_{01}(1 - \int p_L(P) \, \mathrm d \Pi_n(P|X_n)) & \mbox{if $d = 1$}\,, \\
	a_{10}  \int p_U(P) \, \mathrm d \Pi_n(P|X_n)  & \mbox{if $d = 0$}\,.
	\end{array} \right.
\end{equation}
The Bayesian robust forecast is therefore
\begin{equation} \label{eq:efficient.robust.binary}
 d_{b,mm}(X_n) = \mathbb{I} \left[ a_{01} \le \int \left( a_{01} p_L(P)  +  a_{10} p_U(P) \right) \, \mathrm d \Pi_n(P | X_n)  \right] \,.
\end{equation}

\paragraph*{Quadratic Loss.} In view of (\ref{eq:r.q.minimax}), the posterior average maximum risk of choosing $d \in [0,1]$ is
\[
  \int \sup_{\theta \in \Theta_0(P)} \; \mathbb{E}_\theta [ \ell_q(Y, d)  ] \, \mathrm d \Pi_n(P|X_n)
  = \left[ \begin{array}{ll}
  (\int p_U(P) \, \mathrm d \Pi_n(P|X_n) ) (1-2d) + d^2 & \mbox{if $d < \frac{1}{2}$}\,, \\
  (\int p_L(P) \, \mathrm d \Pi_n(P|X_n) )(1-2d) + d^2 & \mbox{if $d > \frac{1}{2}$}\,, \\
  \frac{1}{4} & \mbox{if $d = \frac{1}{2}$} \,.
  \end{array} \right.
\]
The Bayesian robust forecast is therefore 
\be
 d_{q,mm}(X_n) = \left[ \begin{array}{ll}
  \int p_U(P) \, \mathrm d \Pi_n(P|X_n)  & \mbox{if $\int p_U(P) \, \mathrm d \Pi_n(P|X_n)  \leq \frac{1}{2}$}\,, \\
  \int p_L(P) \, \mathrm d \Pi_n(P|X_n)  & \mbox{if $\int p_L(P) \, \mathrm d \Pi_n(P|X_n) \geq \frac{1}{2}$}\,, \\
  \frac{1}{2} & \mbox{otherwise} 
  \end{array} \right. 
\ee

\paragraph*{Log Loss.} The forecast $d_{q,mm}(X_n)$ is also the Bayesian robust forecast for the log loss function $\ell_p$ from (\ref{eq:loss.log}) and decision space $\mathcal D = [0,1]$; see Appendix \ref{appsec:binary}.

\subsubsection{Minimax Regret Forecasts}

\paragraph*{Binary Loss.} In view of (\ref{eq:regret.binary.pointwise}), the posterior average maximum regret for $d \in \{0,1\}$ is
\begin{align*}
  & \int \sup_{\theta \in \Theta_0(P)} \bigg( \mathbb{E}_\theta [ \ell_b(Y,d)  ] - a_{10} {\textstyle \PR_\theta}(Y = 1) \, \wedge \, a_{01} {\textstyle \PR_\theta}(Y = 0) \bigg) \, \mathrm d \Pi_n(P|X_n) \\
 & = \left[ \begin{array}{ll}
 \left(a_{01} - (a_{01} + a_{10}) (\int p_L(P) \, \mathrm d \Pi_n(P|X_n))  \right)_+ & \mbox{if $d = 1$,} \\
 \left( (a_{01} + a_{10}) \, (\int p_U(P) \, \mathrm d \Pi_n(P|X_n)) - a_{01} \right)_+ & \mbox{if $d = 0$.}
 \end{array} \right.
\end{align*}
The Bayesian robust forecast is therefore
\begin{equation} \label{eq:efficient.robust.binary.regret}
 d_{b,mmr}(X_n)
 = \mathbb{I} \left[   \int  \left( {\textstyle \frac{a_{01}}{a_{01} + a_{10}}} - p_L(P)   \right)_+ \! \mathrm d \Pi_n(P | X_n) \leq \int  \left( p_U(P)   - {\textstyle \frac{a_{01}}{a_{01} + a_{10}}} \right)_+ \! \mathrm d \Pi_n(P | X_n) \right] \,.
\end{equation}

\paragraph*{Quadratic Loss.} In view of (\ref{eq:r.q.minimax.regret}), the posterior average maximum risk of choosing $d \in [0,1]$ is
\[
  \int (p_U(P) - d)^2 \mathbb I \left[ d < \frac{p_L(P) + p_U(P)}{2} \right] + (p_L(P) - d)^2  \mathbb I \left[ d \geq \frac{p_L(P) + p_U(P)}{2}\right] \, \mathrm d \Pi_n(P|X_n) \,.
\]
The Bayesian robust forecast $d_{q,mmr}(X_n)$ is the minimizing value for $d \in [0,1]$, which can be computed numerically (e.g. by replacing the integral with the average across a large number of draws from the posterior then minimizing with respect to $d$).

\paragraph*{Log Loss.} In view of (\ref{eq:r.p.minimax.regret}), the posterior average maximum risk of choosing $d \in [0,1]$ is 
\[
 \int \max_{p \in \{p_L(P), p_U(P)\}} \left(p \log \left( \frac{p}{d} \right) + (1-p) \log \left( \frac{1-p}{1-d} \right) \right)  \, \mathrm d \Pi_n(P |X_n)
\]
The Bayesian robust forecast $d_{p,mmr}(X_n)$ is the minimizing value for $d \in [0,1]$, which again can be computed numerically.

\subsubsection{Summary} Table~\ref{tab:summary.binary} summarizes the $\theta$-optimal, the robust, and the Bayesian robust decisions under binary, quadratic, and log loss functions.

\begin{table}[t!]
	\begin{center}
		\begin{tabular}{ll}
			\hline \hline
			\multicolumn{2}{c}{\bf $\theta$-optimal:} \\
			\multicolumn{2}{c}{$d^*_\theta := \argmin_{d \in \mathcal{D}} 
				\mathbb{E}_\theta[\ell(Y,d)]$} \\[4pt]
			\hline
			Binary $\ell_b$ & $d_{b, \theta}^* = \mathbb{I} [ (a_{01} + a_{10}){\textstyle \PR_\theta}(Y = 1) \ge a_{01} ]$ \\
			Quadratic $\ell_q$ & $d_{q,\theta}^* = {\textstyle \PR_\theta}(Y = 1)$ \\
			Log $\ell_p$ & $d_{p,\theta}^* = d_{q,\theta}^* $ \\ \hline
			\multicolumn{2}{c}{\bf Robust (minimax):}\\
			\multicolumn{2}{c}{$d_{mm} := \argmin_{d \in \mathcal{D}} \left( \sup_{\theta \in 
					\Theta_0} \mathbb{E}_\theta[\ell(Y,d)] \right) $} \\[4pt] \hline
			Binary  $\ell_b$ & $d_{b,mm} = \mathbb{I} \left[ a_{01} p_L  +  a_{10} p_U  \geq a_{01}  \right]$ \\[2pt]
			Quadratic $\ell_q$ & $
			d_{q,mm} = \left[ \begin{array}{ll}
			p_U & \mbox{if $p_U \leq \frac{1}{2}$}\,, \\[2pt]
			p_L & \mbox{if $p_L \geq \frac{1}{2}$}\,, \\[2pt]
			\frac{1}{2} & \mbox{otherwise}\,, 
			\end{array} \right. \label{eq:minimax.quadratic}
			$ \\
			Log  $\ell_p$ & $d_{p,mm} = d_{q,mm} $ \\ \hline
			\multicolumn{2}{c}{\bf Robust (minimax regret):}\\
			\multicolumn{2}{c}{$d_{mmr} := \argmin_{d \in \mathcal{D}} \left( \sup_{\theta 
					\in 
					\Theta_0} \mathbb{E}_\theta[\ell(Y,d)] - 
				\mathbb{E}_\theta[\ell(Y,d^*_\theta)]  \right) $}  \\[4pt] \hline
			Binary  $\ell_b$ & $ d_{b,mmr}
			= \mathbb{I} [  ({\textstyle \frac{a_{01}}{a_{01} + a_{10}}} - p_L )_+ \leq ( p_U- {\textstyle \frac{a_{01}}{a_{01} + a_{10}}} )_+ ]$ \\
			Quadratic  $\ell_q$ & $d_{q,mmr} = \frac{1}{2}(p_L + p_U)$ \\
			Log  $\ell_p$ & $d_{q,mmr} = $ see equation (\ref{eq:r.p.minimax.regret.forecast}) \\ \hline
			\multicolumn{2}{c}{\bf Bayesian robust (minimax):}  \\
			\multicolumn{2}{c}{$d_{mm}(X_n) := \argmin_{d \in \mathcal{D}} \int \sup_{\theta \in 
					\Theta_0(P)} \mathbb{E}_\theta[\ell(Y,d)] \, \mathrm d \Pi_n $} \\[4pt] \hline
			Binary  $\ell_b$ & $d_{b,mm}(X_n) = \mathbb{I} [ a_{01} \le \int \left( a_{01} p_L(P)  +  a_{10} p_U(P) \right) \, \mathrm d \Pi_n  ]$ \\[2pt]
			Quadratic  $\ell_q$ & $d_{q,mm}(X_n) = \left[ \begin{array}{ll}
			\int p_U(P) \, \mathrm d \Pi_n  & \mbox{if $\int p_U(P) \, \mathrm d \Pi_n  \leq \frac{1}{2}$}\,, \\[2pt]
			\int p_L(P) \, \mathrm d \Pi_n  & \mbox{if $\int p_L(P) \, \mathrm d \Pi_n \geq \frac{1}{2}$}\,, \\[2pt]
			\frac{1}{2} & \mbox{otherwise} 
			\end{array} \right.  $ \\
			Log  $\ell_p$ & $d_{p,mm}(X_n) = d_{q,mm}(X_n)$ \\ \hline
			\multicolumn{2}{c}{\bf Bayesian robust (minimax regret)}  \\
			\multicolumn{2}{c}{$d_{mmr}(X_n) := \argmin_{d \in \mathcal{D}} \int (
				\sup_{\theta 
					\in 
					\Theta_0(P)}  \mathbb{E}_\theta[\ell(Y,d)] - 
				\mathbb{E}_\theta[\ell(Y,d^*_\theta)] ) \, \mathrm d \Pi_n  $} \\[4pt]
			\hline
			Binary $\ell_b$ & $d_{b,mmr}(X_n)
			= \mathbb{I} [   \int  ( {\textstyle \frac{a_{01}}{a_{01} + a_{10}}} - p_L(P)   )_+  \mathrm d \Pi_n \leq \int  ( p_U(P)   - {\textstyle \frac{a_{01}}{a_{01} + a_{10}}} )_+  \mathrm d \Pi_n ]$ \\
			Quadratic  $\ell_q$ & $d_{q,mmr}(X_n) = \arg \min_{d \in [0,1]} 
			\int (p_U(P) - d)^2 \mathbb I [ d < \frac{p_L(P) + p_U(P)}{2} ] \, \mathrm d \Pi_n $\\
			& \hspace*{2.3cm} $+ \int(p_L(P) - d)^2  \mathbb I [ d \geq \frac{p_L(P) + p_U(P)}{2}] \, \mathrm d \Pi_n
			$ \\
			Log  $\ell_p$ & $d_{p,mmr}(X_n) = \arg \min_{d \in [0,1]} 
			\int \max_{p \in \{p_L(P), p_U(P)\}} (p \log ( \frac{p}{d} ) + (1-p) \log ( \frac{1-p}{1-d} ) )  \, \mathrm d \Pi_n
			$ \\ \hline
		\end{tabular}

\vspace*{0.5cm}

\parbox{14cm}{\caption{ \label{tab:summary.binary} Summary of binary forecasts for binary loss $\ell_b$ from (\ref{eq:loss.binary}), quadratic loss $\ell_q$ from (\ref{eq:loss.quadratic}), and log loss $\ell_p$ from (\ref{eq:loss.log}). To simplify notation, $\Pi_n$ denotes $\Pi_n(P|X_n)$.}}

	\end{center}
	
\end{table}

\subsection{Numerical Illustration}
\label{subsec:example}

We close this section with a numerical illustration to show how uncertainty about the true $\theta \in \Theta_0$ can induce substantial variation in the implied forecast distribution in nonlinear models. We use a panel probit design from \cite{HonoreTamer2006}. The model is as in Example 1 with $\Phi_t$ taken to be the standard normal cdf for all $t$. The distribution $\Pi_{\lambda,y}$ is unspecified, but $\lambda$ is assumed to be supported on the discrete evenly-spaced grid $\{-3,-2.8,\ldots,2.8,3\}$. Under the true data-generating process, $\lambda$ and $Y_{i0}$ are independent with $Y_{i0}$ taking the value $0$ or $1$ with probability $\frac{1}{2}$ and the probability mass for $\lambda$ is assigned by interpolating a $N(0,1)$ distribution on the support points.\footnote{See p.619 in \cite{HonoreTamer2006} for details.}

In this example, we wish to forecast $Y_{iT+1}$ having observed $Y_i^T$. In our earlier notation, the outcome of interest $Y$ represents $Y_{iT+1}$ and the forecast probability $\PR_\theta(Y = 1)$ denotes the probability under $\theta$ that $Y_{iT+1} = 1$ given $Y_i^T$; see display (\ref{eq:prob:panel}). The identified set $\Theta_0$ consists of all $(\beta,\Pi_{\lambda,y})$ that can match the model-implied probabilities of observing sequences $Y_i^T = y^T$ with the true probabilities for all $y^T \in \{0,1\}^T$; see display (\ref{eq:idset:panel}).

We may compute the set of forecast probabilities $\{\PR_\theta(Y = 1) : \theta \in \Theta_0\}$ by adapting linear programming methods from \cite{HonoreTamer2006} as follows. Denote the support of $\Pi_{\lambda,y}$ by $(\lambda_1,y_{01}), \ldots, (\lambda_L,y_{0L})$. As the support of $\Pi_{\lambda,y}$ is discrete, we identify $\Pi_{\lambda,y}$ with a $L$-vector $\pi \in \Delta^{L-1} := \big\{ x \in \mathbb R^L_+ : \sum_{l=1}^L x_l = 1 \big\}$ and write the restrictions defining $\Theta_0$ in  display (\ref{eq:idset:panel}) as $G(\beta) \pi = r$. The matrix $G(\beta)$ is a $2^T \times L$ matrix whose $l^{\mathrm{th}}$ column $G_l(\beta)$ is the $2^T$-vector of model-implied probabilities of observing different realizations of $Y_i^T$ conditional on $\lambda_i = \lambda_l$ and $Y_{i0} = y_{0l}$:
\[
 G_l(\beta) = \Big( p(y^T | y_{0l},\lambda_l;  \beta ) \Big)_{y^T \in \{0,1\}^T} \,,
\]
where
\be
 p(y^T | y_0,\lambda;  \beta ) = \prod_{t=1}^{T} \Phi(\beta y_{t-1} + \lambda)^{y_t}(1-\Phi(\beta y_{t-1} + \lambda))^{1-y_t} \,, \label{eq:example.cond.prob}
\ee
and $r = \big(p(y^T)\big)_{y^T \in \{0,1\}^T}$ is the $2^T$-vector that collects the true probabilities of each realization of the sequences $Y_{i}^T$. The forecast probability ${\textstyle \PR_\theta} ( Y_{iT+1} = 1 |Y_i^T =y^T )$ from display (\ref{eq:prob:panel}) may also be written as $b(\beta)'\pi$ where $b(\beta)$ is an $L$-vector whose $l$\textsuperscript{th} entry is
\[
 b_l(\beta) =  \frac{ \Phi(\beta y_T + \lambda_l) p(y^T | y_{0l},\lambda_l;  \beta ) }{ p(y^T ) } \,,
\]
where
 $y_T$ denotes the element of $y^T$ corresponding to date $T$. Note that we can place $p(y^T)$ in the denominator because $\int   p(y^T | y_0,\lambda;  \beta ) \mathrm{d}\Pi_{\lambda,y}(\lambda,y_0) =p(y^T)$ for any $\theta \in \Theta_0$. 

The problem of computing $p_U$ can be written as
\[
 p_U = \sup_\beta \left( \sup_{\pi \in \Delta^{L-1}} b(\beta)'  \pi \quad \mbox{s.t.} \quad  G(\beta) \pi = r \right) \,,
\]
with the understanding that the value of the inner maximization over $\pi$ is $-\infty$ if there does not exist a $\pi$ for which $G(\beta) \pi = r$. For such values of $\beta$ there does not exist a $\Pi_{\lambda,y}$ such that the model can explain the observed probabilities up to date $T$. As shown in Appendix~\ref{appsec:computation}, the inner optimization over $\Pi_{\lambda,y}$ can be rewritten as a linear program, leading to the equivalent representation
\be
 p_U 
 = \sup_\beta \left(  \inf_{ v \in \mathbb{R}^{K+1}}  \left[  0_{1 \times K} \, , \, 1 \right] v \quad \mbox{s.t.} \quad  A(\beta)  v \leq  - b(\beta) \right) \,, \label{eq:ex7.pU}
\ee
where $K = 2^T$ and $A(\beta) = \left[ G(\beta)' - (1_{L \times 1} \otimes r') \, , \, - 1_{L \times 1} \right]$ with $\otimes$ denoting the Kronecker product. The lower value is computed similarly; see Appendix~\ref{appsec:computation}.

Suppose $T = 2$ and the true $\beta_0 = 0.2$. The identified set for $\beta$ is approximately $[-2.4403,1.2428]$; these are all values of $\beta$ for which there is a $\Pi_{\lambda,y}$ such that the model can explain the observed probabilities up to date $T = 2$. The limits of the identified set for $\beta$ are denoted as grey dashed vertical lines in Figure \ref{f:example1}. For each value of $\beta$ in this set, we compute the smallest and largest values of the forecast probability $\PR_\theta(Y = 1)$ subject to the constraint that $(\beta,\Pi_{\lambda,y}) \in \Theta_0$. The linear programming problem to compute the upper probability conditional on $\beta$ is characterized in parentheses on the right-hand side of (\ref{eq:ex7.pU}).
The range of forecast probabilities as a function of $\beta$ is shown as the shaded regions in Figure \ref{f:example1} for different values of $Y_{iT}$. Maximizing and minimizing with respect to $\beta$ yields the values $p_L$ and $p_U$; these are marked as black dotted horizontal lines in Figure \ref{f:example1}. Although each $(\beta, \Pi_{\lambda,y}) \in \Theta_0$ induces identical distributions over $Y_i^T$, they induce different distributions over $Y_{iT+1}$ and therefore different $\theta$-optimal forecasts. As a consequence, some parameterizations that are indistinguishable based on $T$ observations become distinguishable based on $T+1$ observations and the identified set shrinks over time. This feature is due to the sequential learning about heterogeneous parameters that generates a non-Markovian structure of the model.

\begin{figure}
\centering
\begin{subfigure}{.5\textwidth}
\centering
\includegraphics[width=\linewidth]{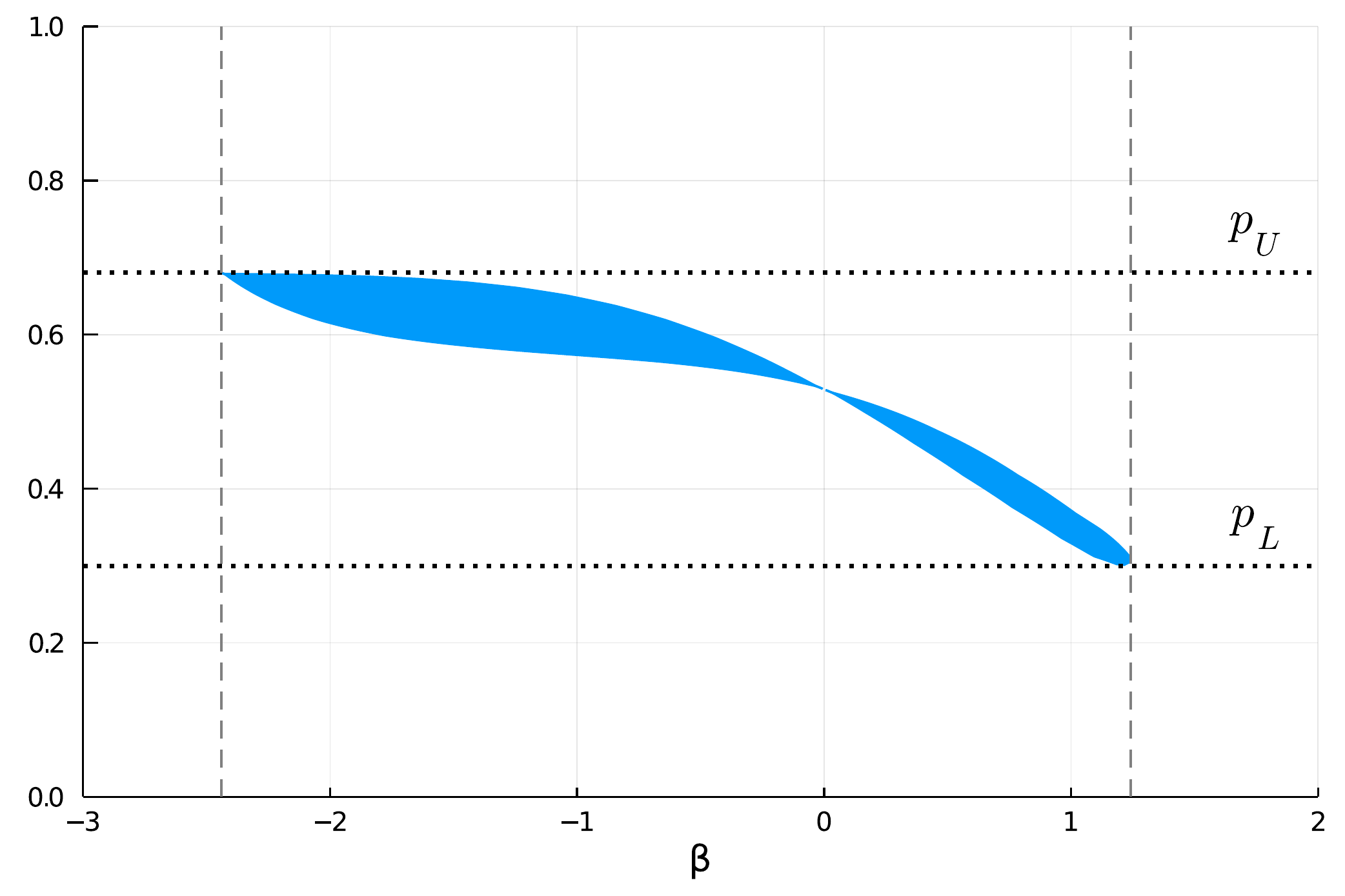}
\caption{$Y_{iT} = 0$}
\end{subfigure}%
\begin{subfigure}{.5\textwidth}
\centering
\includegraphics[width=\linewidth]{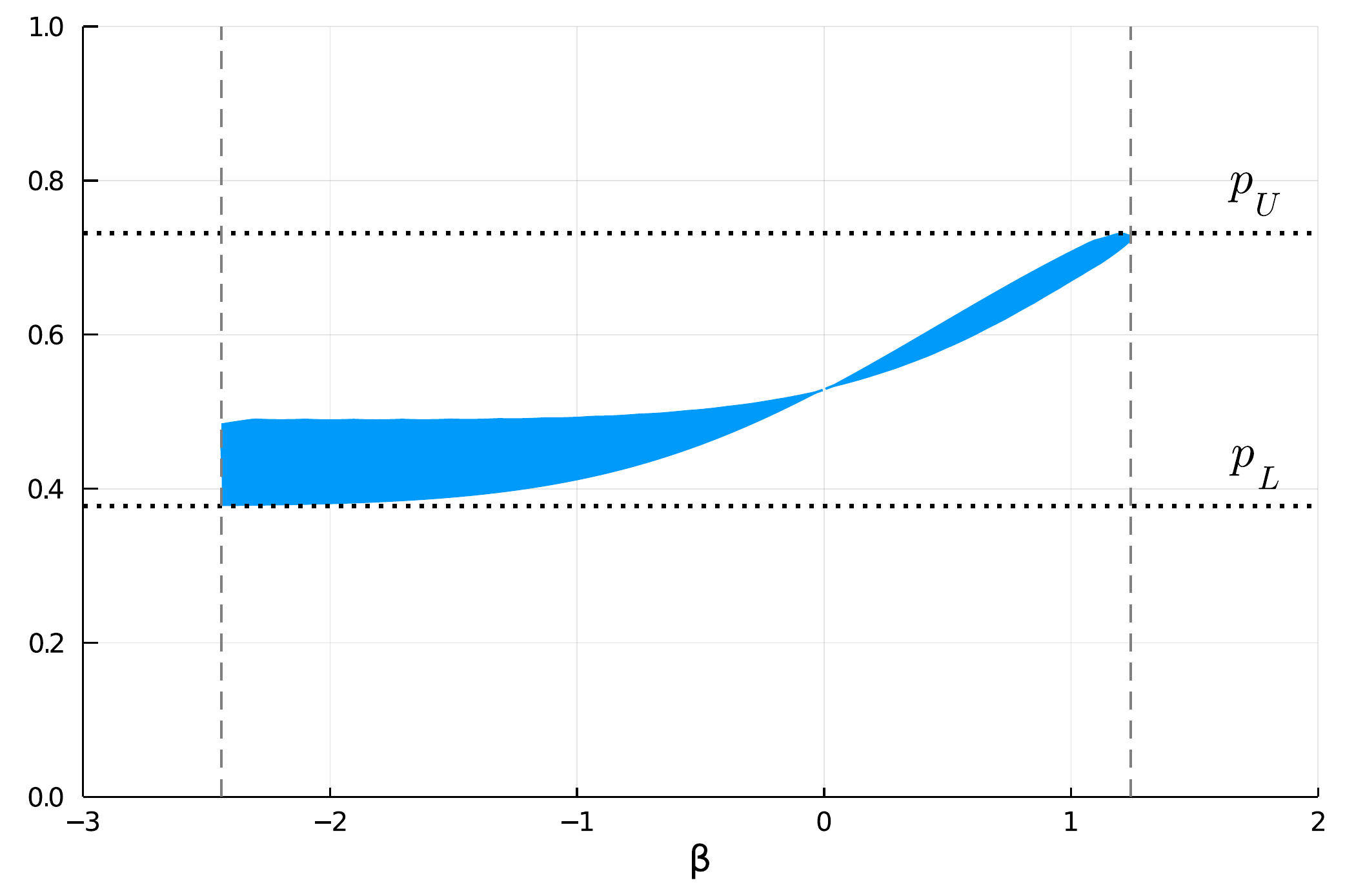}
\caption{$Y_{iT} = 1$}
\end{subfigure}
\parbox{14cm}{\caption{\small \label{f:example1} Panel probit example with $T = 2$ and $\beta_0 = 0.2$. Shaded regions denote the sets $\{\PR_\theta (Y = 1) : (\beta,\Pi_{\lambda,y}) \in \Theta_0\}$ as a function of $\beta$. Black dotted lines denote $p_L$ and $p_U$. }}
\end{figure}

In this numerical example, we consider robust forecasts which take  $\Theta_0$ as known. This is the asymptotic problem faced by the forecaster in a large-$n$, fixed-$T$ setting. 
Suppose we condition on $Y_i^T$ with $Y_{iT} = 0$.\footnote{In this design, the conditional distribution of $Y_{iT+1}$ given $Y_i^T$ depends only on $Y_{iT}$. } 
The set of forecast probabilities $\{\PR_\theta(Y = 1) : \theta \in \Theta_0\}$ is wide, spanning from $p_L = 0.2997$ to $p_U = 0.6803$ (see the left panel of Figure \ref{f:example1}). In particular, there are $\theta \in \Theta_0$ for which $\PR_\theta(Y = 1) < \frac{1}{2}$ so the $\theta$-optimal decision would be $d_{b,\theta}^* = 0$ for these $\theta$. However, there are other $\theta \in \Theta_0$ for which $\PR_\theta(Y = 1) > \frac{1}{2}$ and therefore the corresponding $\theta$-optimal decision would be $d_{b,\theta}^* = 1$ for these $\theta$. Our robust forecasts are useful here as the forecaster has no way to discriminate among $\theta \in \Theta_0$ based on date-$T$ information. As $p_L + p_U < 1$, the minimax and minimax regret forecast for symmetric binary loss is therefore $d_{b,mm} = d_{b,mmr} = 0$. Similarly, when $Y_{iT} = 1$ the set of forecast probabilities $\{\PR_\theta(Y = 1) : \theta \in \Theta_0\}$ is again quite wide, spanning $p_L = 0.3775$ to $p_U = 0.7320$ (see the right panel of Figure \ref{f:example1}). Here $p_L + p_U > 1$ so $d_{b,mm} = d_{b,mmr} = 1$.

\section{Multinomial Forecasts}
\label{sec:multinomial}

We now extend the preceding analysis to multinomial forecasts. We first describe $\theta$-optimal forecasts with known $\theta$ (Subsection~\ref{subsec:multinomial.point}), then describe forecasts that are robust with respect to the set of forecasting models $\{\PR_\theta : \theta \in \Theta_0\}$ with $\Theta_0$ known (Subsection~\ref{subsec:multinomial.robust}), before concluding with efficient robust forecasts that deal with both model and sampling uncertainty (Subsection~\ref{subsec:multinomial.efficient}). The forecasts are summarized in Table~\ref{tab:summary.multinomial} at the end of this section.

\subsection{$\theta$-optimal Forecasts}
\label{subsec:multinomial.point}

Throughout this section we focus on \emph{classification loss} for the decision space $\mathcal{D} = \{0,1,\ldots,M\}$:
\be
 \ell_c(y,d) = \mathbb{I}[y \neq d]\,.
 \label{eq:loss.classification}
\ee
This loss function generalizes binary loss in the symmetric case (i.e., $a_{01} = a_{10}$) to multinomial forecasts.\footnote{It is straightforward to modify what follows to penalize some types misclassifications more heavily than others, as we did in the binary case. We adopt the equal-weighted specification (\ref{eq:loss.classification}) for notational convenience.}
The $\theta$-optimal forecast in this environment under $\PR_\theta$ is the most likely outcome:
\be
 d_{c,\theta}^* \in \arg \max_m \, {\textstyle \PR_\theta} ( Y = m ) \,.
\label{eq:forecast.bayes.classification}
\ee
In the above display we write ``$\in$'' to allow for the possibility of ties.
When the $\arg \max$ is non-singleton, any element of the set of minimizers is a $\theta$-optimal point forecast. Any $\theta$-optimal forecast has risk 
\[
 1 - \max_m \, {\textstyle \PR_\theta} ( Y = m )\,.
\]

\subsection{Robust Forecasts}
\label{subsec:multinomial.robust}

We now derive the minimax and minimax regret forecasts that solve the decision problems (\ref{eq:minimax}) and (\ref{eq:minimax.regret}) for the classification loss function $\ell_c$ from (\ref{eq:loss.classification}) and decision space $\mathcal D = \{0,1,\ldots,M\}$. 

\subsubsection{Minimax Forecasts}

In the multivariate case, the analogues of $p_L$ and $p_U$ are the $M+1$ quantities
\be
 \underline p_{m} = \inf_{\theta \in \Theta_0} {\textstyle \PR_\theta} ( Y = m ) \,, \quad m \in \{0,1,\ldots,M\}\,.
 \label{eq:p.lower.multinomial}
\ee
Computation of $\underline p_m$ using duality methods is discussed in Appendix \ref{appsec:computation}. 
The maximum risk from choosing $d \in \mathcal D$ is 
\begin{equation} \label{eq:p.maximum.risk}
 \sup_{\theta \in \Theta_0} \mathbb{E}_\theta [ \ell_c(Y, d) ] = 1 - \underline p_{d} \,.
\end{equation}
The minimax forecast for classification loss is therefore
\be
	d_{c,mm} \in \arg \max_m \underline p_{m} \label{eq:d.minimax.multivariate}
\ee
and the minimax risk is 
\be
{\cal R}^*_{c,mm} = 1 - \max_m \underline p_{m} \,. 
\ee
As before, the minimax-optimal forecast is not necessarily unique. Non-uniqueness arises when the set of maximizers of $m \mapsto \underline p_m$ is not a singleton. If so, each minimax-optimal forecast differs only in its handling of ties and has the same maximum risk.

\subsubsection{Minimax Regret Forecasts}

For minimax regret forecasts, define
\begin{equation} \label{eq:dp.multinomial}
 \Delta p_{m} := \sup_{\theta \in \Theta_0} \left( \max_{m'} {\textstyle \PR_\theta} ( Y = m' ) - {\textstyle \PR_\theta} ( Y = m ) \right) \,.
\end{equation}
Suppose the forecaster chooses $d$. The difference $\max_{m'} \PR_\theta ( Y = m' ) - \PR_\theta ( Y = d ) $ is the regret from this choice under the forecast distribution $\PR_\theta$. Having chosen $d$, the quantity $\Delta p_d$ is therefore the forecaster's maximum regret over all $\theta \in \Theta_0$. The minimax regret forecast is therefore
\[
 d_{c,mmr} \in \arg \min_{m} \Delta p_m \,.
\]
and the minimax regret is 
\[
 \mathcal R^*_{c,mmr} = \min_{m} \Delta p_m \,.
\]
Unlike the binary case, equivalence of minimax and minimax regret forecasts for classification loss no longer holds when $M \geq 2$; see Appendix \ref{appsec:binary}. Computation of $\Delta p_m$ is discussed in Appendix \ref{appsec:computation}.

\subsection{Efficient Robust Forecasts}
\label{subsec:multinomial.efficient}

In this section we now drop the assumption that $\Theta_0$ is known and consider also the need to estimate features of $\Theta_0$ that are relevant for the forecasting problem from data. We consider the same setup and notation as developed for the binary case in Section \ref{subsec:binary.efficient}.

\subsubsection{Minimax Forecasts}

Here we make dependence on the reduced-form parameter explicit by defining 
\[
 \underline p_{m}(P) := \inf_{\theta \in \Theta_0(P)} \mathbb P_\theta(Y = m)\,.
\]
In view of (\ref{eq:p.maximum.risk}), the posterior average maximum risk of choosing $d \in \mathcal D$ is
\[
 \int \sup_{\theta \in \Theta_0(P)} \mathbb{E}_\theta [ \ell_c(Y, d) ] \, \mathrm d \Pi_n(P|X_n) = 1 - \int \underline p_{d}(P)  \, \mathrm d \Pi_n(P|X_n)
\] 
The Bayesian robust forecast is therefore
\begin{equation}\label{eq:efficient.robust.classification}
 d_{c,mm}(X_n) \in \arg \max_m \left( \int \underline p_{m}(P) \, \mathrm d \Pi_n(P|X_n) \right) \,.
\end{equation} 
In case of ties, any (possibly randomized) tie-breaking rule is optimal. For instance, one could simply choose the smallest value of $m$ among the set of maximizers.

\subsubsection{Minimax Regret Forecasts}

For minimax regret forecasts, define
\[
 \Delta p_m(P) := \sup_{\theta \in \Theta_0(P)} (\max_{m'} \mathbb P_\theta(Y = m') -  \mathbb P_\theta(Y = m))\,.
\]
The posterior average maximum regret of choosing $d \in \mathcal D$ is 
\[
 \int \Delta_{p_d}(P) \, \mathrm d \Pi_n(P|X_n) \,.
\]
The Bayesian robust forecast is therefore
\begin{equation}\label{eq:efficient.robust.classification.regret}
 d_{c,mmr}(X_n) \in \arg \min_{m} \left( \int  \Delta p_{m}(P) \, \mathrm d \Pi_n(P|X_n) \right) \,.
\end{equation}
In case of ties, any (possibly randomized) tie-breaking rule is optimal. For instance, one could simply choose the smallest value of $m$ among the set of minimizers.

\begin{table}[t!]
	
	\begin{center}
		\begin{tabular}{ll}\hline \hline
			$\theta$-optimal & $d_{c,\theta}^* \in \arg \max_m \, {\textstyle \PR_\theta} ( Y = m ) $ \\
			Robust (minimax) & $d_{c,mm} \in \arg \min_{m} \Delta p_m $ \\
			Robust (minimax regret) & $d_{c,mmr} \in \arg \min_{m} \Delta p_m $ \\
			Bayesian robust (minimax) & $d_{c,mm}(X_n) \in \arg \max_m \; \big( \int \underline p_{m}(P) \, \mathrm d \Pi_n(P|X_n) \big)$ \\
			Bayesian robust (minimax regret) & $d_{c,mmr}(X_n) \in \arg \min_{m} \; \big( \int  \Delta p_{m}(P) \, \mathrm d \Pi_n(P|X_n) \big) $ \\ \hline
		\end{tabular}

\vspace*{0.5cm}

\parbox{14cm}{\caption{ \label{tab:summary.multinomial} Summary of multinomial forecasts for classification loss $\ell_c$  from (\ref{eq:loss.classification}).}}

\end{center}
			
\end{table}	

\section{Asymptotic Efficiency for the Robust Forecasting Problem}
\label{sec:uncertainty}

In this section we focus on forecasts that are asymptotically efficient-robust. We continue to evaluate the forecasts by their {\em integrated maximum risk} (or regret), but only require this criterion to be minimized in the limit as the sample size $n$ tends to infinity. This enlarges the class of efficient forecasts to those that are asymptotically equivalent to the Bayesian robust forecast. 

Some interesting findings emerge. First, ``plug-in'' rules, in which an efficient estimator $\hat P$ is plugged into the rules derived in Sections \ref{subsec:binary.robust} and \ref{subsec:multinomial.robust}, are not asymptotically efficient-robust if the key quantities which determine the robust forecast (i.e., $p_L(P)$ and $p_U(P)$ in the binary case) are only directionally differentiable functions of $P$. This stands in contrast with other asymptotic efficiency results for related problems that depend smoothly on first-stage estimators, including point estimation under partial identification \citep{Song2014} and efficient statistical treatment rules under point identification \citep{HiranoPorter2009}, for which plug-in rules are efficient.
Second, forecasts that are constructed via {\em bagging} tend to be 
asymptotically efficient-robust. To construct such forecasts, the posterior 
distribution for $P$ is replaced by the bootstrap distribution of an efficient 
estimator of $P$. The forecast is then chosen to minimize the maximum risk or 
regret over $\Theta_0(P)$ averaged across the bootstrap distribution. As discussed in Remark~\ref{remark:bagging} below, it can be shown that the 
bagged forecasts are asymptotically equivalent to the Bayesian robust forecasts 
even under directional differentiability.

\subsection{Limit Experiment}\label{subsec:efficiency.conditions}

Our approach follows \cite{HiranoPorter2009} and uses Le Cam's limits of experiments framework. As is standard for treatments of asymptotic efficiency (see, e.g., \cite{Vandervaart2000}), we work with a local reparameterization in which the reduced form parameter is $P_{n,h} = P_0 + n^{-1/2}h$ for $P_0$ fixed and $h$ ranging over $\mathbb R^k$. 
Let $\overset{P_{n,h}}{\rightsquigarrow}$ and $\overset{P_{n,h}}{\to}$ denote convergence in distribution and in probability under the sequence of measures $\{F_{n,P_{n,h}}\}_{n \geq 1}$. The model for $X_n$ is \emph{locally asymptotically normal} at $P_0$ if for each $h_0 \in \mathbb R^k$,  the likelihood ratio processes indexed by any finite subset $H \subset \mathbb R^k$ converge weakly to the likelihood ratio in a shifted normal model:
\be \label{e:LAN}
 \bigg( \frac{\mathrm d F_{n,P_{n,h}}}{\mathrm d F_{n,P_{n,h_0}}} \bigg)_{h \in H}  \overset{P_{n,h_0}}{\rightsquigarrow} \;\; \bigg( \exp \left( (h - h_0)' Z - \frac{1}{2} (h - h_0)'I_0 (h - h_0) \right) \bigg)_{h \in H}
\ee
with $Z  \sim N(h_0,I_0^{-1})$ and $I_0$ nonsingular. Let $\mathbb E_h$ and $\mathbb P_h$ denote expectation and probability with respect to $Z \sim N(h,I_0^{-1})$.

\begin{assumption}\label{A1}
	\hspace*{1cm}\\[-3ex]
\begin{enumerate}[nosep]
\item $\mathcal P$ is an open subset of $\mathbb R^k$ with $P_0 \in \mathcal P$;
\item The model for $X_n$ is locally asymptotically normal at each $P_0 \in \mathcal P$.
\end{enumerate}
\end{assumption}

\noindent
In the dynamic binary choice example, Assumption \ref{A1}.1. 
implies we observe all possible realizations of histories $Y_i^T \in \{0,1\}^T$ up to time $T$ with positive probability.

To describe the limit experiment, consider the collection $\mathbb D$ of sequences of forecasts $\{d_n\}_{n \geq 1}$ that converge in distribution under $\{F_{n,P_{n,h}}\}_{n \geq 1}$:
\begin{equation}\label{eq:sequence.D}
 \mathbb D = \left\{ \{d_n\}_{n \geq 1} : d_n(X_n) \overset{P_{n,h}}{\rightsquigarrow} Q_{P_0,h} \mbox{ for all } h \in \mathbb R^k \mbox{ and } P_0 \in \mathcal P \right\} \,,
\end{equation}
where $Q_{P_0,h}$ denotes a probability measure on $\mathcal D$ equipped with its Borel $\sigma$-algebra. 
Assumption \ref{A1} permits application of an asymptotic representation theorem of \cite{Vandervaart1991}. For any $\{d_n\}_{n \geq 1} \in \mathbb D$ there exists a function $d^\infty_{P_0}(Z,U) $ with  $d^\infty_{P_0}(Z,U) \sim Q_{P_0,h}$ where $Z \sim N(h,I_0^{-1})$ and $U \sim\;$Uniform$[0,1]$ independently of $Z$ is a randomization term. As $n \to \infty$, the average excess maximum risk and regret of any sequence of forecasts $\{d_n\}_{n \geq 1} \in \mathbb D$ converges to a limiting counterpart for its representation $d^\infty_{P_0}(Z,U)$ in the limit experiment. 

\subsection{Asymptotic Efficiency}

The robust forecasts derived in Sections \ref{subsec:binary.robust} and \ref{subsec:multinomial.robust} are \emph{oracle} forecasts in the sense that they were obtained under the assumption of knowledge of the true  set $\Theta_0$. To distinguish the oracle forecasts from the data-dependent forecasts $d_{mm}(X_n)$ and $d_{mmr}(X_n)$, in what follows we use the notation $d_{mm}^o(P)$ and $d_{mmr}^o(P)$ to denote the oracle forecasts when $P$ is the true reduced-form parameter.
To facilitate the asymptotic calculations, we evaluate forecasts by their excess maximum risk or regret relative to the oracle. The \emph{excess maximum risk} of $d_n(X_n)$ is
\[
 \Delta\mathcal R_{mm}(d_n; P, X_n) = \sup_{\theta \in \Theta_0(P)} \mathbb E_\theta[ \ell(Y,d_n(X_n))] - \sup_{\theta \in \Theta_0(P)} \mathbb E_\theta[ \ell(Y,d_{mm}^o(P))] \,.
\]

Integration over $(P,X_n)$ leads to the {\em integrated excess maximum risk}
\[
   \Delta {\cal B}_{mm}^n(d_n;\pi) = \int \int \sqrt{n} \Delta\mathcal R_{mm}(d_n; P, X_n) \, \mathrm d\Pi_n(P|X_n) \, \mathrm dF_n(X_n).
\]
We standardize by $\sqrt{n}$ and recenter at the maximum risk of the oracle to ensure that $\Delta {\cal B}_{mm}^n(d_n;\pi)$ converges to a finite but potentially non-zero limit, though this does not change the ranking of forecasts. Therefore, the Bayes robust forecast under minimax risk also minimizes $\Delta {\cal B}_{mm}^n(d_n;\pi)$.
Excess maximum regret $\Delta\mathcal R_{mmr}$ and integrated excess maximum regret $\Delta {\cal B}_{mmr}^n(d_n;\pi)$ are defined similarly, replacing risk in the above display with regret.

To derive the asymptotic counterparts,
we begin by calculating a frequentist risk that averages over the data $X_n$ conditional on $P$ and then integrates over $P$ using the prior $\pi$. To express the frequentist risk, let  $\mathbb{E}_{P_{n,h}}$ denote the expectation with respect to $X_n \sim F_{n,P_{n,h}}$. Using this notation and conducting the change-of-variables from $P_{n,h}$ to $h$, the {\em frequentist excess maximum risk} can be expressed as
\begin{equation} \label{e:Bmmh}
 \Delta \mathcal B_{mm}^n(d_n;P_0,\pi) = \int  \mathbb E_{P_{n,h}}\left[ \sqrt n \Delta\mathcal R_{mm}\left(d_n, P_{n,h}; X_n\right) \right] \, \pi \left( P_{n,h} \right) \, \mathrm d h \,.
\end{equation}
Here we dropped the Jacobian term that arises from the change-of-variables because it simply scales the average excess maximum risk by a power of $n$ without changing the ranking of forecasts. We include $P_0$ in the conditioning set to indicate that the calculations are done locally around $P_0$. The regret $\Delta \mathcal  B_{mmr}^n(d_n;P_0,\pi)$ can be expressed in a similar manner. \emph{asymptotically efficient-robust forecasts} are those that minimize $\lim_{n \to \infty} \Delta \mathcal B_{mm}^n(d_n;P_0,\pi)$ and $\lim_{n \to \infty} \Delta \mathcal B_{mmr}^n(d_n;P_0,\pi)$. Under conditions permitting the interchange of limits and integration, we obtain
\[
 \lim_{n \to \infty} \Delta \mathcal B_{mm}^n(d_n;P_0,\pi) = \pi(P_0) \int \left( \lim_{n \to \infty} \mathbb E_{P_{n,h}}\left[ \sqrt n \Delta\mathcal R_{mm}\left(d_n, P_{n,h}; X_n\right) \right] \right) \, \mathrm d h\,,
\]
and similarly for regret. The limit in parentheses will depend on the sequence of forecasts $\{d_n\}_{n \geq 1}$ through its representation in the limit experiment. Also note that the asymptotic ranking of forecasts does not depend on $\pi$.

The following example illustrates the normalization of the excess maximum risk and the use of the local reparameterization.

\paragraph*{Example 7: Local parameters and frequentist excess maximum risk.} 
\phantom{}
 Suppose that $\mathcal P = (0,1)$, $p_L(P) = P$, and
\[
p_U(P) = \left[ \begin{array}{ll}
\frac{1}{2} & P < \frac{1}{2}\,, \\
(2P - \frac{1}{2}) \wedge 1 & P \geq \frac{1}{2}\,,
\end{array} \right.
\]
For a binary loss function with $a_{01} = a_{10} = 1$, the robust forecast takes the form
\be
d_{mm}^o(P) = \mathbb{I}[ 1 \le p_L(P) + p_U(P)] = \mathbb{I}[ P \ge {\textstyle \frac{1}{2}}].
\label{eq:ex7.dmmoP}
\ee
If the true $P$ is bounded away from $\frac{1}{2}$, then eventually we will learn whether it is less or greater than $\frac{1}{2}$ and make the optimal decision. The most challenging case is when $P$ is very close to $\frac{1}{2}$. Thus, we center the local reparameterization at $P_0=\frac{1}{2}$. Suppose that under $P_{n,h}$ the frequentist sampling distribution and the Bayesian posterior for $P$ are
\[
 \hat P|P_{n,h} \sim N(P_{n,h},n^{-1}) \,, \quad 
 P|X_n \sim N(\hat{P},n^{-1}) \,,
\]
respectively. 
The posterior is obtained under a uniform prior on ${\cal P}$ when $n$ is large enough so that the truncation effect of the prior at the boundary of $(0,1)$ is negligible. Under the local reparameterization, the sampling distribution of $\hat{h} = \sqrt{n}(\hat P-P_0)$ and the posterior distribution for $h = \sqrt n (P - P_0)$ are
\[
 \hat{h}|(P_0,h_0) \sim N(h_0,1) \,, \quad 
 h|(X_n,P_0) \sim N(\hat{h},1) \,,
\]
respectively.

In this example $\hat P$ (equivalently $\hat{h}$) is a sufficient statistic. For any decision $d_n(\hat{h})$, we obtain
\begin{equation*}
\sup_{\theta \in \Theta_0(P_{n,h_0})} \; \mathbb{E}_\theta [ \ell(Y, d_n(\hat{h})) ] 
= \left[ \begin{array}{ll}
1 - p_L(P_{n,h_0}) & \mbox{if $d_n(\hat{h}) = 1 $}\,, \\
p_U(P_{n,h_0})  & \mbox{if $d_n(\hat{h}) = 0 $}\, .
\end{array} \right.
\end{equation*}
Here $1 - p_L(P_{n,h_0})=1/2 - n^{-1/2} h_0$ is linear in $h_0$ whereas
\[
p_U(P_{n,h_0}) = \left\{ \begin{array}{ll} {\textstyle \frac{1}{2}} + 2 n^{-1/2} h_0 & \mbox{if } h_0 \ge 0 \\  
	                          {\textstyle \frac{1}{2}} & \mbox{if } h_0 < 0 \,.
\end{array} \right. 
\]
Using straightforward algebra it can be shown that
\be
\mathbb{E}_{P_{n,h_0}}[\sqrt{n}\Delta {\cal R}_{mm} (d_n(\hat{h}),P_{n,h_0};\hat{h})] 
= \left\{ \begin{array}{rl}
 3 h_0 \mathbb{P}_{n,h_0} \big[ d_n(\hat{h})= 0 \big] & \mbox{if } h_0 \ge 0 \,, \\
 -h_0 \mathbb{P}_{n,h_0} \big[ d_n(\hat{h}) = 1 \big] & \mbox{if } h_0 < 0 \,, \end{array}
 \right. \label{eq:ex7.freqrisk}
\ee
where $\mathbb P_{n,h_0}$ denotes probabilities under $F_{n,P_{n,h_0}}$. It follows that for any sequence $\{d_n\} \in \mathbb D$,
\[
 \lim_{n \to \infty} \mathbb{E}_{P_{n,h_0}}[\sqrt{n}\Delta {\cal R}_{mm} (d_n(\hat{h}),P_{n,h_0};\hat{h})] 
= \left\{ \begin{array}{rl}
 3 h_0 \mathbb{P}_{h_0} \big[ d^\infty_{P_0}(Z)= 0 \big] & \mbox{if } h_0 \ge 0 \,, \\
 -h_0 \mathbb{P}_{h_0} \big[ d^\infty_{P_0}(Z) = 1 \big] & \mbox{if } h_0 < 0 \,, \end{array}
 \right. 
\]
where $Z \sim N(h_0,1)$ under $\mathbb P_{h_0}$.\footnote{Here it is without loss of generality to write $d^\infty_{P_0}$ as a function of $Z$ only; see the discussion in Appendix~\ref{appsec:binary.proofs}.} The formula shows that the $\sqrt{n}$ standardization leads to a well-defined non-trivial limit of the frequentist excess maximum risk. $\Box$

\subsection{Binary forecasts}

In this section we show that the efficient robust binary forecasts that were derived in Section \ref{subsec:binary.efficient} are optimal. For brevity, we focus on discrete forecasts with $\mathcal D = \{0,1\}$ under binary or classification loss. This class of forecasts is also relevant for its connections with statistical treatment rules. 

Say $f : \mathcal P \to \mathbb R^d$ is \emph{directionally differentiable} at $P_0$ if the limit
\[
\lim_{t \downarrow 0} \frac{f(P_0 + t h) - f(P_0)}{t} =: \dot f_{P_0}[h] 
\]
exists for every $h \in \mathbb R^k$, in which case $ \dot f_{P_0}[\cdot]$ is its directional derivative. Note that $ \dot f_{P_0}[\cdot]$ will be positively homogeneous of degree one but not necessarily linear. If $\dot f_{P_0}[h]$ is linear in $h$ we say that $f$ is \emph{fully} differentiable at $P_0$.   We say that the posterior $\Pi_n$ is \emph{consistent} if $ \Pi_n(P \in N |X_n) \overset{P_0}{\to} 1$ for every neighborhood $N$ containing $P_0$ for each $P_0 \in \mathcal P$. Recall that $Z \sim N(h,I_0^{-1})$ in the limit experiment. Let $\mathbb P_h$ denote probability statements with respect to $Z$. Moreover, let $Z^* \sim N(0,I_0^{-1})$ independently of $Z$ and $\mathbb E^*$ denote expectation with respect to $Z^*$.

\medskip

\begin{assumption}\label{a:binary}
	\hspace*{1cm}\\[-3ex]
	\begin{enumerate}[nosep]
		\item
		\begin{enumerate}[nosep]
			\item The functions $p_U$ and $p_L$ are everywhere continuous and everywhere directionally differentiable;
			\item The function $x \mapsto \mathbb P_h\left( \mathbb E^*[ \dot p_{L,P_0}[Z^* + Z] + \dot p_{U,P_0}[Z^* + Z]  |Z ] \leq x \right) $ is continuous at $x = 0$ for each $h \in \mathbb R^k$ and $P_0 \in \mathcal P$ with $a_{01} p_L(P_0) + a_{10} p_U(P_0) = a_{01}$ and $p_U(P_0) > \frac{a_{01}}{a_{01} + a_{10}}$;
		\end{enumerate}
		\item
		\begin{enumerate}[nosep]
			\item The posterior for $P$ is consistent;
			\item For any neighborhood $N$ of $P_0$ there is $\gamma > \frac{1}{2}$ such that $n^{\gamma}  \Pi_n(P \not \in N|X_n) \overset{P_0}{\to} 0$;
		\end{enumerate}
		\item
		\begin{enumerate}[nosep]
			\item At any $P_0 \in \mathcal P$ with $a_{01} p_L(P_0) + a_{10} p_U(P_0) = a_{01}$, for any Borel set $A$,
			\[
			\lim_{n \to \infty} F_{n,P_{n,h}}\left(\int \sqrt n(f(P) - f(P_0)) \, \mathrm d \Pi_n(P|X_n) \in A \right) = \mathbb P_h\left( \mathbb E^*[ \dot f_{P_0}[Z^* + Z] |Z ] \in A \right) 
			\]
			with  $f = (p_L,p_U)$;
			\item Similarly, for any Borel set $A$, 
			\[
			\lim_{n \to \infty} F_{n,P_{n,h}}\left(\int \sqrt n(f(P) - f(P_0))_+ \, \mathrm d \Pi_n(P|X_n) \in A \right) = \mathbb P_h\left( \mathbb E^*[ ( \dot f_{P_0}[Z^* + Z] )_+ |Z ] \in A \right)
			\]
			with $f = (p_L,p_U)$, where $(\cdot)_+$ is applied element-wise.
		\end{enumerate}
	\end{enumerate}
\end{assumption}

\noindent The directional differentiability of Assumption \ref{a:binary}.1 was built into the functional form of $p_U(\cdot)$ in Example 7. Appendix \ref{appsec:computation} shows that in a broad class of problems the extreme probabilities $p_L(P)$ and $p_U(P)$  can be expressed as min-max or max-min problems, where the outer optimization is over homogeneous parameters and the inner optimization is a linear or convex program. It follows that $p_L(\cdot)$ and $p_U(\cdot)$ are  typically only directionally, rather than fully, differentiable functions.\footnote{See, e.g., Theorem 3.1 of \cite{Greenberg1997} for directional differentiability of the value of linear programs, Chapter 4.3 of \cite{BonnansShapiro2000} for directional differentiability of the value of convex programs, and \cite{MilgromSegal2002} and \cite{Shapiro2008} for directional differentiability of min-max problems.} 
Directional differentiability can also be a feature of models defined via moment inequalities (cf. Example 5). 

Assumption \ref{a:binary}.2(a) holds under standard regularity conditions (see, e.g., Chapter 10.4 of \cite{Vandervaart2000}). For Assumption \ref{a:binary}.2(b), note that $\Pi_n(P \not \in N|X_n) $ typically converge to zero exponentially. For instance, in a normal means model with $\bar X_n \sim N(P_{n,h},(nI_0)^{-1})$ and a flat prior on $P$, we have $\Pi_n(P \not \in N|X_n) = O(e^{-c n})$ for some $c > 0$.\footnote{Note $P | X_n \sim N(\bar X_n,(nI_0)^{-1})$ under a flat prior. Choose $\varepsilon > 0$ so that $B_{2\varepsilon}(P_0) \subset N$. Then $P_0(\bar X_n \in B_\varepsilon(P_0)) \to 1$ and whenever $\bar X_n \in B_\varepsilon(P_0)$, we have $\Pi_n(P \not \in N|X_n) \leq \Pi_n(|P - \bar X_n| > \varepsilon|X_n) = O(e^{-c n})$ for any $c < \frac{\varepsilon^2}{2}  \lambda_{\min}(I_0)$.} More generally, the classical posterior consistency results of \cite{Schwartz1965} establish exponential convergence rates.

Assumption \ref{a:binary}.3 is simply assuming that a $\delta$-method applies for the posterior distribution of directionally differentiable functionals of $P$.\footnote{See, e.g., \cite{KitagawaOleaPayneVelez2020} for a formal justification. This may require strengthening our definition of directional differentiability to Hadamard directional differentiability.} For a heuristic justification for Assumption \ref{a:binary}.3(a), consider a normal means model with $\bar X_n \sim N(P_{n,h},(nI_0)^{-1})$. Under a flat prior for $P$, we have $P | X_n \sim N(\bar X_n,(nI_0)^{-1})$. For a directionally differentiable function $f$:
\begin{align*}
& F_{n,P_{n,h}}\left(\int \sqrt n(f(P) - f(P_0)) \, \mathrm d \Pi_n(P|X_n) \in A \right) \\
& = F_{n,P_{n,h}}\left(  \int \sqrt n(f(P) - f(P_0)) \frac{e^{-\frac{1}{2}(P - \bar X_n)'(n I_0)(P - \bar X_n)}}{\sqrt{|2 \pi (n I_0)^{-1}|}} \,  \mathrm d P \in A \right) \\
& \approx F_{n,P_{n,h}}\left( \int \dot f_{P_0}[\sqrt n(P - P_0)] \frac{e^{-\frac{1}{2}(P - \bar X_n)'(nI_0)(P - \bar X_n)}}{\sqrt{|2 \pi (n I_0)^{-1}|}} \, \mathrm d P \in A \right) \\
& = F_{n,P_{n,h}}\left( \int \dot f_{P_0}[\kappa + \sqrt n(\bar X_n - P_0)] \frac{e^{-\frac{1}{2}\kappa'(I_0)\kappa}}{\sqrt{|2 \pi (I_0)^{-1}|}} \, \mathrm d \kappa \in A \right) \,,
\end{align*}
with the final line is by the change of variables $\kappa = \sqrt n (P - \bar X_n)$. The last integral can be rewritten $\mathbb E^*[ \dot f_{P_0}[Z^* + Z] |Z ] $ where $Z^* |Z \sim N(0,I_0^{-1})$ with $Z = \sqrt n (\bar X_n - P_0) \sim N( h,I_0^{-1})$ under $F_{n,P_{n,h}}$. A similar argument provides a heuristic justification for Assumption \ref{a:binary}.3(b).
In the presentation of the asymptotic efficiency results for the binary forecasts we rely on the following definition:

\begin{definition}\label{def:asymp.equiv.forec}
	Given a sequence of forecasts $\{d_n\}_{n \geq 1} \in \mathbb D$, we say that $d_n$ is \emph{asymptotically equivalent} to $d_{b,mm}$ if $d_n(X_n)$ and $d_{b,mm}(X_n)$ have the same asymptotic distribution under the sequence of measures $\{F_{n,P_{n,h}}\}_{n \geq 1}$ for all $P_0 \in \mathcal P$ and $h \in \mathbb R^k$.
\end{definition}

\noindent
Let $\Delta \mathcal B_{b,mm}^n(\,\cdot\,;P_0,\pi)$ and  $\Delta \mathcal B_{b,mmr}^n(\,\cdot\,;P_0,\pi)$ denote integrated excess maximum risk and regret (see display (\ref{e:Bmmh})) for binary loss $\ell_b$ from (\ref{eq:loss.binary}).  We require forecasts to satisfy an additional technical condition, namely condition (\ref{eq:reverse.fatou}) in the Appendix, which permits the interchange of limits and integration. This condition can be verified under more primitive conditions (see Remark \ref{rmk:fatou}). Let $\mathbb D$ denote the set of all sequences of $\{0,1\}$-valued forecasts that converge in the sense of (\ref{eq:sequence.D}). Theorem~\ref{t:optimal.binary} states that forecasts that are asymptotically equivalent to the Bayes forecasts are asymptotically optimal. A proof is provided in Appendix~\ref{appsec:proofs}.

\begin{theorem}\label{t:optimal.binary}
		(i) Let Assumption \ref{A1} and parts (a) of Assumption \ref{a:binary} hold and let $\tilde d_n$ be asymptotically equivalent to $d_{b,mm}$ and satisfy condition (\ref{eq:reverse.fatou}). Then: for all $P_0 \in \mathcal P$,
		\[
		\lim_{n \to \infty} \Delta \mathcal B_{b,mm}^n(\tilde d_n;P_0,\pi) = \inf_{\{d_n\} \in \mathbb D}  \liminf_{n \to \infty} \Delta \mathcal B_{b,mm}^n(d_n;P_0,\pi) \,.
		\]
		(ii) Let Assumptions \ref{A1} and \ref{a:binary} hold and let $\tilde d_n$ be asymptotically equivalent to $d_{b,mmr}$ and satisfy condition (\ref{eq:reverse.fatou}). Then: for all $P_0 \in \mathcal P$,
		\[
		\lim_{n \to \infty} \Delta \mathcal B_{b,mmr}^n(\tilde d_n;P_0,\pi) = \inf_{\{d_n\} \in \mathbb D}  \liminf_{n \to \infty} \Delta \mathcal B_{b,mmr}^n(d_n;P_0,\pi) \,.
		\]
\end{theorem}

\begin{remark}\label{remark:bayes} \normalfont
	The asymptotic efficiency extends to Bayes forecasts derived under priors that differ from the ``objective'' prior that is used to compute the integrated risk but assign positive density in the neighborhood of $P_0$. In large samples, the posterior is dominated by the likelihood function and the shape of the prior density does not affect the asymptotic form of the posterior distribution. The optimality result also extends to forecasts derived under a misspecified likelihood function, as long as this likelihood function leads to a large-sample posterior that reproduces the asymptotic form of the posterior under the ``true'' likelihood function.  $\Box$
\end{remark}

\begin{remark}\label{remark:bagging} \normalfont
	Given the asymptotic equivalence of posterior distributions of parameters and the bootstrap distributions of their efficient estimators,\footnote{This equivalence carries over to directionally differentiable function (see, e.g., \cite{KitagawaOleaPayneVelez2020}).} bagged estimators that replace posterior averaging with averaging across the bootstrap distribution of an efficient estimator of $P$ will yield asymptotically efficient forecasts under suitable modification of the above regularity conditions. $\Box$
\end{remark}

\begin{remark} \normalfont
	The optimal forecasting problem with $\mathcal D = \{0,1\}$ has a similar form to the optimal treatment problem studied by \cite{HiranoPorter2009} and our proofs follow similar arguments. In their setting, the oracle treatment rule is of the form $\mathbb I[g(P_0) \geq 0]$ where $g$ is (fully) differentiable. Their asymptotically efficient rule replaces $g(P_0)$ by $g(\hat P)$ where $\hat P$ is an efficient estimator of $P_0$. When $p_L$ and $p_U$ are directionally, rather than fully, differentiable, the optimal forecasts that we derive are of a different form from plugging $\hat P$ into the oracle rules. This difference arises because 
	\[
	\int \dot f_{P_0}[\sqrt n(P - P_0)] \, \mathrm d \Pi_n(P|X_n) \neq \dot f_{P_0}\left[\int \sqrt n(P - P_0) \, \mathrm d \Pi_n(P|X_n)\right] 
	\] 
	under directional differentiability. 
	If $p_L$ and $p_U$ are fully differentiable then both sides of the above display are equal and plugging in $\hat P$ into the oracle rule is asymptotically efficient. $\Box$
\end{remark}

\noindent
As we formalize in Proposition~\ref{prop:inefficient.binary} below, asymptotic equivalence to $d_{b,mm}$ (respectively, $d_{b,mmr}$) is a \emph{necessary} condition for a forecast $\tilde d_n$ to be asymptotically efficient-robust under minimax risk (respectively, regret) under a side condition ensuring that ties occur with probability zero. 

In view of Definition \ref{def:asymp.equiv.forec}, we say that asymptotic equivalence fails at $P_0$ if there exists some $h_* \in \mathbb R^k$ for which $d_n(X_n)$ and $d_{b,mm}(X_n)$ have \emph{different} asymptotic distributions under the sequence of measures $\{F_{n,P_{n,h_*}}\}_{n \geq 1}$ with $P_{n,h_*} = P_0 + n^{-1/2} h_*$. A proof for the following proposition is provided in Appendix~\ref{appsec:proofs}.

\begin{proposition}\label{prop:inefficient.binary}
(i) Let Assumption \ref{A1} and parts (a) of Assumption \ref{a:binary} hold, let $d_{b,mm}$ satisfy condition (\ref{eq:reverse.fatou}), and let $\{\tilde d_n\}_{n \geq 1} \in \mathbb D$  be a sequence of forecasts for which $\tilde d_n$ and $d_{b,mm}$ are \emph{not} asymptotically equivalent. Then for any $P_0$ at which asymptotic equivalence of $\tilde d_n$ and $d_{b,mm}$ fails:
\[
 \liminf_{n \to \infty} \Delta \mathcal B^n_{b,mm}(\tilde d_n;P_0,\pi) >  \inf_{\{d_n\} \in \mathbb D}  \liminf_{n \to \infty} \Delta \mathcal B_{b,mm}^n(d_n;P_0,\pi) 
\]
provided either (a) or (b) holds: \\
(a) $a_{01} p_{L}(P_0) + a_{10} p_{U}(P_0) \neq a_{01}$, \\
(b) $a_{01} p_{L}(P_0) + a_{10} p_{U}(P_0) = a_{01}$ and $\mathbb E^*\big[a_{01} \dot p_{L,P_0}[Z^* + Z] + a_{10} \dot p_{U,P_0}[Z^* + Z] \,| \, Z \big] \neq 0$ a.e.\\[4pt]
(ii) Let Assumptions \ref{A1} and \ref{a:binary} hold, let $d_{b,mmr}$ satisfy condition (\ref{eq:reverse.fatou}), and let $\{\tilde d_n\}_{n \geq 1} \in \mathbb D$ be a sequence of forecasts for which $\tilde d_n$ and $d_{b,mmr}$ are \emph{not} asymptotically equivalent. Then for any $P_0$ at which asymptotic equivalence of $\tilde d_n$ and $d_{b,mmr}$ fails:
\[
 \liminf_{n \to \infty} \Delta \mathcal B_{b,mmr}^n(\tilde d_n;P_0,\pi) > \inf_{\{d_n\} \in \mathbb D}  \liminf_{n \to \infty} \Delta \mathcal B_{b,mmr}^n(d_n;P_0,\pi) 
\]
provided either (a), (b), or (c) holds with $a = \frac{a_{01}}{a_{01} + a_{10}}$: \\
(a) $p_L(P_0) + p_U(P_0) \neq 2 a$, \\
(b) $p_L(P_0) + p_U(P_0) = 2 a$, $p_U(P_0) > p_L(P_0)$, and $\mathbb E^*\big[ \dot p_{L,P_0}[Z^* + Z] + \dot p_{U,P_0}[Z^* + Z] \,| \,Z \big] \neq 0$ a.e.,\\
(c) $p_L(P_0) = p_U(P_0) = a$ and $\mathbb E^* \big[(\dot p_{L,P_0}[Z^* + Z])_- + (\dot p_{U,P_0}[Z^* + Z])_+ \,| \, Z \big] \neq 0$ a.e.
\end{proposition}

\paragraph*{Example 7: (continued).}  The oracle forecast under the minimax \emph{risk} criterion was given in display (\ref{eq:ex7.dmmoP}).
In order to compute the Bayesian robust forecast we need to evaluate
\begin{eqnarray*}
	\lefteqn{\int (p_L(P) + p_U(P)) \, \mathrm d \Pi_n(P|X_n)} \\
	&=& \frac{1}{\sqrt{2\pi}} \int_{-\infty}^0 \big[ 1/2 + n^{-1/2} h + 1/2 \big] \exp \left\{ - \frac{1}{2}(h-\hat{h})^2\right\} \mathrm dh \\
	&& + \frac{1}{\sqrt{2\pi}} \int_0^{+\infty} \big[ 1/2 + n^{-1/2} h + 1/2+2 n^{-1/2}h \big] \exp \left\{ - \frac{1}{2}(h-\hat{h})^2\right\} \mathrm dh \\   
	&=& 1+ n^{-1/2} \bigg[  \hat{h} + 2\Phi_N(\hat{h}) \hat{h}  + 2 \phi_N(\hat{h})  \bigg] \,,
\end{eqnarray*}
where $\hat h = \sqrt n (\hat P - P_0)$ with $P_0 = \frac{1}{2}$ and $\Phi_N$ and $\phi_N$ denote the standard normal cdf and pdf. As $\hat P$ is a sufficient statistic, the forecasts only depend on the data $X_n$ though $\hat P$ or, equivalently, through $\hat h$.
We may deduce that 
\be
d_{b,mm}(\hat{h}) = \mathbb{I} \left[  \hat{h} \ge - \frac{2 \phi_N(\hat{h})}{1+2\Phi_N(\hat{h})}  \right].
\ee
Note that the term on the right-hand side of the inequality in the indicator function is always negative.
The plug-in forecast is obtained by replacing the unknown $P$ by $\hat P$, leading to $d^{plug}_{b,mm}(\hat P) = \mathbb I[p_L(\hat P) + p_U(\hat P) \geq 1]$. In the present example, the plug-in forecast can be expressed equivalently in terms of $\hat h$:
\be
d_{b,mm}^{plug}(\hat{h}) = \mathbb{I} \left[  \hat{h} \ge 0  \right].
\ee
The plug-in forecast is not asymptotically equivalent to the Bayesian robust forecast. In particular, the Bayesian robust forecast predicts $Y = 1$ more aggressively than the plug-in forecast. It can be verified by direct calculation in this example that this more aggressive forecast is asymptotically efficient-robust. It also follows from Proposition \ref{prop:inefficient.binary} that the plug-in forecast is not asymptotically efficient-robust. This is seen by noting that $a_{01} = a_{10} = 1$, $p_{L}(P_0) + p_U(P_0) = 1$, $\dot p_{L,P_0}[h] = h$, $\dot p_{U,P_0}[h] = 2(h)_+$, and so $\mathbb E^*\big[a_{01} \dot p_{L,P_0}[Z^* + Z] + a_{10} \dot p_{U,P_0}[Z^* + Z] \, \big| \,Z \big]$ reduces to $Z + 2\mathbb E^*[(Z^* + Z)_+ \,| \,Z ]$ which is nonzero almost everywhere.

\begin{figure}[t]
	\begin{center}
		\includegraphics[width = .65\textwidth]{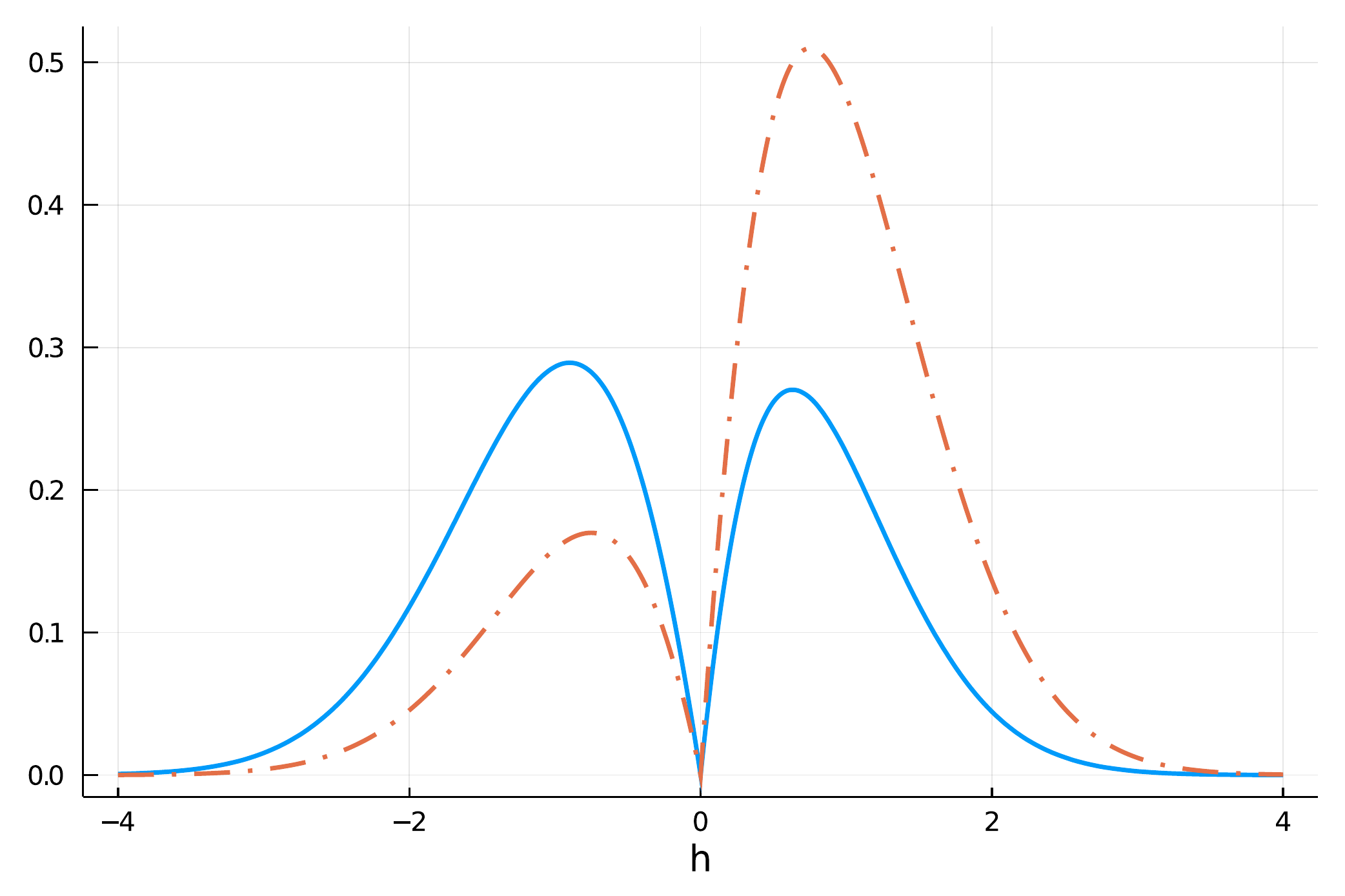}
		\parbox{14cm}{\caption{\small \label{f:limit_risk} Frequentist excess maximum {risk} in the limit experiment of the efficient robust forecast $d_{b,mm}$ (solid line) and the forecast $d^{plug}_{mm}$ based on plugging in an efficient estimator of $P$ (dot-dashed line) as a function of the location parameter $h_0$ for the Example~7.}}
	\end{center}
\end{figure}

To quantify the inefficiency of $d_{b,mm}^{plug}(\hat{h})$ relative to $d_{b,mm}(\hat{h})$, 
straightforward algebraic manipulations using (\ref{eq:ex7.freqrisk}) allow us to derive formulas for the frequentist excess maximum risk of $d_{b,mm}(\hat{h})$ and $d_{b,mm}^{plug}(\hat{h})$ as a function of $h_0$. The results are plotted in Figure~\ref{f:limit_risk}. The plug-in forecast is inferior to the Bayesian robust forecast from an integrated risk perspective: the area under the curve corresponding to $d_{b,mm}$ is around 20\% smaller than that under the curve corresponding to the plug-in forecast. While $d_{b,mm}$ was designed to be optimal from an integrated risk perspective, it also dominates the plug-in forecast from a minimax perspective: the maximum excess maximum risk of the plug-in forecast in the limit experiment is around 75\% larger than that of $d_{b,mm}$. 

\begin{figure}[t]
	\begin{center}
		\includegraphics[width = .65\textwidth]{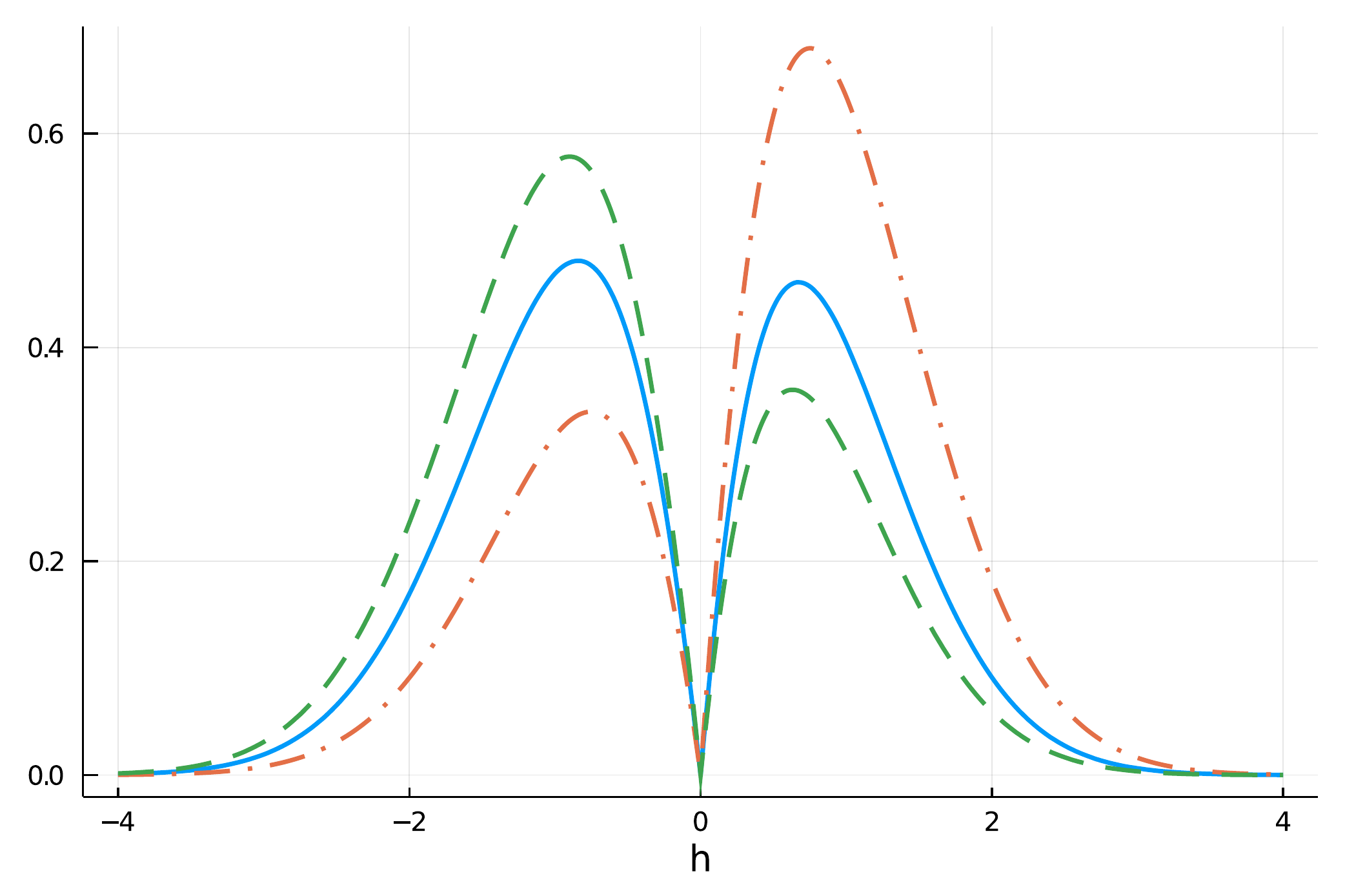}
		\parbox{14cm}{\caption{\small \label{f:limit_regret} Frequentist excess maximum {regret} in the limit experiment of the efficient robust forecast $d_{b,mmr}$ (solid line), the forecast $d^{plug}_{mmr}$ based on plugging an efficient estimator of $P$ (dot-dashed line), and the forecast $d_{b,mm}$ (dashed line) as a function of the location parameter $h_0$ for the Example~7.}}
	\end{center}
\end{figure}

Similar calculations can be made under the \emph{regret} criterion. The oracle forecast is of the form $d^o_{b,mmr}(P) = \mathbb{I} [  (\frac{1}{2} - p_L(P) )_+ \leq ( p_U(P) - \frac{1}{2} )_+ ]$. Similar calculations as for the risk criterion can be used to obtain a formula for the Bayesian robust forecast. In turns out that the plug-in forecast remains unchanged. 
Frequentist excess maximum regrets as a function of $h_0$ are plotted in Figure \ref{f:limit_regret}. Again, the plug-in forecast is not asymptotically efficient-robust and dominated by the Bayesian efficient robust forecast once we average across $h_0$. Its integrated excess maximum regret is around 8\% smaller and maximum excess maximum regret is around 41\% smaller. Also shown is the excess maximum regret of a forecast which plugs the posterior means of $p_L(P)$ and $p_U(P)$ into the oracle: $d^\dagger(X_n) = \mathbb{I} [  (\frac{1}{2} - \int p_L( P) \, \mathrm d \Pi_n(P|X_n) )_+ \leq ( \int p_U( P) \, \mathrm d \Pi_n(P|X_n)  - \frac{1}{2} )_+ ]$. This forecast is equivalent to the minimax forecast $d_{b,mm}$ and is therefore optimal for minimizing integrated excess maximum  risk but not necessarily integrated excess maximum regret. Figure \ref{f:limit_regret} shows that $d_{b,mmr}$ also dominates $d^\dagger$ in terms of both its average (2.5\% smaller) and maximum (21\% smaller) excess maximum regret in the limit experiment.  $\Box$  

\begin{remark}\normalfont
Consider the numerical example from Section \ref{subsec:example}. The optimization problems $p_U(P)$ and $p_L(P)$ can be recast as the value of max-min and min-max problems in which the reduced-form parameter $P$ enters the objective function. As is well known \citep{MilgromSegal2002,Shapiro2008}, the value of max-min and min-max problems is typically only directionally, rather than fully, differentiable.   $\Box$
\end{remark}

\subsection{Multinomial Forecasts}

We now turn to extending the asymptotic efficiency result to multinomial forecasts that are asymptotically equivalent to the Bayesian robust forecast from Section \ref{subsec:multinomial.efficient}. To do so, we first state some additional regularity conditions.

\begin{assumption}\label{a:multinomial}
	\hspace*{1cm}\\[-3ex]
	\begin{enumerate}[nosep]
		\item\begin{enumerate}[nosep]
			\item The functions $\underline p_0,\ldots,\underline p_M$ are everywhere continuous and everywhere directionally differentiable;
			\item The functions $\Delta p_0,\ldots,\Delta p_M$ are everywhere continuous and everywhere directionally differentiable;
		\end{enumerate}
		\item The posterior for $P$ is consistent;
		\item \begin{enumerate}[nosep]
			\item At any $P_0 \in \mathcal P$ with $\underline p_m(P_0) = \underline p_{m'}(P_0)$ for some $m' \neq m$ and $\underline p_m(P_0) \geq \underline p_k(P_0)$ for all $k \in \{0,\ldots,M\}$,  for any Borel set $A$ we have
			\[
			\lim_{n \to \infty} F_{n,P_{n,h}}\left(\int \sqrt n(f(P) - f(P_0)) \, \mathrm d \Pi_n(P|X_n) \in A \right) = \mathbb P_h\left( \mathbb E^*[ \dot f_{P_0}[Z^* + Z] |Z ] \in A \right) 
			\]
			with  $f = (\underline p_0,\ldots,\underline p_M)$;
			\item At any $P_0 \in \mathcal P$ with $\Delta p_m(P_0) = \Delta p_{m'}(P_0)$ for some $m' \neq m$ and $\Delta p_m(P_0) \leq \Delta p_k(P_0)$ for all $k \in \{0,\ldots,M\}$,  for any Borel set $A$ we have
			\[
			\lim_{n \to \infty} F_{n,P_{n,h}}\left(\int \sqrt n(f(P) - f(P_0)) \, \mathrm d \Pi_n(P|X_n) \in A \right) = \mathbb P_h\left( \mathbb E^*[ \dot f_{P_0}[Z^* + Z] |Z ] \in A \right) 
			\]
			with  $f = (\Delta p_0,\ldots,\Delta p_M)$;
		\end{enumerate}
	\end{enumerate}
\end{assumption}

\noindent
Assumption \ref{a:multinomial} is similar to Assumption \ref{a:binary}. In particular, a heuristic justification for Assumption \ref{a:multinomial}.3 follows similar reasoning to that presented earlier for Assumption \ref{a:binary}.3.

We now present the asymptotic efficiency results for multinomial forecasts. Let $\Delta \mathcal B_{c,mm}^n(\,\cdot\,;P_0,\pi)$ and  $\Delta \mathcal B_{c,mmr}^n(\,\cdot\,;P_0,\pi)$ denote integrated excess maximum risk and regret (see display (\ref{e:Bmmh})) for classification loss $\ell_c$ from (\ref{eq:loss.classification}). Also let $\mathbb D$ denote the set of all sequences of $\{0,\ldots,M\}$-valued forecasts that converge in the sense of (\ref{eq:sequence.D}).

\begin{theorem}\label{t:optimal.multinomial}
	(i) Let Assumption \ref{A1} and Assumption \ref{a:multinomial}.1(a), \ref{a:multinomial}.2, and \ref{a:multinomial}.3(a) hold and let $\tilde d_n$ be asymptotically equivalent to $d_{c,mm}$ and satisfy condition (\ref{eq:reverse.fatou}). Then: for all $P_0 \in \mathcal P$,
		\[
		\lim_{n \to \infty} \Delta \mathcal B_{c,mm}^n(\tilde d_n;P_0,\pi) = \inf_{\{d_n\} \in \mathbb D}  \liminf_{n \to \infty} \Delta  \mathcal B_{c,mm}^n(d_n;P_0,\pi) \,.
		\]
		(ii) Let Assumption \ref{A1} and Assumption \ref{a:multinomial}.1(b), \ref{a:multinomial}.2, and \ref{a:multinomial}.3(b) hold and let $\tilde d_n$ be asymptotically equivalent to $d_{c,mmr}$ and satisfy condition (\ref{eq:reverse.fatou}). Then: for all $P_0 \in \mathcal P$,
		\[
		\lim_{n \to \infty} \Delta  \mathcal B_{c,mmr}^n(\tilde d_n;P_0,\pi) = \inf_{\{d_n\} \in \mathbb D}  \liminf_{n \to \infty} \Delta  \mathcal B_{c,mmr}^n(d_n;P_0,\pi) \,.
		\]
\end{theorem}

\noindent
As with Remark \ref{remark:bagging}, bagged forecasts in which the posterior distribution is replaced with the bootstrap distribution of an efficient estimator of $P$ can be shown to be asymptotically efficient-robust under a suitable modification of the regularity conditions. As with Proposition \ref{prop:inefficient.binary}, it is possible to show that forecasts that are not asymptotically equivalent to the $d_{c,mm}$ and $d_{c,mmr}$ are not asymptotically efficient-robust under side conditions ruling out ties.

\section{Conclusion}
\label{sec:conclusion}

In this paper we proposed use of robust forecasts that are obtained by solving a minimax risk or minimax regret problem to deal with uncertainty about the forecast distribution. We also derived asymptotically efficient-robust forecasts that deal with the estimation of the set of forecast distributions. In addition to being useful for forecasting binary and multinomial outcomes,  these methods have wide applicability in environments in which a forecaster is concerned about structural breaks, model misspecification, or a policy maker has to make treatment assignments.

{
	\singlespacing
	\bibliographystyle{chicago}
	\bibliography{ref_partial-id-forecasting}{}
}


\clearpage

\renewcommand{\thepage}{A.\arabic{page}}
\setcounter{page}{1}

\begin{appendix}
	\markright{Online Appendix -- This Version: December 10, 2020 }
	\renewcommand{\theequation}{A.\arabic{equation}}
	\setcounter{equation}{0}
	
	\begin{center}
		
		{\Large {\bf Online Appendix: Robust Forecasting}}
		
		{\bf Timothy Christensen, Hyungsik Roger Moon, and Frank Schorfheide}
	\end{center}
	
	\section{Computation}
	\label{appsec:computation}

	The challenge in implementing the minimax and minimax regret forecasts is to solve the extremum problems $p_L$ and $p_U$ from (\ref{eq:p.lower}) and (\ref{eq:p.upper}) in the binary case, or $\underline p_m$  and $\Delta p_{m}$ from (\ref{eq:p.lower.multinomial}) and (\ref{eq:dp.multinomial}) in the multinomial case.
	
	We show how to compute these quantities in a class of models in which (i) vector of model parameters $\theta$ may be partitioned as $\theta = (\phi,\Pi)$, where $\phi$ is a low-dimensional parameter and $\Pi $ is a probability measure, and (ii) both the forecast probabilities and restrictions defining the set $\Theta_0$ are linear in $\Pi$. This nests semiparametric panel data models we study (Examples 1--4) and several other models, such as game-theoretic models (Example 5). In the next subsection, we show how linear programming techniques similar to \cite{HonoreTamer2006} may be used when the support of $\Pi$ is discrete.	Subsection~\ref{subsec:binary.computation.continuous} studies the continuous case.

	\subsection{Computing Extreme Probabilities: the Discrete Case}
	\label{subsec:binary.computation.discrete}

	\subsubsection{Binary forecasts}

	We consider a class of problems where the forecast probabilities and restrictions that define the  set $\Theta_0$ are linear in $\Pi$, where $\Pi$ has discrete support. We can identify $\Pi$ with a vector $\pi \in \Delta^{L-1}$, where $L$ is the number of points of support of $\Pi$ and $\Delta^{L-1} = \{x \in \mathbb R^L_+ : \sum_{i=1}^L x_i = 1\}$. We further assume that we can write the forecast probability as
	\be
	b(\phi)'\pi \,, \label{eq:m.p}
	\ee
	where $b(\phi)$ is a $L$-vector that may depend on the homogeneous parameters, and the restrictions defining $\Theta_0$ as
	\be
	G(\phi) \pi = r \,, \label{eq:m.G}
	\ee
	where $G(\phi)$ is a $K \times L$ matrix and $r \in \mathbb R^K$.
	
	Consider, for example, the semiparametric panel data model (Example 1). In that setting, the low-dimensional parameter $\phi$ is $\beta$, the probability measure $\Pi$ is the joint distribution $\Pi_{\lambda,y}$ of $(\lambda_i,Y_{i0})$, and the parameter space is $\Theta = \{(\beta,\Pi_{\lambda,y})\}$. The identified set is the collection of all $(\beta,\Pi_{\lambda,y})$ such that the model-implied probabilities of observing each realization of $Y_i^T$ is equal to the population probability $p(y^T)$; see display (\ref{eq:idset:panel}). The model-implied probabilities are given by
	\[ 
	p(y^T | \beta , \Pi_{\lambda,y}) = \int p(y^T | y_0,\lambda;  \beta ) \, \mathrm{d} \Pi_{\lambda,y} (\lambda,y_0) \,,
	\]
	with $p(y^T | y_0,\lambda;  \beta )$ from display (\ref{eq:example.cond.prob}). Because $p(y^T|\beta,\Pi_{\lambda,y})=p(y^T)$ for any $\theta \in \Theta_0$, 
	the forecast probability given $Y_i^T = y^T$ is
	\[
	{\textstyle \PR_\theta} ( Y_{iT+1} = 1 |Y_i^T =y^T ) = \frac{ \int  \Phi(\beta y_{iT} + \lambda) p(y^T | y_0,\lambda;  \beta ) \mathrm{d}\Pi_{\lambda,y}(\lambda,y_0) }
	{ p(y^T ) } \,.
	\]

	Returning to the general case with forecast probabilities as in (\ref{eq:m.p}) and restrictions defining $\Theta_0$ as in (\ref{eq:m.G}), we can write $p_U$ as 
	\[
	p_U = \sup_\phi \left( \sup_{\pi \in \Delta^{L-1}} b(\phi)'  \pi \quad \mbox{s.t.} \quad  G(\phi) \pi = r \right) \,.
	\]
	As we show in the following proposition, the inner optimization over $\pi$ can be written as a linear program, simplifying computation.
	
	\begin{proposition} \label{prop:duality.binary}
		The program 
		\[
		p_\phi^{\phantom *} = \sup_{\pi \in \Delta^{L-1}} b(\phi)'  \pi \quad \mbox{s.t.} \quad  G(\phi) \pi = r
		\]
		has an equivalent dual formulation
		\[
		p^*_\phi = \inf_{ v \in \mathbb{R}^{K+1}}  \left[  0_{1 \times K} \, , \, 1 \right] v \quad \mbox{s.t.} \quad  A(\phi)  v \leq  - b(\phi) \,,
		\]
		where $A(\phi) = \left[ G(\phi)' - (1_{L \times 1} \otimes r') \, , \, - 1_{L \times 1} \right]$ with $\otimes$ denoting  the Kronecker product.
	\end{proposition}
	
	\noindent
	In view of Proposition \ref{prop:duality.binary}, we may compute $p_U$ by solving
	\be
	p_U 
	= \sup_\phi \left(  \inf_{ v \in \mathbb{R}^{K+1}}  \left[  0_{1 \times K} \, , \, 1 \right] v \quad \mbox{s.t.} \quad  A(\phi)  v \leq  - b(\phi) \right) \,.\label{eq:extreme.dual3}
	\ee
	If $\phi$ is not feasible, i.e., if there does not exist $\pi \in \Delta^{L-1}$ solving (\ref{eq:m.G}), then the inner linear program returns no solution. In this case, we set the value of the inner minimization problem to $-\infty$.
	The smallest forecast probability $p_L$ is computed similarly:
	\be
	p_L  = \inf_\phi \left( \sup_{ v \in \mathbb{R}^{K+1}}  \left[  0_{1 \times K} \, , \, -1 \right]' v \quad \mbox{s.t.} \quad  A(\phi)  v \leq  b(\phi) \right) \,,  \label{eq:extreme.dual5}
	\ee
	where we set the value of the inner linear program to $+\infty$ if it has no solution.
	
	\subsubsection{Multinomial forecasts}
	
	In the multinomial case, we consider a setting in which $\Theta_0$ is defined as in display (\ref{eq:m.G}) for suitable $G(\phi)$ and the forecast probabilities of each of the outcomes $m = 0,1,\ldots,M$ can be written as 
	\[
	b_m(\phi)'\pi
	\]
	for each $m$.

	For minimax forecasts, the lower probabilities $\underline p_m$ from (\ref{eq:p.lower.multinomial}) are computed analogously to $p_L$, replacing $b(\phi)$ in (\ref{eq:extreme.dual5}) with $b_m(\phi)$:
	\[
	\underline p_m  = \inf_\phi \left( \sup_{ v \in \mathbb{R}^{K+1}}  \left[  0_{1 \times K} \, , \, -1 \right]' v \quad \mbox{s.t.} \quad  A(\phi)  v \leq  b_m(\phi) \right) \,,  
	\]
	for $m = 0,1,\ldots,M$, where we set the value of the inner linear program to $+\infty$ if it has no solution.
	For minimax regret forecasts, the terms $\Delta p_m$ from (\ref{eq:dp.multinomial}) can be computed analogously to (\ref{eq:extreme.dual3}). To do so, first note that for each $m' = 0,1,\ldots,M$ we can compute
	\[
	\sup_{\theta \in \Theta_0} \left( {\textstyle \PR_\theta} ( Y = m' ) - {\textstyle \PR_\theta} ( Y = m ) \right)
	\]
	by replacing the term $b(\phi)$ in (\ref{eq:extreme.dual3}) with $b_{m'}(\phi) - b_m(\phi)$. The value $ \Delta p_m$ is then the maximum over all such $m'$:
	\[
	\Delta p_m  = \max_{m'}\sup_\phi \left(  \inf_{ v \in \mathbb{R}^{K+1}}  \left[  0_{1 \times K} \, , \, 1 \right] v \quad \mbox{s.t.} \quad  A(\phi)  v \leq( b_m(\phi) - b_{m'}(\phi) )\right)  \,,
	\]
	where we again set the value of the inner linear program to $-\infty$ if it has no solution.

	\subsection{Computing Extreme Probabilities: the Continuous Case}
	\label{subsec:binary.computation.continuous}

	\subsubsection{Binary forecasts}

	We first consider a class of problems where the forecast probabilities and restrictions that define $\Theta_0$ are linear in $\Pi$, where $\Pi$ is a probability measure on $(X,\mathcal X)$ where $\mathcal X$ denotes the Borel $\sigma$-field on $X$. We  restrict $\Pi$ to have density with respect to some $\sigma$-finite dominating measure $\nu$ (e.g. Lebesgue measure) and identify each $\Pi$ with its density $\pi$ with respect to $\nu$.\footnote{This nests the previous discrete case by taking $X$ to be the set of $L$ points of discrete support for $\Pi$ and $\nu$ to be counting measure.} We consider a setting where forecast probabilities can be written analogously to (\ref{eq:m.p}) as
	\[
	\int b(x;\phi) \pi(x) \, \mathrm d \nu(x) \,,
	\]
	where $b(\,\cdot\,;\phi) : X \to \mathbb{R}$ is a bounded function for each $\phi$. We first consider a class of problems in which the set $\Theta_0$ is defined via a moment restriction similar to (\ref{eq:m.G}), namely
	\[
	\int g(x;\phi)\pi(x)\,  \mathrm d \nu(x) = r \,,
	\]
	where $g(\,\cdot\,;\phi) : X \to \mathbb{R}^K$ is a vector of moment functions. 
	
	The semiparametric panel data model (Example 1) is of this form, where we now relax the assumption of discrete support for $(\lambda,y_0)$ and allow the joint distribution $\Pi_{\lambda,y}$ to be an arbitrary distribution on $\mathbb R \times \{0,1\}$. The dominating measure $\nu$ is the product of Lebesgue measure on $\mathbb R$ and counting measure on $\{0,1\}$.

	Let $\Pi_\phi$ denote the set of all densities $\pi$ with respect to $\nu$, for which $\int g(x;\phi)\pi(x)\,  \mathrm d \nu(x) $ is finite and $\int \pi(x) \, \mathrm d \nu(x) = 1$. We then have
	\begin{equation} \label{eq:id.continuous}
	\Theta_0 = \left\{(\phi,\pi) : \pi \in \Pi_\phi\,, \int g(x;\phi)\pi(x)\,  \mathrm d \nu(x) = r \right\} \,.
	\end{equation}
	
	In this setting, we can write $p_U$ as 
	\[
	p_U = \sup_\phi \left( \sup_{\pi \in \Pi_\phi} \int b(x;\phi) \pi(x) \, \mathrm d \nu(x) \quad \mbox{s.t.} \quad \int g(x;\phi)\pi(x)\,  \mathrm d \nu(x) = r \right) \,.
	\]
	The inner optimization over $\pi$ has a dual program. Although this dual formulation does not simplify computation a great deal, it can be approximated by a more tractable, finite-dimensional convex program. In what follows, let $\mathrm{ri}(A)$ denote the relative interior of a set $A$. The following proposition collects results from \cite{CsiszarMatus2012} (for the dual formulation) and \cite{ChristensenConnault2019} (for the approximation by a finite-dimensional convex program).
	
	\begin{proposition} \label{prop:duality.binary.continuous}
		If 
		\[
		r \in \mathrm{ri} \left( \left\{ \int g(x;\phi) \pi(x) \, \mathrm d \nu(x) : \pi \in \Pi_\phi \right\}  \right) \,,
		\]
		then the program
		\[
		p^{\phantom *}_\phi = \sup_{\pi \in \Pi_\phi} \int b(x;\phi) \pi(x) \, \mathrm d \nu(x) \quad \mbox{s.t.} \quad \int g(x;\phi)\pi(x)\,  \mathrm d \nu(x) = r 
		\]
		has an equivalent dual formulation
		\[
		p^*_\phi = \inf_{\mu : \nu\textnormal{-}\mathrm{ess} \sup_{x}(b(x;\phi)+\mu'(g(x;\phi)-r)) < +\infty} \bigg( \nu\textnormal{-}\mathrm{ess}  \sup_{x} \left( b(x;\phi) + \mu'(g(x;\phi)-r) \right) \bigg) \,.
		\]
		In addition, if $\Pi^*$ has a strictly positive density $\pi \in \Pi_\phi$ and $\mathbb E^{\Pi_*}\left[ e^{c \|g(X;\phi)\|} \right]$ is finite for each $c \geq 0$, then
		\[
		p^*_\phi = \lim_{\delta \to \infty} \left( \sup_{\eta \geq 0, \mu} -\eta \log \mathbb E^{\Pi_*} \left[ e^{-\eta^{-1}\left(  b(X;\phi) + \mu'(g(X;\phi)-r) \right)} \right]  - \eta \delta \right) \,,
		\]
		where $\mathbb{E}^{\Pi_*}[\,\cdot\,]$ denotes expectation is taken under the distribution $\Pi_*$.
	\end{proposition}
	
	\noindent
	In view of Proposition \ref{prop:duality.binary.continuous}, we may compute $p_U$ using the approximation
	\[
	p_U \approx \sup_\beta \left( \inf_{\eta \geq 0, \mu} \eta \log \mathbb E^{\Pi_*} \left[ e^{\eta^{-1}\left(  b(X;\phi) + \mu'(g(X;\phi)-r) \right)} \right] \right) + \eta \delta  \,,
	\]
	which is valid for large $\delta$. The lower probability $p_L$ can be computed analogously:
	\be
	p_L \approx \inf_\beta \left( \sup_{\eta \geq 0, \mu} -\eta \log \mathbb E^{\Pi_*} \left[ e^{-\eta^{-1}\left(  b(X;\phi) + \mu'(g(X;\phi)-r) \right)} \right] \right) - \eta \delta  \,. \label{eq:extreme.binary.continuous.lower}
	\ee

	Similar techniques may also be used when $\Theta_0$ arises out of robustness concerns; see Example 2. To that end, we can consider a class of models where forecast probabilities and restrictions defining $\Theta_0$ are linear in $\Pi$, but where we now restrict $\Pi$ to the class 
	\[
	\Pi_{\phi,\delta} = \{ \Pi : K(\Pi \| \Pi_\phi) \leq \delta\} \,,
	\]
	where $\delta \geq 0$ and $K(\Pi\|\Pi_\phi)$ is the Kullback--Leibler divergence between $\Pi$ and a reference density $\Pi_\phi$. In the context of Example 2, $\Pi_\phi$ is a correlated random effects distribution indexed by auxiliary parameters $\xi$, and $\phi = (\beta,\xi)$. The identified set is now
	\begin{equation} \label{eq:id.continuous.robust}
	\Theta_0 = \left\{(\phi,\Pi) : \Pi \in \Pi_{\phi,\delta}\,, \int g(x;\phi)\,  \mathrm d \Pi(x) = r \right\} \,.
	\end{equation}
	With this notion of the identified set, we may apply well known duality methods to compute the extreme probabilities using the dual representations
	\begin{align}
	p_U & = \sup_\phi \bigg( \inf_{\eta \geq 0, \mu} \eta \log \mathbb E^{\Pi_\phi} \left[ e^{\eta^{-1}\left(  b(X;\phi) + \mu'(g(X;\phi)-r) \right)} \right] \bigg) + \eta \delta \,,  \notag \\
	p_L & =  \inf_\phi \bigg( \sup_{\eta \geq 0, \mu} -\eta \log \mathbb E^{\Pi_\phi} \left[ e^{-\eta^{-1}\left(  b(X;\phi) + \mu'(g(X;\phi)-r) \right)} \right] \bigg) - \eta \delta  \,, \label{eq:extreme.binary.continuous.lower.2}
	\end{align}
	which are valid whenever $\mathbb E^{\Pi_\phi}\left[ e^{c \|g(X;\phi)\|} \right]$ is finite for each $c \geq 0$ and each $\phi$, and 
	\[
	r \in \mathrm{ri} \left( \left\{ \int g(x;\phi) \, \mathrm d \Pi(x) : \Pi \in \Pi_{\phi,\delta} \right\}  \right) \,;
	\]
	see, e.g., \cite{ChristensenConnault2019} for a formal statement. Similar dual representations apply for neighborhoods constrained by other $\phi$-divergences.
	
	\subsubsection{Multinomial forecasts}
	
	Multinomial forecasts can be implemented similarly using the reformulations described above. For minimax forecasts, if the forecast probabilities are each of the form
	\[
	\int b_m(x,\phi)\, \mathrm d \Pi(x)
	\]
	for $m = 0,1,\ldots,M$, then each $\underline p_m$ can be computed as in (\ref{eq:extreme.binary.continuous.lower}) or (\ref{eq:extreme.binary.continuous.lower.2}), replacing $b$ with $b_m$. For minimax regret forecasts, each $\Delta p_m$ can be computed as 
	\[
	\Delta p_m \approx \max_{m'} \sup_\beta \left( \inf_{\eta \geq 0, \mu} \eta \log \mathbb E^{\Pi_*} \left[ e^{\eta^{-1}\left(  b_{m'}(X;\phi) - b_m(X;\phi) + \mu'(g(X;\phi)-r) \right)} \right] \right) + \eta \delta  \,,
	\]
	when $\Theta_0$ is of the form (\ref{eq:id.continuous}). A similar computation applies when $\Theta_0$ is of the form (\ref{eq:id.continuous.robust}), replacing $\Pi_*$ with $\Pi_\phi$.
	
	\section{Further Results on Robust Binary Forecasts}
	\label{appsec:binary}

	\subsection{Equivalence of Minimax forecasts under Quadratic and Logarithmic Loss}
	
	Here we show that the minimax forecast under quadratic loss is also minimax under logarithmic loss. We first rule out a few pathological cases. Suppose the econometrician chooses $d = 0$. If $p_U > 0$ then the maximum risk is $+\infty$, which is obtained by the maximizing agent choosing any $\theta \in \Theta_0$ with $\PR_\theta(Y = 1)>0$. Thus, it is only optimal to choose $d = 0$ when $p_U = 0$, in which case $\PR_\theta(Y = 1) = 0$ for all $\theta \in \Theta_0$. A parallel argument shows it is only optimal to choose $d = 1$ when $p_L  = 1$. More generally, if $p_L = p_U$ then it is optimal to choose $d$ to be their common value. 
	Now suppose that $p_L < p_U$. Problem (\ref{eq:minimax}) becomes
	\begin{align*}
	\inf_{ d \in \mathcal{D}} \; \sup_{\theta \in \Theta_0} \; \mathbb{E}_\theta [ \ell_p(Y,  d)  ] 
	& = \inf_{d \in [0,1]} \sup_{p \in [p_L, p_U]}  - p \log  d - (1-p) \log (1- d) \\
	& = \sup_{p \in [p_L, p_U]} \inf_{ d \in [0,1]}  - p \log  d - (1-p) \log (1- d) \,,
	\end{align*}
	where the first equality is because for any $ d \in [0,1]$, the maximum risk is obtained at either $p_L$ or $p_U$, and the second equality is by the minimax theorem.  The inner minimum is achieved at $d = p$, and the outer maximum is achieved by taking $p \in [p_L, p_U]$ to be as close to $\frac{1}{2}$ as possible.

	\subsection{Equivalence of Robust Binary Forecasts under Classification Loss}

	We now show that the minimax and minimax regret forecasts are identical under classification loss. First suppose $p_L > \frac{1}{2} $. In this case, the $\theta$-optimal decision is $d_{b,\theta}^* = 1$ for all $\theta \in \Theta_0$ and so $d_{b,mmr} = d_{b,mm} = 1$. Similarly, when $p_U < \frac{1}{2} $ the $\theta$-optimal decision is $d_{b,\theta}^*$ for all $\theta \in \Theta_0$ and so $d_{b,mmr} = d_{b,mm} = 0$. It remains to consider the case in which $p_L \leq \frac{1}{2}$ and $p_U \geq \frac{1}{2}$ both hold. It is then straightforward to deduce that
	\begin{align*}
	d_{b,mmr}
	& = \mathbb{I} [ {\textstyle \frac{1}{2}} - p_L  \leq  p_U - {\textstyle \frac{1}{2}}  ]  = \mathbb{I} [ 1 \leq  p_L + p_U ] = d_{b,mm} \,.
	\end{align*}

	\subsection{Non-equivalence of Minimax and Minimax Regret Forecasts when $M \geq 2$}\label{subsec:binary.multivariate.nonequal}
	
	Unlike the binary case ($M = 1$), minimax and minimax regret forecasts are no longer equal for classification loss when $M \geq 2$. 
	To see this, consider an example with $M = 3$ in which $\Theta_0 = \{\theta_1, \theta_2, \theta_3\}$ with $\theta_1 = (\frac{1}{2}, \frac{1}{2}, 0, 0)'$, $\theta_2 = (\frac{1}{3}, \frac{1}{3}, 0, \frac{1}{3})'$, and $\theta_3 = (\frac{1}{5}, \frac{1}{5}, 0, \frac{4}{5})'$, where we identify each parameter with its vector of forecast probabilities for the outcomes in the set $\mathcal D = \{0,1,2,3\}$. 
	The $\theta$-optimal forecasts for classification loss are $d_{c,\theta_1}^* \in \{0,1\}$ (i.e., both $d_{\theta_1,c}^* = 0$ and $d_{c,\theta_1}^* = 1$ are $\theta$-optimal forecasts under $\theta_1$), $d_{c,\theta_2}^* \in \{0,1,3\}$, and $d_{c,\theta_3}^* = 3$. 
	
	For the minimax decision, we have $\underline p_0 = \frac{1}{5}$, $\underline p_1 = \frac{1}{5}$, $\underline p_2 = 0$, and $\underline p_3  = 0$. Therefore, $d_{c,mm}\in \{0,1\}$ is the minimax decision for classification loss and the minimax risk is ${\cal R}^*_{c,mm} = \frac{4}{5}$.
	
	For the minimax regret decision, note that the regret from choosing $m = 0,1,2,3$ under $\theta_1$ is $(0,0,\frac{1}{2},\frac{1}{2})$. Similarly, under $\theta_2$ and $\theta_3$ the regrets are $(0,0,\frac{1}{3},0)$ and $(\frac{3}{5},\frac{3}{5},\frac{4}{5},0)$. Therefore, $\Delta p_0 = \frac{3}{5}$, $\Delta p_1 = \frac{3}{5}$, $\Delta p_2 = \frac{4}{5}$, and $\Delta p_3 = \frac{1}{2}$. The minimax regret forecast is $d_{c,mmr} = 3$ and its maximum regret is $\mathcal R^*_{c,mmr} = \frac{1}{2}$.
	
	Similarly, with $M = 2$ and $\theta_1 = (\frac{1}{2}, \frac{1}{2}, 0)'$, $\theta_2 = (\frac{1}{3}, \frac{1}{3}, \frac{1}{3})'$, and $\theta_3 = (\frac{1}{5}, \frac{1}{5}, \frac{4}{5})'$, we have that the minimax forecast is $d_{c,mm} \in \{0,1\}$ whereas the minimax regret forecast is $d_{c,mmr} = 2$.

	\section{Proofs}\label{appsec:proofs}
	
	\subsection{Preliminaries}
	
	Our approach to establishing asymptotic efficiency follows \cite{HiranoPorter2009}. First, we characterize the asymptotic representation of the forecast in the limit experiment. Second, we show that these are optimal with respect to average excess maximum risk and regret in the limit experiment. Finally, we invoke a version of their Lemma 1 which allows us to approximate average excess maximum risk or regret with finite $n$ by that in the limit experiment.  The next two subsections describe preliminary results for steps 1 and 2 of this approach for binary and multinomial forecasts. The final subsection presents proofs of the main result.
	
	To simplify notation, throughout the proofs we write $\Pi_n(P)$ for the posterior $\Pi_n(P|X_n)$, $\mathrm d \Pi_n$ in place of $\mathrm d \Pi_n(P|X_n)$. We adopt the convention that $+\infty \times 0 = 0$.  We also require a limiting counterpart to excess maximum risk and regret criteria. To this end, for any sequence $\{d_n\}_{n \geq 1} \in \mathbb D$, $P_0 \in \mathcal P$, and perturbation direction $h \in \mathbb R^k$, we define local asymptotic excess maximum risk as
	\[
	\mathcal L_{mm}(\{d_n\}_{n \geq 1};P_0,h) = \lim_{n \to \infty} \mathbb E_{P_{n,h}}\left[ \sqrt n \Delta\mathcal R_{mm}\left(d_n, P_0 + n^{-1/2} h; X_n\right) \right] \,.
	\]
	Local asymptotic excess maximum regret $\mathcal L_{mmr}(\{d_n\}_{n \geq 1};P_0,h)$ is defined similarly, replacing excess maximum risk $\Delta\mathcal R_{mm}$ in the above display with excess maximum regret $\Delta\mathcal R_{mmr}$. The local asymptotic excess maximum risk and regret of $\{d_n\}_{n \geq 1} \in \mathbb D$ will only depend on $\{d_n\}_{n \geq 1}$ through its asymptotic representation $d^\infty$. Note the form of $d^\infty$ may depend on $P_0$, but we suppress this dependence to simplify notation.
  We can therefore write
	\begin{align*}
	\mathcal L_{mm}(\{d_n\}_{n \geq 1}; P_0, h) & = \mathcal L_{mm}^\infty(d^\infty;P_0,h) \,, \\
	\mathcal L_{mmr}(\{d_n\}_{n \geq 1}; P_0, h) & = \mathcal L_{mmr}^\infty(d^\infty;P_0,h)
	\end{align*}
	for some functionals $\mathcal L_{mm}^\infty$ and $\mathcal L_{mmr}^\infty$. We say that $d^\infty_*$ is \emph{optimal for average local asymptotic excess maximum risk} in the limit experiment if it is a flat prior Bayes rule:
	\[
	\int \mathcal L_{mm}^\infty(d^\infty_*;P_0,h) \, \mathrm d h = \inf_{d^\infty} \int \mathcal L_{mm}^\infty(d^\infty;P_0,h) \, \mathrm d h \,,
	\]
	where the infimum on the right-hand side is taken over all such (possibly randomized) $\mathcal D$-valued forecasts $d^\infty(Z,U)$ in the limit experiment. Optimality for average local asymptotic excess maximum regret in the limit experiment is defined similarly. We sometimes simply say \emph{optimal in the limit experiment} when the notion of optimality (local asymptotic minimax risk or regret) is obvious form the context.

	\subsection{Supplementary Lemmas: Binary Forecasts}\label{appsec:binary.proofs}
	
	For binary forecasts, both $\Delta\mathcal R_{mm}(d, P)$ and $\Delta\mathcal R_{mmr}(d, P)$ can be written as linear functions of $d$. Therefore, local asymptotic excess maximum risk and regret depend on $\{d_n\}_{n \geq 1} \in \mathbb D$ only through $\lim_{n \to \infty} \mathbb E_{P_{n,h}}[d_n(X_n)]$ which takes the form $\mathbb E_h[d^\infty(Z)]$ \cite[Theorem 15.1]{Vandervaart2000} where $\mathbb E_h$ denotes expectation with respect to $Z \sim N(h,I_0^{-1})$. That is not to say that the asymptotically optimal forecast cannot be randomized. Rather, $d^\infty(z)$ represents the average (with respect to the randomization) probability that $d^\infty(Z,U) = 1$ when $Z = z$. 
		
	 To simplify notation, let $p_{U0} := p_U(P_0)$ and $p_{L0} := p_L(P_0)$. There are three cases to consider for the next lemma: Case 1, $a_{01} p_{L0} + a_{10} p_{U0} > a_{01}$; Case 2, $a_{01} p_{L0} + a_{10} p_{U0} < a_{01}$; and Case 3, $a_{01} p_{L0} + a_{10} p_{U0} = a_{01}$. 
	
	\begin{lemma}\label{lem:asymptotic.binary.risk}
		Let Assumptions \ref{A1} and parts (a) of Assumption \ref{a:binary} hold. Then: 
		\begin{enumerate}[nosep]
			\item[(i)] $d_{b,mm}$ has the asymptotic representation
			\[
			d^\infty_{b,mm}(Z) = 
			\left[ \begin{array}{ll}
			1 & \mbox{in case 1,} \\
			0 & \mbox{in case 2,} \\
			\mathbb{I} \left[ \mathbb E^*[ \dot f_{P_0}[Z^* + Z]  |Z ] \geq 0 \right] & \mbox{in case 3,}
			\end{array} \right.
			\]
			where $f(P) = a_{01} p_L(P) + a_{10}p_U(P)$;
			\item[(ii)] local asymptotic excess maximum risk of $\{d_n\}_{n \geq 1} \in \mathbb D$ is
			\begin{align*}
			\mathcal L_{mm}^\infty(d^\infty ; P_0, h) 
			& = 
			\left[ \begin{array}{ll}
			+\infty \times (1 - \mathbb E_h[ d^\infty(Z) ])  & \mbox{in case 1,} \\
			+\infty \times \mathbb E_h[ d^\infty(Z) ]\phantom{1-()}  & \mbox{in case 2,} \\
			(\dot f_{P_0}[h])_{+} - \mathbb E_h[ d^\infty(Z) ] (\dot f_{P_0}[h]) & \mbox{in case 3,}
			\end{array} \right.
			\end{align*}
			where $\lim_{n \to \infty} \mathbb E_{P_{n,h}}[d_n(X_n)] = \mathbb E_h[d^\infty(Z)]$;
			\item[(iii)] $d_{b,mm}^\infty(Z)$ is optimal in the limit experiment.
		\end{enumerate} 
	\end{lemma}

	\begin{proof}[Proof of Lemma \ref{lem:asymptotic.binary.risk}]
		\underline{Part (i):} As $d_{b,mm}(X_n)$ is discrete, establishing convergence in distribution of $d_{b,mm}(X_n)$ under $\{F_{n,P_{n,h}}\}_{n \geq 1}$ is equivalent to characterizing $\lim_{n \to \infty} F_{n,P_{n,h}} (d_{b,mm}(X_n) = 1)$.  
		
		For Case 1, by Assumption \ref{a:binary}.1(a), for any $\varepsilon > 0$ there is a neighborhood $N$ of $P_0$ upon which $|a_{01} p_L(P) + a_{10} p_U(P) - a_{01} p_{L0} - a_{10} p_{U0}| < \varepsilon$. By posterior consistency (Assumption \ref{a:binary}.2(a)) and the fact that $0 \leq p_U,p_L \leq 1$, we have 
		\[
		\left| \int \left( a_{01} p_L(P) + a_{10} p_U(P) \right) \, \mathrm d \Pi_n - a_{01} p_{L0} - a_{10} p_{U0} \right| \leq \varepsilon  \Pi_n(P \in N) + 2(a_{01} + a_{10}) \Pi_n(P \not \in N) \overset{P_0}{\to} \varepsilon \,.
		\]
		As $\varepsilon $ was arbitrary, $ \int ( a_{01} p_L(P) + a_{10} p_U(P) ) \, \mathrm d \Pi_n \overset{P_0}{\to} a_{01} p_{L0} + a_{10} p_{U0}$. Therefore, $d_{b,mm}(X_n) \overset{P_0}{\to} 1$. As $\{F_{n,P_0}\}_{n \geq 1}$ and $\{F_{n,P_{n,h}}\}_{n \geq 1}$ are contiguous by Le Cam's first lemma and Assumption \ref{A1}, it follows that $d_{b,mm}(X_n) \overset{P_{n,h}}{\to} 1$ for any $h \in \mathbb R^k$. Case 2 follows similarly.
		For Case 3, we may write
		\[
		d_{b,mm}(X_n) = \mathbb{I} \left[ \int a_{01} \sqrt n (p_L(P) - p_{L0}) + a_{10} \sqrt n( p_U(P) - p_{U0}) \, \mathrm d \Pi_n \geq 0 \right] \,.
		\]
		By Assumption \ref{a:binary}.3(a) with $A = \{(x,y) : a_{01} x + a_{10} y \geq 0\}$, we have 
		\[
		\lim_{n \to \infty} F_{n,P_{n,h}} (d_{b,mm}(X_n) = 1) = \mathbb P_h\left( \mathbb E^*[ a_{01} \dot p_{L,P_0}[Z^* + Z] + a_{10} \dot p_{U,P_0}[Z^* + Z]  |Z ] \geq 0 \right) \,.
		\]
		
		\underline{Part (ii):} The excess maximum risk of $d \in \mathcal D$ is
		\[
		\Delta\mathcal R_{mm}(d, P) = d (a_{01} - a_{01} p_L(P) - a_{10} p_U(P)) - (a_{01} - a_{01} p_L(P) - a_{10} p_U(P))_-  \,,
		\]
		where $a_- = \min\{a,0\}$. For Case 1, for all $n$ large enough we have 
		\[
		\mathbb E_{P_{n,h}}[ \sqrt n \Delta \mathcal R_{mm}(d_n(X_n), P_{n,h}) ] = \sqrt n \times ( a_{01} p_L(P_{n,h}) + a_{10} p_U(P_{n,h}) - a_{01}) \times \left( 1 - \mathbb E_{P_{n,h}}[d_n(X_n)] \right) \,,
		\]
		where $\liminf_{n \to \infty} (a_{01} p_L(P_{n,h}) + a_{10} p_U(P_{n,h}) - a_{01}) > 0$ and $\mathbb E_{P_{n,h}}[d_n(X_n)] \to \mathbb E_h[d^\infty(Z)]$. Case 2 follows by similar arguments. For Case 3, rearranging slightly we have
		\begin{align*}
		\mathbb E_{P_{n,h}}[ \sqrt n \Delta \mathcal R_{mm}(d_n(X_n), P_{n,h}) ] 
		& =  \sqrt n (a_{01} p_L(P_{n,h}) + a_{10} p_U(P_{n,h}) - a_{01} p_{L0} - a_{10} p_{U0})_+ \\
		& \quad - \mathbb E_{P_{n,h}}[d_n(X_n)] \times \sqrt n (a_{01} p_L(P_{n,h}) + a_{10} p_U(P_{n,h}) - a_{01} p_{L0} - a_{10} p_{U0}) \,,
		\end{align*}
		where $\mathbb E_{P_{n,h}}[d_n(X_n)] \to \mathbb E_h[d^\infty(Z)]$ and 
		\begin{align*}
		\sqrt n (a_{01} p_L(P_{n,h}) + a_{10} p_U(P_{n,h}) - a_{01} p_{L0} - a_{10} p_{U0})_{\phantom +} & \to \phantom{(} a_{01} \dot p_{L,P_0}[h] + a_{10} \dot p_{U,P_0}[h] \,, \\
		\sqrt n (a_{01} p_L(P_{n,h}) + a_{10} p_U(P_{n,h}) - a_{01} p_{L0} - a_{10} p_{U0})_+ & \to ( a_{01} \dot p_{L,P_0}[h] + a_{10} \dot p_{U,P_0}[h])_{+}
		\end{align*} by Assumption \ref{a:binary}.1(a). 
		
		\underline{Part (iii):} From part (ii), we see that $d^\infty(Z) = 1$ (almost everywhere) is optimal in Case 1 and $d^\infty(Z) = 0$ (almost everywhere) is optimal in Case 2. In Case 3, we have
		\begin{align*}
		& \int (a_{01} \dot p_{L,P_0}[h] + a_{10} \dot p_{U,P_0}[h])_{+} - \mathbb E_h[ d^\infty(Z) ] \times (a_{01} \dot p_{L,P_0}[h] + a_{10} \dot p_{U,P_0}[h]) \, \mathrm d h \\
		& \propto \int \int \left( (a_{01} \dot p_{L,P_0}[h] + a_{01} \dot p_{U,P_0}[h])_{+} - d^\infty(z)  \times (a_{01} \dot p_{L,P_0}[h] + a_{10} \dot p_{U,P_0}[h])  \right)e^{-\frac{1}{2} (z - h)' I_0 (z - h) } \, \mathrm d z \, \mathrm d h  \,.
		\end{align*}
		Swapping the order of integration and minimizing pointwise in $z$, we obtain
		\[
		d^\infty(z) = \mathbb I \left[  \int  (a_{10} \dot p_{L,P_0}[h] + a_{01} \dot p_{U,P_0}[h]) e^{-\frac{1}{2} (z - h)' I_0 (z - h) } \, \mathrm d h  \geq 0 \right]  \,.
		\]
		Equivalently, $d^\infty(z) = \mathbb{I} \left[ \mathbb E^*[ a_{10} \dot p_{L,P_0}[Z^* + Z] + a_{01} \dot p_{U,P_0}[Z^* + Z]  |Z = z] \geq 0 \right]$. This is the same asymptotic representation as was derived in Part (i).
	\end{proof}

	Let $a = \frac{a_{01}}{a_{01} + a_{10}}$ There are four cases to consider for the next lemma, namely: Case 1, $p_{L0} + p_{U0} > 2a$; Case 2, $p_{L0} + p_{U0} < 2a$; Case 3, $p_{L0} + p_{U0} = 2a$ and $p_{U0} >  a$; and Case 4, $p_{L0} = p_{U0} = a$.

	\begin{lemma}\label{lem:asymptotic.binary.regret}
		Let Assumption \ref{A1} and \ref{a:binary} hold. Then: 
		\begin{enumerate}[nosep]
			\item[(i)] $d_{b,mmr}$ has the asymptotic representation
			\begin{align*}
			d^\infty_{b,mmr}(Z) 
			& = 
			{ \left[ \begin{array}{l}
				1 \quad \mbox{in case 1,} \\
				0 \quad \mbox{in case 2,} \\
				\mathbb{I} \left[ \mathbb E^*[ \dot p_{L,P_0}[Z^* + Z] + \dot p_{U,P_0}[Z^* + Z] |Z ] \geq 0 \right] \phantom{()_+()_-} \quad \mbox{in case 3,} \\
				\mathbb{I} \left[ \mathbb E^*[ (\dot p_{L,P_0}[Z^* + Z])_- + (\dot p_{U,P_0}[Z^* + Z])_+|Z]  \geq 0 \right] \quad \mbox{in case 4;} \\
				\end{array} \right.
			}
			\end{align*}
			\item[(ii)] local asymptotic excess maximum regret of $\{d_n\}_{n \geq 1} \in \mathbb D$ is
			\begin{align*}
			& \mathcal L_{mmr}^\infty( d^\infty ; P_0, h) \\
			& = {\small
				\left[ \begin{array}{l}
				+\infty \times (1 - \mathbb E_h[ d^\infty(Z) ])  \quad \mbox{in case 1,} \\
				+\infty \times \mathbb E_h[ d^\infty(Z) ]\phantom{1-()}  \quad \mbox{in case 2,} \\
				(a_{01} + a_{10})\left( (\dot p_{L,P_0}[h] + \dot p_{U,P_0}[h])_+ -\mathbb E_h[ d^\infty(Z) ] ( \dot p_{L,P_0}[h] + \dot p_{U,P_0}[h]) \right) \phantom{()_+()_-()_+()_-} \quad \mbox{in case 3,} \\
				(a_{01} + a_{10})\left( ((\dot p_{L,P_0}[h])_- + (\dot p_{U,P_0}[h])_+ )_+ - \mathbb E_h[ d^\infty(Z) ] ( (\dot p_{L,P_0}[h])_- + (\dot p_{U,P_0}[h])_+)  \right) \quad \mbox{in case 4,}
				\end{array} \right.
			}
			\end{align*}
			where $\lim_{n \to \infty} \mathbb E_{P_{n,h}}[d_n(X_n)] = \mathbb E_h[d^\infty(Z)]$;
			\item[(iii)] $d_{b,mmr}^\infty(Z)$ is optimal in the limit experiment.
		\end{enumerate} 
	\end{lemma}

	\begin{proof}[Proof of Lemma \ref{lem:asymptotic.binary.regret}]
		\underline{Part (i):} Cases 1 and 2 follow by similar arguments to the Proof of Lemma \ref{lem:asymptotic.binary.risk}(i). For Case 3, let $\kappa := p_{U0} - a = a - p_{L0}$ and note $\kappa > 0$. We have
		\begin{align*}
		F_{n,P_{n,h}} (d_{b,mmr}(X_n) = 1) 
		& = F_{n,P_{n,h}} \left( {\textstyle \int  \left( a - p_L(P)   \right)_+ \! \mathrm d \Pi_n \leq \int  \left( p_U(P) - a \right)_+ \! \mathrm d \Pi_n } \right) \\
		& \geq F_{n,P_{n,h}} \left( {\textstyle \int  \left( a - p_L(P)   \right)_+ \! \mathrm d \Pi_n \leq \int  \left( p_U(P)  - a \right)  \mathrm d \Pi_n } \right) \\
		& = F_{n,P_{n,h}} \left( {\textstyle \int  \left( \kappa - ( p_L(P) - p_{L0})   \right)_+ \! \mathrm d \Pi_n \leq \int  \left( p_U(P) - p_{U0} \right)  \mathrm d \Pi_n }  + \kappa  \right) \,.
		\end{align*}
		As $(x - y)_{+} - x = \max(-y, -x)$ and hence $(x - y)_+ + y - x = \max(0, y-x)$, we can rewrite the preceding inequality as
		\begin{align}
		& F_{n,P_{n,h}} (d_{b,mmr}(X_n) = 1) \notag \\
		& \geq F_{n,P_{n,h}} \left( {\textstyle \int  \left( ( p_L(P) - p_{L0})  - \kappa \right)_+ \! \mathrm d \Pi_n \leq \int  \left( p_L(P) + p_U(P) - p_{L0} - p_{U0} \right)  \mathrm d \Pi_n } \right) \,. \label{eq:asymptotic.binary.regret.p1}
		\end{align}
		By continuity of $p_L(P)$ at $P_0$ (by Assumption \ref{a:binary}.1(a)) and posterior consistency, we can choose a neighborhood $N_\kappa$ of $P_0$ upon which $|p_L(P) - p_L(P_0)| < \kappa$. By Assumption \ref{a:binary}.2(b), there exists $\gamma > \frac{1}{2}$ such that $n^\gamma \Pi_n(P \not \in N_\kappa) \overset{P_0}{\to} 0$. As $0 \leq p_L \leq 1$, we therefore have the bound 
		\[
		n^\gamma \int  \left( ( p_L(P) - p_L(P_0))  - \kappa \right)_+ \! \mathrm d \Pi_n \leq 2 n^\gamma \Pi_n(P \not \in N_\kappa) \overset{P_0}{\to} 0 \,.
		\]
		By contiguity, convergence holds under $P_{n,h}$ for all $h \in \mathbb R_k$. We therefore have that 
		\[
		F_{n,P_{n,h}} \left( {\textstyle \int  \left( \sqrt n( p_L(P) - p_L(P_0))  - \sqrt n \kappa \right)_+ \! \mathrm d \Pi_n } \leq n^{\gamma - \frac{1}{2}} \right) \to 1
		\]
		for all $h \in \mathbb R^k$. We may therefore rewrite (\ref{eq:asymptotic.binary.regret.p1}) as 
		\begin{align*}
		F_{n,P_{n,h}} (d_{b,mmr}(X_n) = 1)
		& \geq F_{n,P_{n,h}} \left( {\textstyle n^{\gamma - \frac{1}{2}} \leq \int  \left( p_L(P) + p_U(P) - p_{L0} - p_{U0} \right)  \mathrm d \Pi_n } \right) - o(1) \\
		& \geq F_{n,P_{n,h}} \left( {\textstyle \varepsilon \leq \int  \left( p_L(P) + p_U(P) - p_{L0} - p_{U0} \right)  \mathrm d \Pi_n } \right) - o(1) \\
		& \to \mathbb P_h\left( \mathbb E^*[ \dot p_{L,P_0}[Z^* + Z] + \dot p_{U,P_0}[Z^* + Z]  |Z ] \geq \varepsilon \right) 
		\end{align*}
		for any $\varepsilon > 0$, where the final line is by Assumption \ref{a:binary}.3(a) with $A = \{(x,y) : x + y \geq \varepsilon\}$. Similarly,
		\begin{align*}
		F_{n,P_{n,h}} (d_{b,mmr}(X_n) = 1)
		& \leq F_{n,P_{n,h}} \left( {\textstyle -\varepsilon \leq \int  \left( p_L(P) + p_U(P) - p_{L0} - p_{U0} \right)  \mathrm d \Pi_n } \right) + o(1) \\
		& \to \mathbb P_h\left( \mathbb E^*[ \dot p_{L,P_0}[Z^* + Z] + \dot p_{U,P_0}[Z^* + Z]  |Z ] \geq -\varepsilon \right) 
		\end{align*}
		for every $\varepsilon > 0$. The desired convergence now follows by Assumption \ref{a:binary}.1(b). 
		
		Finally, consider Case 4 ($p_{U0} = p_{L0} = a$). We may write
		\[
		d_{b,mmr}(X_n) = \mathbb{I} \left[   0 \leq \int  \left(  p_L(P) - p_{L0}  \right)_- + \left( p_U(P) - p_{U0} \right)_+ \! \mathrm d \Pi_n \right] \,.
		\]
		It follows by Assumption \ref{a:binary}.3(b) taking $A = \{(x,y) : x + y \geq 0\}$ that 
		\[
		\lim_{n \to \infty} F_{n,P_{n,h}} (d_{b,mmr}(X_n) = 1) \to \mathbb P_h\left( \mathbb E^*[ (\dot p_{L,P_0}[Z^* + Z])_- + (\dot p_{U,P_0}[Z^* + Z])_+  |Z ] \geq 0 \right) \,.
		\]
		
		\underline{Part (ii):} First note that excess maximum regret of $d \in \mathcal D$ is
		\begin{align*}
		\Delta\mathcal R_{mmr}(d, P) 
		& = d (a_{01} - (a_{01} + a_{10}) p_L(P))_+ + (1-d) ((a_{01} + a_{10}) p_U(P) - a_{01})_+ \\
		& \quad - (a_{01} - (a_{01} + a_{10}) p_L(P))_+  \wedge  ((a_{01} + a_{10}) p_U(P) - a_{01})_+ \,.
		\end{align*}
		For Case 1, note $((a_{01} + a_{10}) p_U(P_{n,h}) - a_{01})_+ = (a_{01} + a_{10}) p_U(P_{n,h}) - a_{01} > 0$ holds for $n$ sufficiently large because $p_{U0} > a$ in this case. Moreover, in this case $a_{01} - (a_{01} + a_{10}) p_{L0} < (a_{01} + a_{10}) p_{U0} - a_{01}$, so the term $\sqrt n ((a_{01} + a_{10}) p_U(P_{n,h}) - a_{01})_+$ will dominate the term $\sqrt n (a_{01} - (a_{01} + a_{10}) p_L(P_{n,h}))_+$ asymptotically. It follows that for any $\{d_n\}_{n \geq 1} \in \mathbb D$,
		\[
		\lim_{n \to \infty} \mathbb E_{P_{n,h}}[ \sqrt n \Delta \mathcal R_{mmr}(d_n(X_n), P_{n,h}) ] = \lim_{n \to \infty} \sqrt n \times ( p_U(P_{n,h}) - a ) \times \left( 1 - \mathbb E_{P_{n,h}}[d_n(X_n)] \right) \,,
		\]
		where $p_U(P_{n,h}) \to p_{U0} > a$ and $\mathbb E_{P_{n,h}}[d_n(X_n)] \to \mathbb E_h[d^\infty(Z)]$. 
		Case 2 follows similarly.
		
		For Case 3, first note that for $n$ sufficiently large we have 
		\begin{align*}
		(a_{01} - (a_{01} + a_{10}) p_L(P_{n,h}))_+  & = a_{01} - (a_{01} + a_{10}) p_L(P_{n,h}) \,, \\
		((a_{01} + a_{10}) p_U(P_{n,h}) - a_{01})_+  & = (a_{01} + a_{10}) p_U(P_{n,h}) - a_{01} \,.
		\end{align*}
		Letting $\sqrt n (a_{01} - (a_{01} + a_{10})p_{L0}) = \sqrt n ((a_{01} + a_{10}) p_{U0} - a_{01}) = \sqrt n \kappa$ where $\kappa > 0$ and taking $n$ sufficiently large, we  therefore obtain
		\begin{align*}
		& \mathbb E_{P_{n,h}}[ \sqrt n \Delta \mathcal R_{mmr}(d_n(X_n), P_{n,h}) ] \\
		& = \sqrt n \times \mathbb E_{P_{n,h}} [d_n(X_n)] \left(\kappa - (a_{01} + a_{10}) (p_L(P_{n,h}) - p_{L0}) \right) \\
		& \quad  + \sqrt n \times (1-\mathbb E_{P_{n,h}} [d_n(X_n)]) \left(\kappa + (a_{01} + a_{10})  ( p_U(P_{n,h}) - p_{U0}) \right)  \\
		& \quad  - \sqrt n \times \left( \left(\kappa - (a_{01} + a_{10}) (p_L(P_{n,h}) - p_{L0}) \right) \wedge \left(\kappa + (a_{01} + a_{10})( p_U(P_{n,h}) -  p_{U0}) \right)\right) \\
		& = (a_{01} + a_{10})  \Big( \mathbb E_{P_{n,h}} [d_n(X_n)] \times -  \sqrt n (p_L(P_{n,h}) - p_{L0}) + (1-\mathbb E_{P_{n,h}} [d_n(X_n)]) \times \sqrt n ( p_U(P_{n,h}) - p_{U0})   \\
		& \quad  -  \left( \left( - \sqrt n (p_L(P_{n,h}) - p_{L0}) \right) \wedge \left( \sqrt n ( p_U(P_{n,h}) - p_{U0}) \right) \right) \Big) \,,
		\end{align*}
		which converges to
		\[
		(a_{01} + a_{10})  \Big( - \mathbb E_h[d^\infty(Z)] \dot p_{L,P_0}[h] + (1 - \mathbb E_h[d^\infty(Z)]) \dot p_{U,P_0}[h] - \left( (-\dot p_{L,P_0}[h]) \wedge (\dot p_{U,P_0}[h] )\right) \Big) 
		\]
		by Assumption \ref{a:binary}.1(a). The stated form now follows because $x - ((-y) \wedge x) = (x + y)_+$.
		
		Finally, for Case 4 $a_{01} - (a_{01} + a_{10})p_{L0} = (a_{01} + a_{10}) p_{U0} - a_{01} = 0$. By similar logic to Case 3.,
		\begin{align*}
		& \mathbb E_{P_{n,h}}[ \sqrt n \Delta \mathcal R_{mmr}(d_n(X_n), P_{n,h}) ] \\
		& = (a_{01} + a_{10}) \times \sqrt n \times \Big( \mathbb E_{P_{n,h}} [d_n(X_n)]  \left( - (p_L(P_{n,h}) - p_{L0}) \right)_+ \\
		& \quad  + (1-\mathbb E_{P_{n,h}} [d_n(X_n)])  \left(   p_U(P_{n,h}) - p_{U0} \right)_+   - \left( - (p_L(P_{n,h}) - p_{L0}) \right)_+ \wedge \left(   p_U(P_{n,h}) - p_{U0} \right)_+ \Big) \,,
		\end{align*}
		which converges to
		\[
		(a_{01} + a_{10})  \Big( \mathbb E_h[d^\infty(Z)] (-\dot p_{L,P_0}[h])_+ + (1 - \mathbb E_h[d^\infty(Z)]) (\dot p_{U,P_0}[h])_+ - \left( (-\dot p_{L,P_0}[h])_+ \wedge (\dot p_{U,P_0}[h] )_+ \right) \Big) 
		\]
		again by Assumption \ref{a:binary}.1(a). The result follows from $a - (b \wedge a) = (a - b)_+$ and $-(-a)_+ = a_-$.
		
		\underline{Part (iii):} From part (ii), we see that $d^\infty(Z) = 1$ (almost everywhere) is optimal in Case 1 and $d^\infty(Z) =0$ (almost everywhere) is optimal in Case 2. In Case 3, by similar arguments to the proof of Lemma \ref{lem:asymptotic.binary.regret}(iii) we see that average asymptotic excess maximum regret is minimized with
		\[
		d^\infty_{P_0}(z) = \mathbb I \left[  \int  (\dot p_{L,P_0}[h] + \dot p_{U,P_0}[h]) e^{-\frac{1}{2} (z - h)' I_0 (z - h) } \, \mathrm d h  \geq 0 \right]  \,,
		\]
		whereas a minimizing choice in Case 4 is
		\[
		d^\infty_{P_0}(z) = \mathbb I \left[  \int  ((\dot p_{L,P_0}[h])_- + (\dot p_{U,P_0}[h])_+) e^{-\frac{1}{2} (z - h)' I_0 (z - h) } \, \mathrm d h  \geq 0 \right]  \,.
		\]
	\end{proof}

	\subsection{Supplementary Lemmas: Multinomial Forecasts}
	
	For multinomial forecasts, the excess maximum risk $\Delta\mathcal R_{mm}(d, P)$ and regret $\Delta\mathcal R_{mmr}(d, P)$ are linear in the indicator functions $\mathbb I[d = m]$ for $m = 0,\ldots,M$. Therefore, local asymptotic excess maximum risk and regret depend on $\{d_n\}_{n \geq 1} \in \mathbb D$ through $\lim_{n \to \infty} \mathbb E_{P_{n,h}}[\mathbb I(d_n(X_n) = m)]$ which can be written $\mathbb E_h[d^\infty_{m}(Z)]$ for each $m$ \cite[Theorem 15.1]{Vandervaart2000}. The term $d^\infty_m(z)$ represents the average (with respect to randomization) probability that $d^\infty(Z,U) = m$ when $Z = z$.
	
	Deriving the asymptotic representation of $d_{c,mm}(X_n)$ requires a tie-breaking rule, so in the derivation below we take the smallest element of the set of maximizers. To simplify notation, let $\underline p_{m0} := \underline p_m(P_0)$. It is without loss of generality to reorder the indices so that $\underline p_{00} \geq \underline p_{10} \geq \ldots \geq \underline p_{M0}$. There are two cases, namely: Case 1, $\underline p_{00} > \underline p_{10}$; and Case 2, $\underline p_{00} = \underline p_{10} = \ldots = \underline p_{k0}$ for some $k \in \{1, \ldots, M\}$ with $\underline p_{k0} > \underline p_{(k+1)0}$ if $k < M$.

	\begin{lemma}\label{lem:asymptotic.multinomial.risk}
		Let Assumptions \ref{A1}, \ref{a:multinomial}.1(a), \ref{a:multinomial}.2, and \ref{a:multinomial}.3(a) hold. Then: 
		\begin{enumerate}[nosep]
			\item[(i)] $d_{c,mm}$ has the asymptotic representation
			\[
			d^\infty_{c,mm,m}(Z) = 
			\left[ \begin{array}{ll}
			\mbox{$1$ if $m = 0$ and $0$ if $m \in \{1,\ldots,M\}$} & \mbox{in case 1,} \\
			\mathbb I[ ( \mathbb E^*[ \dot{\underline p}_{m,P_0}[Z^* + Z]|Z] > \max_{0 \leq m' \leq m-1}  \mathbb E^*[ \dot{\underline p}_{m',P_0}[Z^* + Z]|Z] ) \mbox{ and } \\
			\quad ( \mathbb E^*[ \dot{\underline p}_{m,P_0}[Z^* + Z]|Z] \geq \max_{m+1 \leq m' \leq k} \mathbb E^*[ \dot{\underline p}_{m',P_0}[Z^* + Z]|Z] )  ] \\
			\mbox{if $m \in \{0,\ldots,k\}$ and $0$ if $m \in \{k+1,\ldots,M\}$}  & \mbox{in case 2,}
			\end{array} \right.
			\]
			where the maximum over an empty index is $-\infty$;
			\item[(ii)] local asymptotic excess maximum risk of $\{d_n\}_{n \geq 1} \in \mathbb D$ is
			\begin{align*}
			\mathcal L_{mm}^\infty(d^\infty ; P_0, h) 
			& = 
			\left[ \begin{array}{ll}
			+\infty \times (1 - \mathbb E_h[ d^\infty_0(Z) ])  & \mbox{in case 1,} \\
			\sum_{m=0}^k \mathbb E_h[ d^\infty_{m}(Z)] ( \max_{0 \leq m' \leq k} \dot{\underline p}_{m',P_0}[h] - \dot{\underline p}_{m,P_0}[h] ) \\
			\quad + \infty \times (1 - \sum_{m = 0}^k \mathbb E_h[ d^\infty_{m}(Z)]) & \mbox{in case 2,}
			\end{array} \right.
			\end{align*}
			where $\lim_{n \to \infty} \mathbb E_{P_{n,h}}[\mathbb I[d_n(X_n) = m]] = \mathbb E_h[d^\infty_{m}(Z)]$;
			\item[(iii)] $(d_{c,mm,m}^\infty(Z))_{m=0}^M$ is optimal in the limit experiment.
		\end{enumerate} 
	\end{lemma}
	
	\begin{proof}[Proof of Lemma \ref{lem:asymptotic.multinomial.risk}]
		\underline{Part (i):} Case 1 follows by similar arguments to the proof of Lemma \ref{lem:asymptotic.binary.risk}. For Case 2, if $k < M$ we can deduce by continuity of the $\underline p_m$ and posterior consistency that $F_{n,P_{n,h}} (d_{c,mm}(X_n) > k) \to 0$ and $F_{n,P_{n,h}} (  \min_{0 \leq m \leq k}\int \underline p_{m}(P) \, \mathrm d \Pi_n  >  \max_{m > k}\int \underline p_{m}(P) \, \mathrm d \Pi_n ) \to 1$. Let $ \min_{0 \leq m \leq k}\int \underline p_{m}(P) \, \mathrm d \Pi_n  >  \max_{m > k}\int \underline p_{m}(P) \, \mathrm d \Pi_n$. For $m \in \{0,1\,\ldots,k\}$, under the above tie-breaking rule we then have
		\begin{align*}
		\mathbb I[d_{c,mm}(X_n) = m] &  =  \mathbb I \left[  \int \underline p_{m}(P) \, \mathrm d \Pi_n > \max_{0 \leq m' \leq m-1} \int \underline p_{m'}(P) \, \mathrm d \Pi_n  \right] \\
		& \quad \times   \mathbb I \left[  \int \underline p_{m}(P) \, \mathrm d \Pi_n \geq \max_{m+1 \leq m' \leq k}\int \underline p_{m'}(P) \, \mathrm d \Pi_n  \right] \,.
		\end{align*}
		As $\underline p_{00} = \underline p_{10} = \ldots = \underline p_{k0}$, we may rewrite the previous expression as
		\begin{align*}
		\mathbb I[d_{c,mm}(X_n) = m] &  =   \mathbb I \left[  \int \sqrt n (\underline p_{m}(P) - \underline p_{m0}) \, \mathrm d \Pi_n > \max_{0 \leq m' \leq m-1} \int \sqrt n ( \underline p_{m'}(P) - \underline p_{m'0}) \, \mathrm d \Pi_n  \right]  \\
		& \quad \times  \mathbb I \left[  \int \sqrt n (\underline p_{m}(P) - \underline p_{m0}) \, \mathrm d \Pi_n \geq \max_{m+1 \leq m' \leq k} \int \sqrt n ( \underline p_{m'}(P) - \underline p_{m'0}) \, \mathrm d \Pi_n  \right] \,.
		\end{align*}
		Therefore by Assumptions \ref{a:multinomial}.1(a) and \ref{a:multinomial}.3(a) with $A = \{(x_0,x_1,\ldots,x_M) : x_m > x_{m'}$ if $m' \in \{0,\ldots,m-1\}$ and $x_m \geq x_{m'}$ if $m' \in \{m+1,\ldots,k\}\}$, we have
		\begin{align*}
		\lim_{n \to \infty} F_{n,P_{n,h}}(d_{c,mm}(X_n) = m)
		& = \mathbb P_h \Big( \Big( \mathbb E^*[ \dot{\underline p}_m[Z^* + Z]|Z] > \max_{0 \leq m' \leq m-1}  \mathbb E^*[ \dot{\underline p}_{m'}[Z^* + Z]|Z] \Big)    \\
		& \quad \quad \quad \quad \mbox{and } \Big( \mathbb E^*[ \dot{\underline p}_m[Z^* + Z]|Z] \geq \max_{m+1 \leq m' \leq k} \mathbb E^*[ \dot{\underline p}_{m'}[Z^* + Z]|Z] \Big) \Big) \,.
		\end{align*}

		\underline{Part (ii):} The excess maximum risk of $d \in \mathcal D$ is
		\[
		\Delta \mathcal R_{mm}(d,P) = \sum_{m=0}^M \mathbb I[ d = m ] \left( \max_{0 \leq m' \leq M} \underline p_{m'}(P) - \underline p_m(P) \right) \,.
		\]
		For Case 1, by continuity of $\underline p_m(\cdot)$ for all $m$ (under Assumption \ref{a:multinomial}.1(a)) we have $\underline p_m(P_{n,h}) \to \underline p_{m0}$ for all $m$ and $\max_{0 \leq m' \leq M} \underline p_{m'}(P_{n,h}) \to \underline p_{00}$. Then for all $n$ sufficiently large,
		\[
		\mathbb E_{P_{n,h}}[ \sqrt n \Delta \mathcal R_{mm}(d_n(X_n), P_{n,h}) ] = \sqrt n  \sum_{m=1}^M \mathbb E_{P_{n,h}} \left[ \mathbb I[ d_n(X_n) = m ] \right] \left( \underline p_0(P_{n,h}) - \underline p_m(P_{n,h}) \right) \,,
		\]
		where $\liminf_{n \to \infty}( \underline p_0(P_{n,h}) - \underline p_m(P_{n,h})) > 0$ for $m \geq 1$ and $ \mathbb E_{P_{n,h}} \left[ \mathbb I[ d_n(X_n) = m ] \right]  \to \mathbb E_h[ d^\infty_{m}(Z)]$. 
		Now consider Case 2. Again by continuity, for $n$ sufficiently large we have 
		\begin{align*}
		\mathbb E_{P_{n,h}}[ \sqrt n \Delta \mathcal R_{mm}(d_n(X_n), P_{n,h}) ] & = \sqrt n \sum_{m=1}^k \mathbb E_{P_{n,h}} \left[ \mathbb I[ d_n(X_n) = m ] \right] \left( \max_{0 \leq m' \leq k} \underline p_{m'}(P_{n,h}) - \underline p_m(P_{n,h}) \right) \\
		& \quad + \sqrt n \sum_{m=k+1}^M \mathbb E_{P_{n,h}} \left[ \mathbb I[ d_n(X_n) = m ] \right] \left( \max_{0 \leq m' \leq k} \underline p_{m'}(P_{n,h}) - \underline p_m(P_{n,h}) \right) \,,
		\end{align*}
		where the second sum is zero if $k = M$. If $k < M$, by similar arguments to Case 1 we have
		\[
		\sqrt n \sum_{m=k+1}^M \mathbb E_{P_{n,h}} \left[ \mathbb I[ d_n(X_n) = m ] \right] \left( \max_{0 \leq m' \leq k} \underline p_{m'}(P_{n,h}) - \underline p_m(P_{n,h}) \right) \to + \infty \times \sum_{m = k+1}^M \mathbb E_h[ d^\infty_m(Z)] \,.
		\]
		Moreover, for $m \leq k$ by Assumption \ref{a:multinomial}.1(a) we have
		\begin{align*}
		\sqrt n \left( \max_{0 \leq m' \leq k} \underline p_{m'}(P_{n,h}) - \underline p_m(P_{n,h}) \right) & = \left( \max_{0 \leq m' \leq k} (\sqrt n (\underline p_{m'}(P_{n,h}) - \underline p_{m' 0}) ) - \sqrt n (\underline p_m(P_{n,h}) - \underline p_{m0}) \right) \\
		& \to \max_{0 \leq m' \leq k} \dot{\underline p}_{m',P_0}[h] - \dot{\underline p}_{m,P_0}[h] \,.
		\end{align*} 
		
		\underline{Part (iii):} For Case 1, from part (ii), we see that $d^\infty_{0} = 1$ (almost everywhere) is optimal. For Case 2, from part (ii), we see that $d^\infty_m = 0$ (almost everywhere) is optimal for all $m > k$. For the remaining values of $m$, we have
		\begin{align*}
		& \int \sum_{m=0}^k \mathbb E_h[ d^\infty_{m}(Z)] \left( \max_{0 \leq m' \leq k} \dot{\underline p}_{m',P_0}[h] - \dot{\underline p}_{m,P_0}[h] \right) \, \mathrm d h \\ 
		& \propto \int \int d^\infty_{m,P_0}(z) \left( \max_{0 \leq m' \leq k} \dot{\underline p}_{m',P_0}[h] - \dot{\underline p}_{m,P_0}[h] \right) e^{-\frac{1}{2} (z - h)' I_0 (z - h) } \, \mathrm d z \, \mathrm d h \,.
		\end{align*}
		Changing the order of integration and minimizing pointwise in $z$, we see that if $M(z)$ denotes the set of maximizers of 
		\[
		\int \dot{\underline p}_{m,P_0}[h] e^{-\frac{1}{2} (z - h)' I_0 (z - h) }  \, \mathrm d h\,,
		\]
		then setting $d^\infty_{m}(z) = 0$ for $m \not \in M(z)$ and $d^\infty_{m}(z) \geq 0$ for $m \in M(z)$ with $\sum_{m \in M(z)} d^\infty_{m}(z) = 1$ is optimal. The tie-breaking rule used in part (i) is a special case with $d^\infty_{m}(z) = 1$ if $m = \min M(z)$.
	\end{proof}

	Characterizing $d_{c,mmr}(X_n)$ again requires a tie-breaking rule. In the derivation below we take the smallest element of the set of minimizers. To simplify notation, let $\tau_{m}(P) = \Delta p_{m}(P)$ and $\tau_{m0} = \tau_{m}(P_0)$ for $m = 0,1,\ldots,M$. Without loss of generality, reorder the indices so that $\tau_{00} \leq \tau_{10} \leq \ldots \leq  \tau_{M0}$. There are two cases, namely: case 1, $\tau_{00} < \tau_{10}$; and case 2, $\tau_{00} = \tau_{10} = \ldots = \tau_{k0}$ for some $k \in \{1, \ldots, M\}$ with $\tau_{k0} < \tau_{(k+1)0}$ if $k < M$.

	\begin{lemma}\label{lem:asymptotic.multinomial.regret}
		Let Assumptions \ref{A1}, \ref{a:multinomial}.1(b), \ref{a:multinomial}.2, and \ref{a:multinomial}.3(b) hold. Then: 
		\begin{enumerate}[nosep]
			\item[(i)] $d_{c,mmr}$ has the asymptotic representation
			\[
			d^\infty_{c,mmr,m}(Z) = 
			\left[ \begin{array}{ll}
			\mbox{$1$ if $m = 0$ and $0$ if $m \in \{1,\ldots,M\}$} & \mbox{in case 1,} \\
			\mathbb P_h ( ( \mathbb E^*[ \dot{\tau}_m[Z^* + Z]|Z] < \min_{0 \leq m' \leq m-1}  \mathbb E^*[ \dot{\tau}_{m'}[Z^* + Z]|Z] ) \mbox{ and} \\
			\quad  ( \mathbb E^*[ \dot{\tau}_m[Z^* + Z]|Z] \leq \min_{m+1 \leq m' \leq k} \mathbb E^*[ \dot{\tau}_{m'}[Z^* + Z]|Z] ) ) \\
			\mbox{if $m \in \{0,\ldots,k\}$ and $0$ if $m \in \{k+1,\ldots,M\}$}  & \mbox{in case 2,}
			\end{array} \right.
			\]
			where the minimum over an empty index is $+\infty$;
			\item[(ii)] local asymptotic excess maximum risk of $\{d_n\}_{n \geq 1} \in \mathbb D$ is
			\begin{align*}
			\mathcal L_{mmr}^\infty(d^\infty ; P_0, h) 
			& = 
			\left[ \begin{array}{ll}
			+\infty \times (1 - \mathbb E_h[ d^\infty_0(Z) ])  & \mbox{in case 1,} \\
			\sum_{m=0}^k \mathbb E_h[d^\infty_m(Z)] (\dot \tau_{m,P_0}[h] - \min_{0 \leq m' \leq k} \dot \tau_{m',P_0}[h]) \\
			\quad + \infty \times ( 1 - \sum_{m=0}^k \mathbb E_h[d^\infty_m(Z)]) & \mbox{in case 2,}
			\end{array} \right.
			\end{align*}
			where $\lim_{n \to \infty} \mathbb E_{P_{n,h}}[\mathbb I(d_n(X_n) = m)] = \mathbb E_h[d^\infty_{m}(Z)]$;
			\item[(iii)] $(d_{c,mmr,m}^\infty(Z))_{m=0}^M$ is optimal in the limit experiment.
		\end{enumerate} 
	\end{lemma}
	
	\begin{proof}[Proof of Lemma \ref{lem:asymptotic.multinomial.regret}]
		Follows by similar arguments to the proof of Lemma \ref{lem:asymptotic.multinomial.risk}.
	\end{proof}

	\subsection{Main results}
	
	Theorems \ref{t:optimal.binary} and \ref{t:optimal.multinomial} are proved using the following lemma, which is a very slight generalization of Lemma 1 of \cite{HiranoPorter2009}. We include a proof for completeness. It applies to both minimax risk and regret criteria, so we drop the subscripts $mm$ and $mmr$ on $\mathcal L$, $\Delta \mathcal B$, and $\mathcal R$.
	
	\begin{lemma}\label{lem:HP}
		Let $\mathcal L(\{d_n\}_{n \geq 1}; P_0, h) = \mathcal L^\infty(d^\infty;P_0,h)$ hold for every $P_0 \in \mathcal P$, $h \in \mathbb R^k$, and $\{d_n\}_{n \geq 1} \in \mathbb D$, where $d^\infty$ denotes the asymptotic representation of $\{d_n\}_{n \geq 1} \in \mathbb D$, and let the prior $\Pi$ have a strictly positive, continuously differentiable density $\pi$ on $\mathcal P$. Then: (i) for any $\{d_n\}_{n \geq 1} \in \mathbb D$, 
		\[
		\liminf_{n \to \infty} \Delta \mathcal B^n(d_n;P_0,\pi) \geq \pi(P_0) \inf_{d^\infty} \int \mathcal L^\infty(d^\infty;P_0,h) \, \mathrm d h
		\]
		(ii) If, in addition, $\{d_n^*\}_{n\geq 1} \in \mathbb D$ and its asymptotic representation $d^\infty_*$ solves 
		\[
		\int \mathcal L^\infty(d^\infty_*;P_0,h) \, \mathrm d h = \inf_{d^\infty} \int \mathcal L^\infty(d^\infty;P_0,h) \, \mathrm d h\,,
		\]
		and $d_n^*$ satisfies
		\begin{equation}\label{eq:reverse.fatou}
		\limsup_{n \to \infty} \Delta \mathcal B^n(d_n^*;P_0,\pi) 
		\leq \int \limsup_{n \to \infty}  \mathbb E_{P_{n,h}}\left[ \sqrt n \Delta\mathcal R\left(d_n^*, P_{n,h}; X_n\right) \right] \, \pi \left( P_{n,h} \right) \, \mathrm d h \,,
		\end{equation}
		then: 
		\[
		\lim_{n \to \infty} \Delta \mathcal B^n(d_n^*;P_0,\pi) =  \pi(P_0)  \inf_{d^\infty} \int \mathcal L^\infty(d^\infty;P_0,h) \, \mathrm d h \,,
		\]
		and hence
		\[
		\lim_{n \to \infty} \Delta \mathcal B^n(d_n^*;P_0,\pi) = \inf_{\{d_n\} \in \mathbb D} \liminf_{n \to \infty} \Delta \mathcal B^n(d_n;P_0,\pi) \,.
		\]
	\end{lemma}
	
	\begin{remark}\label{rmk:fatou}
		By the reverse Fatou lemma, condition (\ref{eq:reverse.fatou}) holds if there exists a non-negative function $g(h)$ with $ \mathbb E_{P_{n,h}}\left[ \sqrt n \Delta\mathcal R\left(d_n^*, P_{n,h}; X_n\right) \right] \, \pi \left( P_{n,h} \right) \leq g(h)$ for each $n$ and $\int g(h) \, \mathrm d h < \infty$.
	\end{remark}
	
	\begin{proof}[Proof of Lemma \ref{lem:HP}]
		\underline{Part (i):} follows by Fatou's lemma and definition of $\Delta \mathcal B^n(d_n;P_0,\pi)$:
			\begin{align*}
			\liminf_{n \to \infty} \Delta \mathcal B^n(d_n;P_0,\pi) & \geq \int \liminf_{n \to \infty}  \mathbb E_{P_{n,h}}\left[ \sqrt n \Delta\mathcal R\left(d_n, P_0 + n^{-1/2} h; X_n\right) \right] \, \pi \left( P_0 + n^{-1/2} h \right) \, \mathrm d h \\
			& = \pi(P_0) \int \mathcal L^\infty(d^\infty;P_0,h) \, \mathrm d h \,,
			\end{align*}
			where $d^\infty$ denotes the asymptotic representation of $\{d_n\}_{n \geq 1} \in \mathbb D$.
			
			\underline{Part (ii):} By condition (\ref{eq:reverse.fatou}) and optimality of $d^\infty_*$ in the limit experiment, we have
			\begin{align*}
			\limsup_{n \to \infty} \Delta \mathcal B^n(d_n^*;P_0,\pi) & \leq \int \limsup_{n \to \infty}  \mathbb E_{P_{n,h}}\left[ \sqrt n \Delta\mathcal R\left(d_n^*, P_{n,h}; X_n\right) \right] \, \pi \left( P_{n,h} \right) \, \mathrm d h \\
			& = \pi(P_0) \int \mathcal L^\infty(d^\infty_* ;P_0,h) \, \mathrm d h \\
			& = \pi(P_0) \inf_{d^\infty} \int \mathcal L^\infty(d^\infty;P_0,h) \, \mathrm d h \,.
			\end{align*}
			Combining with part (i) applied to $\{d_n^*\}_{n \geq 1}$, we obtain
			\[
			\lim_{n \to \infty} \Delta \mathcal B^n(d_n^*;P_0,\pi) = \pi(P_0) \inf_{d^\infty} \int \mathcal L^\infty(d^\infty;P_0,h) \, \mathrm d h \,.
			\]
			The final result is immediate from part (i).
	\end{proof}

	\begin{proof}[Proof of Theorem \ref{t:optimal.binary}]
		\underline{Part (i):} First note that as $\tilde d_n$ is binary, establishing convergence in distribution under $\{F_{n,P_{n,h}}\}_{n \geq 1}$ is equivalent to characterizing $\lim_{n \to \infty} F_{n,P_{n,h}} (\tilde d_n(X_n) = 1)$. Lemma \ref{lem:asymptotic.binary.risk}(i) establishes that $d_{b,mm}$ converges in distribution along every sequence $\{F_{n,P_{n,h}}\}_{n \geq 1}$. Asymptotic equivalence of $\tilde d_n$ and $d_{b,mm}$ implies $\lim_{n \to \infty} F_{n,P_{n,h}} (\tilde d_n(X_n) = 1) = \lim_{n \to \infty} F_{n,P_{n,h}} (d_{b,mm}(X_n) = 1) $ for all $h \in \mathbb R^k$ and all $P_0 \in \mathcal P$.  Therefore, $\tilde d_n$ has the same asymptotic representation as $d_{b,mm}$ from Lemma \ref{lem:asymptotic.binary.risk}(i). As this asymptotic representation is optimal in the limit experiment (cf. Lemma \ref{lem:asymptotic.binary.risk}(iii)) and $d_n$ satisfies condition (\ref{eq:reverse.fatou}) by assumption, the desired conclusion now follows by Lemma \ref{lem:HP}.
			
			\underline{Part (ii):} Follows similarly by Lemmas \ref{lem:asymptotic.binary.regret} and \ref{lem:HP}. 
	\end{proof}
	
\begin{proof}[Proof of Proposition \ref{prop:inefficient.binary}]
\underline{Part (i):}
By Fatou's lemma and definition of $\Delta \mathcal B^n_{b,mm}(d_n;P_0,\pi)$, we have
\begin{align*}
 \liminf_{n \to \infty} \Delta \mathcal B^n_{b,mm}(\tilde d_n;P_0,\pi) & \geq \int \liminf_{n \to \infty}  \mathbb E_{P_{n,h}}\left[ \sqrt n \Delta\mathcal R_{mm}\left(\tilde d_n, P_0 + n^{-1/2} h; X_n\right) \right] \, \pi \left( P_0 + n^{-1/2} h \right) \, \mathrm d h \\
 & = \pi(P_0) \int \mathcal L^\infty_{mm}(\tilde d^\infty;P_0,h) \, \mathrm d h \,,
\end{align*}
where $\tilde d^\infty$ denotes the asymptotic representation of $\{\tilde d_n\}_{n \geq 1} \in \mathbb D$ and $\pi(P_0) > 0$. By the proof of Theorem \ref{t:optimal.binary} we also have
\[
  \inf_{\{d_n\} \in \mathbb D}  \liminf_{n \to \infty} \Delta \mathcal B_{b,mm}^n(d_n;P_0,\pi) = \lim_{n \to \infty} \Delta \mathcal B^n_{b,mm}(d_{b,mm};P_0,\pi) = \pi(P_0) \int \mathcal L^\infty_{mm}( d^\infty_{b,mm};P_0,h) \, \mathrm d h \,.
\] 
Therefore, it suffices to show that 
\begin{equation} \label{eq:inequality.necessary.binary.risk}
 \int \mathcal L^\infty_{mm}(\tilde d^\infty;P_0,h) \, \mathrm d h > \int \mathcal L^\infty_{mm}( d^\infty_{b,mm};P_0,h) \, \mathrm d h \,.
\end{equation}

First, suppose that $a_{01} p_{L}(P_0) + a_{10} p_{U}(P_0) \neq a_{01}$. This corresponds to Cases 1 and 2 of Lemma \ref{lem:asymptotic.binary.risk}. As asymptotic equivalence fails, we have
\[
 \lim_{n \to \infty} F_{n,P_{n,h_*}}(\tilde d_n(X_n) = 1) \neq  \lim_{n \to \infty} F_{n,P_{n,h_*}}(d_{b,mm}(X_n) = 1)
\]
for some $P_0 \in \mathcal P$ and $h_*$ in $\mathbb R^k$. We may restate the above display in terms of the asymptotic representations:
\[
 \mathbb E_{h_*}[\tilde d^\infty(Z)]  \neq  \mathbb E_{h_*}[ d^\infty_{b,mm}(Z)] \,.
\]
By H\"older's inequality we may deduce that the functions $h \mapsto \mathbb E_{h}[\tilde d^\infty(Z)]$ and $h \mapsto \mathbb E_{h}[ d^\infty_{b,mm}(Z)]$ are both continuous at $h_*$. Therefore, there exists a set $H \subset \mathbb R^k$ with positive Lebesgue measure upon which $\mathbb E_{h}[\tilde d^\infty(Z)]  \neq  \mathbb E_{h}[ d^\infty_{b,mm}(Z)]$ for all $h \in H$.

If $P_0$ is as in Case 1 of Lemma \ref{lem:asymptotic.binary.risk}, then $\mathbb E_h[\tilde d^\infty(Z)] < 1$ for all $h \in H$. This, in turn, implies that $\mathcal L_{mm}^\infty(\tilde d^\infty ; P_0, h) = +\infty$ for all $h \in H$. By contrast, $\mathcal L_{mm}^\infty(d^\infty_{b,mm} ; P_0, h) = 0$ for all $h \in \mathbb R^k$. The proof when $P_0$ satisfies the conditions of Case 2 of Lemma \ref{lem:asymptotic.binary.risk} follows similarly.

Now suppose that $a_{01} p_{L}(P_0) + a_{10} p_{U}(P_0) \neq a_{01}$, which corresponds to Case 3 of Lemma \ref{lem:asymptotic.binary.risk}. Let $f(P) = a_{01}p_L(P) + a_{10}p_U(P) $. By Lemma \ref{lem:asymptotic.binary.risk}(ii), to prove inequality (\ref{eq:inequality.necessary.binary.risk}) it suffices to show
\[
 \int \left(\tilde d^\infty(z) \int (\dot f_{P_0}[h]) e^{-\frac{1}{2}(z-h)'I_0(z-h)}\, \mathrm d h \right)  \mathrm d z < \int \left( d^\infty_{b,mm}(z)  \int (\dot f_{P_0}[h]) e^{-\frac{1}{2}(z-h)'I_0(z-h)}\, \mathrm d h \right) \mathrm d z \,.
\]
The function $\tilde d^\infty_{b,mm}(z) = \mathbb I \left[  \int  (\dot f_{P_0}[h]) e^{-\frac{1}{2} (z - h)' I_0 (z - h) } \, \mathrm d h  \geq 0 \right]$ maximizes
\[
  d(z) \times \int (\dot f_{P_0}[h]) e^{-\frac{1}{2}(z-h)'I_0(z-h)}\, \mathrm d h
\]
over all $[0,1]$-valued functions of $z$, so the preceding inequality holds weakly. To establish a strict inequality, note the functions $\tilde d^\infty(z)$ and $d^\infty_{b,mm}(z)$ must disagree on a set of positive Lebesgue measure, say $\mathcal Z$. For each $z \in \mathcal Z$ we must have one of the following:
\begin{itemize}[nosep]
\item[(i)] $\int (\dot f_{P_0}[h]) e^{-\frac{1}{2}(z-h)'I_0(z-h)}\, \mathrm d h > 0$ and $d^\infty(z) < 1$;
\item[(ii)] $\int (\dot f_{P_0}[h]) e^{-\frac{1}{2}(z-h)'I_0(z-h)}\, \mathrm d h < 0$ and $d^\infty(z) > 0$;
\item[(iii)] $\int (\dot f_{P_0}[h]) e^{-\frac{1}{2}(z-h)'I_0(z-h)}\, \mathrm d h  = 0$.
\end{itemize}
However, the condition $\mathbb E^*[a_{01} \dot p_{L,P_0}[Z^* + Z] + a_{10} \dot p_{U,P_0}[Z^* + Z] |Z ] \neq 0$ a.e.
implies that case (iii) only holds on a set of zero Lebesgue measure. Therefore, for almost every $z \in \mathcal Z$ either case (i) or (ii) must hold, which establishes the desired inequality.

\underline{Part (ii):} This follows by Lemma \ref{lem:asymptotic.binary.regret} using similar arguments to Part (i).
\end{proof}

\begin{proof}[Proof of Theorem \ref{t:optimal.multinomial}]
		The proof follows by similar arguments to Theorem \ref{t:optimal.binary}, using Lemmas \ref{lem:asymptotic.multinomial.risk} and \ref{lem:HP} for part (i) and  Lemmas \ref{lem:asymptotic.multinomial.regret} and \ref{lem:HP} for part (ii).
\end{proof}

	\subsection{Results on Computation}
	
	\begin{proof}[Proof of Proposition \ref{prop:duality.binary}]
		Dropping dependence of the $L$-vector $b$ and $K \times L$ matrix $G$ on the low-dimensional parameter $\phi$, the primal problem is
		\[
		\sup_{\pi \in \mathbb R^L} \; b'\pi \quad \mbox{subject to} \quad (G - r \otimes 1_{1 \times L}') \pi = 0, \; 1_{1 \times L}\pi - 1 = 0, \; \pi \ge 0\,,
		\]
		where the final inequality holds element-wise. The Lagrangian is
		\[
		\sup_{\pi \in \mathbb R^L} \inf_{\mu \in \mathbb R^K, \zeta \in \mathbb R,\kappa \in \mathbb R_+^L} \; \mathcal{L}(\pi,\mu,\zeta,\kappa).
		\]
		Here $\mu$, $\zeta$, and $\kappa$ are the Lagrange multipliers on the three constraints and
		\begin{eqnarray*}
			\mathcal{L}(\pi,\mu,\zeta,\kappa)
			&=& b'\pi + \mu' \left( G - r \otimes 1_{1 \times L}' \right) \pi + \zeta \left( 1_{1 \times L}\pi -  1 \right) + \kappa' \pi \\
			&=& \left(b + \left(G - r \otimes 1_{1 \times L}'\right)' \mu + \zeta 1_{1 \times L}' + \kappa \right)' \pi - \zeta \,.
		\end{eqnarray*}
		By duality, we have
		\[
		\sup_{\pi} \inf_{\mu, \zeta,\kappa}  \; \mathcal{L}(\pi,\mu,\zeta,\kappa)
		= \inf_{\mu, \zeta,\kappa}  \sup_{\pi}  \; \mathcal{L}(\pi,\mu,\zeta,\kappa).
		\]
		
		For fixed $\mu$, $\zeta$, and $\kappa$, consider the problem 
		\[
		\sup_{\pi}  \; \mathcal{L}(\pi,\mu,\zeta,\kappa) = \sup_{\pi}  \; \underbrace{\left(b + \left(G - r \otimes 1_{1 \times L}'\right)' \mu + \zeta 1_{1 \times L}' + \kappa \right)'}_{=:b^*(\mu,\zeta,\kappa)'} \pi - \zeta \,.
		\]
		This value can be made $+\infty$ by assigning arbitrarily large positive values to any element of $\pi$ for which $b^*(\mu,\zeta,\kappa)$ has a positive entry, and an arbitrarily large negative value to any element of $\pi$ for which $b^*(\mu,\zeta,\kappa)$ has a negative entry. The minimizing agent would therefore choose 
		\[
		\kappa^* = \kappa^*(\zeta,\mu) = - \left( \left( b_l + (G_l - r)' \mu + \zeta \right) \wedge 0 \right)_{l \in \{1,\ldots,L\}} \in \mathbb R_+^L
		\]
		so that all entries of $m^*(\mu,\zeta,\kappa^*)$ are non-negative:
		\[
		m^*(\mu,\zeta,\kappa^*) = \left( \left( b_l + (G_l - r)' \mu + \zeta \right) \vee 0 \right)_{l \in \{1,\ldots,L\}} \,,
		\]
		and then choose
		\[
		\zeta^* = \zeta(\mu) = - \max_{l \in \{1,\ldots,L\}} \left( b_l + (G_l - r)' \mu \right)
		\]
		so that every entry of $b^*(\mu,\zeta^*,\kappa^*)$ is zero. Any $\zeta \leq \zeta^*$ will suffice for this purpose, but values of $\zeta$ strictly less than $\zeta^*$ will result in a higher value of the minimizing agent's objective.
		Combining the intermediate results, we obtain
		\begin{eqnarray*}
			\sup_{\pi} \inf_{\mu, \zeta,\kappa} \; \mathcal{L}(\pi,\mu,\zeta,\kappa) 
			&=& \inf_{\mu} \; \bigg( \max_{l \in \{1,\ldots,L\}} \; b_l + \mu'(G_l-r)  \bigg) .
		\end{eqnarray*}
		
		This min-max problem may be restated as a linear program by introducing an additional variable $t \in \mathbb R$ for the minimizing agent:
		\begin{align*}
		\inf_{\mu} \; \bigg( \max_{l \in \{1,\ldots,L\}} \; b_l + \mu'(G_l-r)  \bigg) 
		& = \inf_{\mu,t} \; t \quad \mbox{s.t.} \quad t \geq \left(  b_l + \mu'(G_l-r)  \right) \,,\; l = 1,\ldots,L \\
		& = \inf_{\mu,t} \; t \quad \mbox{s.t.} \quad t 1_{L \times 1} \geq \left(  b + (G' - (1_{L \times 1} \otimes r')) \mu \right) \\
		& = \inf_v \;\; [0_{1\times K}, 1] \,v \quad \mbox{s.t.} \quad A v \leq -b \,,
		\end{align*}
		where $v = [\mu',t]' \in \mathbb R^{K+1}$ and $ A = [G'- (1_{L \times 1} \otimes r'), -1_{L\times 1}]$.
	\end{proof}
	
	\begin{proof}[Proof of Proposition \ref{prop:duality.binary.continuous}]
		The dual representation follows from \cite{CsiszarMatus2012}. Large-$\delta$ behavior is established in \cite{ChristensenConnault2019}.
	\end{proof}

\end{appendix}

\newpage


\end{document}